\documentclass{article}
\usepackage{iclr2025_conference,times}

\usepackage{amsmath,amsfonts,bm}

\def\eqref#1{equation~\ref{#1}}

\def\1{\bm{1}}

\DeclareMathAlphabet{\mathsfit}{\encodingdefault}{\sfdefault}{m}{sl}
\SetMathAlphabet{\mathsfit}{bold}{\encodingdefault}{\sfdefault}{bx}{n}

\usepackage[colorlinks=true, citecolor=blue, linkcolor=blue, urlcolor=blue]{hyperref}
\usepackage{url}

\usepackage[utf8]{inputenc} %
\usepackage[T1]{fontenc}    %
\usepackage{booktabs}       %
\usepackage{amsfonts}       %
\usepackage{nicefrac}       %
\usepackage{microtype}      %
\usepackage{xcolor}         %

\usepackage{tabularx} %
\usepackage{lipsum}   %

\usepackage{longtable}

\usepackage[ruled,norelsize,linesnumbered,noend]{algorithm2e} %
\usepackage{amsmath}
\usepackage{amsthm} %
\usepackage{thmtools}
\usepackage{cleveref} 
\usepackage{graphicx}
\usepackage{svg}
\usepackage{subcaption}
\usepackage{wrapfig}

\newtheorem{theorem}{Theorem}[section] %
\newtheorem{lemma}[theorem]{Lemma}     %
\newtheorem{corollary}[theorem]{Corollary} %
\newtheorem{definition}{Definition}[section] %
\newtheorem{remark}{Remark}

\usepackage{thmtools}

\newcount\Comments  %
\newcommand{\kibitz}[2]{\ifnum\Comments=1{\color{#1}{#2}}\fi}
\newcommand{\mjc}[1]{\kibitz{blue}{[MJC: #1]}}
\newcommand{\es}[1]{\kibitz{purple}{[ES: #1]}}

\title{Truthful Aggregation of LLMs\\
with an Application to Online Advertising}

\author{%
Ermis Soumalias \\
University of Zurich \\
ETH AI Center \\
\texttt{ermis@ifi.uzh.ch}
\And
Michael J. Curry \\
Harvard University \\
University of Zurich \\
ETH AI Center \\
\texttt{mcurry@seas.harvard.edu}
\And
Sven Seuken \\
University of Zurich \\
ETH AI Center \\
\texttt{seuken@ifi.uzh.ch}
}

\iclrfinalcopy
\begin{document}

\maketitle
\begin{abstract}
The next frontier of online advertising is revenue generation from LLM-generated content. We consider a setting where advertisers aim to influence the responses of an LLM to align with their interests, while platforms seek to maximize advertiser value and ensure user satisfaction. The challenge is that advertisers' preferences generally conflict with those of the user, and advertisers may misreport their preferences. To address this, we introduce MOSAIC, an auction mechanism that ensures that truthful reporting is a dominant strategy for advertisers and that aligns the utility of each advertiser with their contribution to social welfare. Importantly, the mechanism operates without LLM fine-tuning or access to model weights and provably converges to the output of the optimally fine-tuned LLM as computational resources increase. Additionally, it can incorporate contextual information about advertisers,  which significantly improves social welfare. Through experiments with a publicly available LLM, we show that MOSAIC leads to high advertiser value and platform revenue with low computational overhead. While our motivating application is online advertising, our mechanism can be applied in any setting with monetary transfers, making it a general-purpose solution for truthfully aggregating the preferences of self-interested agents over LLM-generated replies.
\end{abstract} %

\mjc{did a small amount of shortening. correlation is now in appendix, please take a look at that. 2 is done. with comments removed, we are currently at 10 pages}

\section{Introduction}
\textit{Large language models (LLMs)} are becoming ubiquitous -- as coding assistants, as chat interfaces, and even as alternatives to search engines \citep{bommasani2022opportunities}.
To ensure their usefulness, it is essential to closely align LLM outputs with user preferences.
However, in general, there may be multiple interested parties who disagree on the desired behavior of \textit{the same} LLM.
This leads to the question of \textit{how to generate LLM replies given multiple conflicting preferences}.

The motivating application for our work is online advertising, which is the main source of revenue for large tech companies such as Google, Meta, and Twitter. For example, in $2023$, Meta's advertising revenue of $132$ billion USD was more than $97$\% of its total revenue \citep{meta2024}. Ad auctions are the workhorse mechanism to determine the placement and prices of commercial content~\citep{varian2007,edelman2007}.
As platforms are now working on integrating advertisements into LLM-generated content~\citep{Criddle24Perplexity}, new auction mechanisms are needed for this application.

We present a new auction mechanism for this problem, which we call \textit{MOSAIC (\textbf{M}echanism for \textbf{O}ptimally \textbf{S}ampling and \textbf{A}ggregating LLM Outputs with \textbf{I}ncentive \textbf{C}ompatibility)}. The advertisers are paying, not for some particular item or bundle, but rather to influence the output generated by an LLM in a direction closer to their own preferences.
Although our mechanism could be of interest whenever one has to aggregate the preferences of multiple self-interested agents over LLM behavior (as long as it is reasonable to charge monetary payments), we see online advertising as the most natural setting of interest.
For this reason, we refer to participants in the mechanism as \textit{advertisers}.

\subsection{Problem setting}
Consider a situation where a user queries an LLM with a specific question.
We assume the following:
\begin{itemize}
    \item There is a \textit{reference LLM} that can produce useful replies to the user's query. 
    \item There are \textit{advertisers}, each with their own LLM, who want to influence the reply to the user. 
\end{itemize}
In our model, each advertiser is represented via her own LLM
or directly with a \textit{reward function}, similar to the function used to fine-tune LLMs in the \textit{Reinforcement Learning from Human Feedback} pipeline \citep{azar2023}. 
For this reason, we refer to an advertiser's value for a reply as her \textit{reward}.
The auctioneer's task is, given the user's query and the advertisers' preferences, to \textit{produce a reply that is useful for the user while also generating high rewards for the  advertisers}.

\paragraph{Example.}
A user queries an LLM \emph{``What are effective ways to learn a musical instrument online?''} 
There are two advertisers:
\emph{MusicMastery}, which offers online music lessons, and
\emph{InstaTune}, which sells musical instruments. MOSAIC merges the preferred replies from the advertisers' LLMs and the reference LLM to produce one combined reply to the user:
\begin{quote}
    \textit{Effective ways to learn a musical instrument online include:\\
1. Online music lessons: Platforms like \textbf{MusicMastery} offer one-on-one virtual tutoring, group classes, and pre-recorded lessons.\\
2. Interactive learning tools: \textbf{InstaTune} offers a range of interactive learning tools, such as apps, software, and games, that make learning fun and engaging. }
\end{quote}

\subsection{Overview of Contributions}

In this paper, we present MOSAIC, an auction mechanism designed to aggregate the preferences of multiple self-interested advertisers over LLM-generated replies. The first challenge we address is that \mbox{MOSAIC} must produce replies for which the advertisers receive high rewards, but without steering the LLM's behavior too far from that of the user-centric reference LLM. We address this by drawing a strong connection to the RLHF objective (e.g., \cite{ziegler2020finetuning}), equipping MOSAIC with a hyperparameter that enables the auctioneer to balance between the usefulness of the produced reply to the user
and the advertisers in an interpretable and principled way (\Cref{subsection:the_formal_model}).

Second, leveraging an importance sampling-based technique, MOSAIC is the only mechanism in the literature that converges to the optimal distribution as computational resources increase (\Cref{Corollary:mechanism_follows_optimal_distribution}).
Third, MOSAIC can incorporate contextual information, similar to how sponsored search auctions utilize advertiser descriptions. 
This accelerates convergence and increases value for the advertisers and revenue for the auctioneer (\Cref{subsection:efficiency_results}).

Fourth, we address technical feasibility and practicality. We adopt the desiderata proposed by
\citet{duetting2023mechanism}, who argued that auctions must be compatible with existing LLM technology, only using ``information obtainable from current models'' in such a way that they are ``easy to integrate into the system'' and rely only on ``easy manipulations of LLM output.'' Furthermore, it must be computationally feasible to run the auction repeatedly and with different advertiser preferences. In particular, we cannot afford the expensive process of fine-tuning an LLM's weights for each possible query. To address this, we have designed MOSAIC to work by only post-processing multiple LLM outputs -- it requires only ``API access'', not even viewing the model weights 
(\Cref{section:practical_considerations}).

Fifth, we address the fact that advertisers might have an interest in misreporting their preferences (which include their \textit{LLMs}), analogous to over-/underbidding in traditional auctions. To address this challenge, we employ an allocation rule that is compatible with Rochet payments \citep{rochet1987}. These payments uniquely ensure that MOSAIC is strategyproof (i.e., that truthful reporting is a dominant strategy 
for each advertiser, no matter what the others do (\Cref{theorem:truthful_payments})). 
Crucially, Rochet payments ensure that MOSAIC is strategyproof even if the allocation rule has not converged.

Sixth, we introduce an offset to the Rochet payments, to ensure that
each advertiser's utility is proportional to her contribution to the social welfare. 
This alignment is important for the long-term success of the market, incentivizing only the most relevant advertisers to participate. 

Finally, we discuss \textit{individual rationality (IR),} which guarantees that a participant is weakly better off from participating in the mechanism. 
In \Cref{subsec:differences}, we discuss the unique properties of our setting that make \textit{ex-post} IR unattainable. However, in \Cref{subsec:almost_ir}, we show that MOSAIC is ``almost ex-post IR'' in a well-defined sense. In \Cref{subsection:efficiency_results}, we experimentally show that, for the distribution of instances tested, it is \textit{ex-ante} IR (i.e., advertisers are in expectation better off by participating).

In \Cref{section:experiments}, we provide experimental results for the online advertising domain. We demonstrate that MOSAIC quickly converges to the optimal LLM with low computational cost,
generating significant value for the advertisers
and revenue for the auctioneer while also being useful to the user.

\section{Related Work}
In their pioneering work, \citet{duetting2023mechanism} 
introduced the field of mechanism design for LLMs, 
proposing a novel sequential mechanism where the output sequence is generated token by token, and advertisers bid for their LLM to produce the next token.  
This work laid the foundation for future work in this area. 
However, 
several challenges hinder the adoption of their approach in the real-world:
(i) Advertisers face the \textit{exposure problem}: small changes in the generated sequence, such as the addition of a word like "not," can completely alter its meaning. An advertiser might pay a significant amount for the tokens generated up to a certain point, only for the continuation to unexpectedly negate or distort her intended message.\footnote{
As an example, suppose advertiser A has successfully bid for all the tokens in the sequence "Planning your next vacation? For the cheapest flights to Hawaii, book via..." However, she fails to win the bid for the crucial next token, which would have been her brand name. This leaves her having paid for the preceding tokens without getting the benefit of associating her brand with the message.} 
(ii) The mechanism is easily manipulable if the assumption that advertisers cannot misreport their LLMs is dropped.
(iii) For a given prompt, an advertiser's spend grows with the length of the generated sequence. 
MOSAIC addresses all of the above challenges.

\citet{dubey2024auctions} proposed a generalization of the position auction \citep{varian2007, edelman2007} to a setting where each advertiser is interested in having some specific text ad displayed alongside the organic results. In their mechanism, an LLM module coupled with an auction module work in tandem to merge the ads into a single ad summary in an incentive-compatible way.

\citet{Hajiaghayi24Ad} considered advertisers bidding in an auction to have their ads placed into various portions of a retrieval-augmented generation (RAG) response. Their auction only allows advertisers to report (and misreport) scalar values to influence the inclusion of non-misreportable advertising texts in the combined output. This is in contrast to MOSAIC, which allows the aggregation of full LLMs that advertisers may arbitrarily misreport.

\citet{Feizi24Online} presented an abstract design for an LLM advertising system and detailed a number of research challenges that would have to be overcome in the course of implementation.
\citet{Conitzer24Social} drew connections between \textit{computational social choice} and LLM alignment.
\citet{Fish23Generative} presented work in the opposite direction: how can LLMs be used to solve problems in social choice?
\citet{Harris24Algorithmic} studied Bayesian persuasion in an abstract setting where a ``simulator'' (for example, a realistic LLM) of the agent is available. \citet{Werner24Experimental} gave preliminary empirical evidence that LLM-generated content can influence consumers.

In work that appeared after the initial version of this paper, \citet{Bergemann24DataDriven} consider a setting where bidders have both private types reflecting their preferences and private signals reflecting information about a world state (for example, user interests). When applied to LLM auctions, they identify their setting as a generalization of ours, where the signals known to the advertisers may actually improve the output for the user, and where deviations from the reference LLM may be measured by arbitrary distance functions.
They analyze conditions under which a welfare maximizing mechanism can be implemented truthfully (eliciting both types and signals) in posterior equilibrium, but do not consider practical implementations of their mechanism.

\section{Framing Sequence Generation as a Mechanism Design Problem} \label{section:game_theory_problem}

\subsection{Formal Model} \label{subsection:the_formal_model}

A \textit{user} issues a query $x$.
There is a \textit{reference LLM} $\pi_{\text{ref}}$ that the auctioneer aims not to deviate from too much (e.g., because it is responsible for providing useful replies to the user).
Additionally, there is a set $N$ of $n$ advertisers who have their own preferences for the reply (i.e., a token sequence) that will be returned to the user. We use the terms \textit{sequence} and \textit{reply} interchangeably. 

An LLM can be abstracted as a mapping from (partial) token sequences to a distribution over the next token, or equivalently as an implicit distribution over token sequences. 
We use this second abstraction, i.e., $\pi_i(y | x)$ denotes the probability that advertiser $i$'s LLM $\pi_i$ assigns to output sequence (i.e., reply) $y$ for the user query $x$.

We let $r_i(x, y)$ denote advertiser $i$'s \textit{reward} for sequence $y$, given query  $x$.\footnote{
In theory, the advertiser's reward for a generated sequence could also depend on user-specific information, but we abstract that away. 
Equivalently, we assume that the reward function $r_i$ provided by advertiser $i$ is specific to the given user that asked the query $x$.}
Informally, the auctioneer's goal is to sample the final sequence from a distribution that optimizes the  advertisers' expected rewards without substantially diverging from $\pi_{\text{\text{ref}}}$. 
This is motivated by traditional online ad auctions, where the implicit goal is to maximize advertiser value subject to ensuring a chosen level of usefulness to the user (e.g., the balance between allocating search slots to ads and native content).
In the context of an LLM environment, the analogous notion of ``usefulness to the user'' is represented by the closeness of the final distribution to $\pi_{\text{ref}}$.
Formally, the goal is to maximize:
\begin{gather} \label{eq:objective}
    J(\pi) =  \mathbb{E}_{y\sim\pi} \left [\sum_{i \in N} r_i(x, y) \right ] - 
    \tau D_{\text{KL}}(\pi(\cdot | x) || \pi_{\text{\text{ref}}}(\cdot | x)),   
\end{gather}
where $\tau > 0 $ is a hyperparameter enabling the auctioneer to balance producing replies 
closer to the reference LLM or with higher reward for the  advertisers, and $D_{\text{KL}}$ is the Kullback-Leibler divergence. 

This objective is analogous to the standard Reinforcement Learning from Human Feedback (RLHF) approach \citep{ziegler2020finetuning}, but replaces the human feedback reward function $r_{\text{HF}}(x,y)$ with the 
advertisers' aggregate reward function $r(x,y) =  \sum_{i \in N} r_i(x, y)$.
For an overview of RLHF, we recommend \citet[\S 3]{rafailov2023dpo}.

The optimal solution {$\pi_r^*$} to the optimization problem in (\ref{eq:objective}) was derived by \citet{peters2007}: 
\begin{gather} \label{eq:opt_distribution}
    \pi_r^* (y | x) = \frac{1}{Z(x)} \pi_{\text{\text{ref}}}(y|x) \exp \left ( 
    \frac{1}{\tau} \sum_{i \in N} r_i(x, y)
    \right ), 
\end{gather}
where $Z(x) = \sum_{y \in T^*} \pi_{\text{ref}}(y|x) \exp \left ( \frac{1}{\tau} \sum_{i \in N} r_i(x, y)\right ) $ is the partition function.

The performance goal we seek to maximize is the quantity $\tau \log \pi_{\text{ref}} (y| x) + \sum_{i \in N} r_i (x,y) $; as it aggregates  the interests of both the user and the advertisers, we refer to it as the \emph{social welfare} of $y$.

Let $\widehat{R}$ be the set of all possible reports by the  advertisers. 
A \textit{mechanism} is defined as a pair $(\pi, p)$. 
The \textit{allocation rule} 
$\pi: \widehat{R} \rightarrow  ( T^* \rightarrow  \Delta(T^*))$ 
maps any report profile 
$\widehat{r} = (\widehat{r}_1, \widehat{r}_2, \dots, \widehat{r}_n) \in \widehat{R}$ of the  advertisers' rewards 
to an LLM $\pi_{\widehat{r}}$, which in turn is a mapping from a user query $x$ to a distribution over token sequences $\delta(T^*)$. 
We denote the LLM that the allocation rule $\pi$ induces for reports $\widehat{r}$ as  $\pi_{\widehat{r}}$, 
and the optimal LLM for those reports  (i.e., the maximizer of \Cref{eq:objective}) as  $\pi_{\widehat{r}}^*$.
The \textit{payment rule} $p: \widehat{R} \rightarrow \mathbb{R}^n$ maps any report profile $\widehat{r}$ of the advertisers' rewards to a payment profile $p(\widehat{r})$, where $p_i(\widehat{r})$ is the payment of the $i$-th advertiser to the mechanism.

A mechanism is \emph{strategyproof} if and only if, for any advertiser $i \in N$, 
truthfully reporting her reward function is always the optimal strategy, regardless of the reports from other advertisers. In our setting, this property holds in expectation because the reply is \emph{sampled} from the LLM induced by the allocation rule.

\begin{definition}[Strategyproof]
{A mechanism $(\pi, p)$ is \emph{strategyproof} if, for any advertiser $i \in N$, true reward function $r_i$, reported reward function $\widehat{r}_i$, reported reward functions $\widehat{r}_{-i}$ by the other advertisers and prompt $x$:
  $  \mathbb{E}_{y \sim \pi_{(r_i, \widehat{r}_{-i} )} (\cdot | x) } 
  [ r_i(x, y) - p_i(r_i, \widehat{r}_{-i} ) ] \ge
  \mathbb{E}_{y \sim \pi_{(\widehat{r}_i, \widehat{r}_{-i} )} (\cdot | x) } 
  [ r_i(x, y) - p_i(\widehat{r}_i, \widehat{r}_{-i} ) ]
    ,$
where $r_i(x, y) - p_i (\widehat{r}_i, \widehat{r}_{-i} )$ is the utility of advertiser $i$ for reply $y$ when her payment is $p_i(\widehat{r}_i, \widehat{r}_{-i})$.}
\end{definition}

\subsection{The shortcomings of VCG in this setting} \label{section:vcg_problems}
At first sight, it may seem that the Vickrey–Clarke–Groves (VCG) mechanism would be suitable for our setting~\citep{vickrey1961,clarke1971,groves1973}. VCG selects the outcome that maximizes the sum of all agents' values.
This can be either the single optimal sequence in \Cref{eq:objective} or the optimal distribution of \Cref{eq:opt_distribution}. 
The VCG mechanism has a corresponding payment rule to incentivize truthful reporting: it charges each agent her externality, that is,
the total reduction in value (respectively expected value) for the other agents that her participation in the mechanism caused.

However, in our setting, VCG is \textit{not} a viable option:
The VCG allocation rule requires calculating the \textit{exact} optimal solution to the optimization problem, which is intractable for choosing an LLM to maximize \Cref{eq:objective} and is even difficult for choosing a single optimal sequence.
If a sub-optimal solution is chosen, VCG's strategyproofness is no longer guaranteed~\citep{Nisan07Computationally,Nisan1999Algorithmic,Lehmann2002Truth}.
Thus, VCG is not a suitable mechanism in our setting.

\section{The MOSAIC Mechanism: Allocation Rule} \label{section:allocation_rule}
\subsection{Convergence to Optimality, Advertiser Contexts and Importance Sampling}

In this section, we introduce MOSAIC's allocation rule.
The high-level idea is that first, 
a set of $M$ \textit{candidate replies} are generated based on some LLM $\pi_{\text{gen}}$ responsible for that task. 
Then, the probability of returning each candidate reply is re-weighted based on the advertisers' reports and the reference LLM $\pi_{\text{ref}}$ so that in the limit, as $M \to \infty$, the probability of returning each reply converges to its probability under the optimal distribution of \Cref{eq:opt_distribution}. This approach shares similarities with the rejection sampling approach that has been used at various points in the LLM training pipeline (e.g.,~\citet{bai2022constitutional}).  
We defer all proofs to \Cref{appendix:allocationdetails}.

\begin{algorithm}[t!] 
\caption{Allocation Rule for MOSAIC} \label{alg:static_mechanism_improved}
\DontPrintSemicolon  %
\KwIn{User prompt $x$, reference LLM $\pi_{\text{\text{ref}}}$, LLM used for candidate reply generation $\pi_{\text{\text{gen}}}$, 
advertiser reward functions $\{ r_i \}_{i = 1}^n$, number of candidate replies to generate $M$, reference LLM weight $\tau$}
\KwOut{Reply $y$ drawn according to the optimal distribution as defined in \Cref{eq:objective} for the aggregate reward function $r(x, y) = \sum_{i = 1}^N r_i(x,y)$}

\SetKwInput{KwInput}{User Prompt $x$}                %
\SetKwInput{KwOutput}{Output}              %
\SetKwInput{KwInitialization}{Initialization} %
\SetKw{KwTo}{in}                           %
\SetKwFor{For}{for}{do}{endfor}            %
\SetKwIF{If}{ElseIf}{Else}{if}{then}{else if}{else}{endif} %

\BlankLine
Sample $y_j \sim \pi_{\text{gen}}(\cdot | x ), \; 1 \le j \le M$   \;
Calculate $r(x, y_j) = \sum_{i = 1}^N r_i(x, y_j), \; 1 \le j \le M$   \;
\Return $y \sim \text{softmax} \left 
(\frac{r(x, y_1)}{\tau} + \log \frac{\pi_{\text{ref}}(y_1 | x)}{\pi_{\text{gen}}(y_1 | x  )},  
\dots, 
\frac{r(x, y_M)}{\tau} + \log \frac{\pi_{\text{ref}}(y_M | x)}{\pi_{\text{gen}}(y_M | x  )}
\right )   $  \;
\end{algorithm}

\begin{corollary} \label{Corollary:mechanism_follows_optimal_distribution}
For any reported reward functions $r \in R$ by the advertisers and any LLM $\pi_{\text{gen}}$ such that $\pi_{\text{ref}}$ is absolutely continuous with respect to $\pi_{\text{gen}}$,
the MOSAIC policy 
${\pi_{r, M}(\cdot | x)}$ induced by \Cref{alg:static_mechanism_improved} using $M$ candidate replies converges to an optimal solution for the platform's objective (\Cref{eq:objective}) as $M \to \infty$. Formally: 
\begin{gather}
    \lim_{M \to \infty} \pi_{r, M}(\cdot | x) =  \pi^*_r( \cdot | x) \in \arg \max_{\pi \in \Delta(T^*)} \mathbb{E}_{y \sim \pi (\cdot | x)}[r(x, y )] - \tau D_{\text{KL}} (\pi || \pi_{\text{ref}})
\end{gather}
\end{corollary}

Based on \Cref{Corollary:mechanism_follows_optimal_distribution}, MOSAIC converges to the optimal distribution for \textit{any} LLM $\pi_{\text{gen}}$, provided that $\pi_{\text{ref}}$ is absolutely continuous with respect to $\pi_{\text{gen}}$. While $\pi_{\text{gen}} = \pi_{\text{ref}}$ is an intuitive choice, it would result in impractically slow convergence rates.
The reason is that $\pi_{\text{ref}}$ may only generate replies with high rewards for the advertisers exceedingly rarely.
Advertisers have high rewards for responses that explicitly mention their brands, 
but $\pi_{\text{ref}}$ considers replies with mentions of specific brands extremely unlikely. 
Consequently, generating candidate replies directly from $\pi_{\text{ref}}$ leads \Cref{alg:static_mechanism_improved} to sample from a set of low-reward candidates, hindering performance.

To address this challenge, we generate candidate sequences not from \(\pi_{\text{ref}}(\cdot | x)\), but instead from a \textit{context-aware LLM}, \(\pi_{\text{con}}(\cdot | x; c)\). The instance-specific context \(c\) is designed to bridge the gap between the parts of the output space favored by $\pi_{\text{ref}}$ and those valued by the advertisers. 

We formalize this intuition in \Cref{app:importancesampling}, where we make a connection between MOSAIC's allocation rule and an importance-based sampling estimator for $\pi^*_{r}$, proving the following lemma:
\begin{restatable}{lemma}{estimatorvariance} \label{lemma:estimator_variance_main}
For any LLM $\pi_{\text{gen}}$ such that $\pi_{\text{ref}}$ is absolutely continuous with respect to $\pi_{\text{gen}}$,
the variance of the policy $ \pi_{r, M}(\cdot | x)$ as an estimator for 
$\mathbb{E}_{y \sim \pi^*_r(\cdot |  x )} [\pi^*_{r}(y|  x )]$ is
\begin{align}
    \text{Var}(\widehat{\mu}_{\text{IS}}) & = 
    \frac{1}{M}
    \left ( \sum_{y \in Y} \frac{\pi^*_{r} (y | x)^4}{\pi_{\text{gen}}(y | x)}
    - \left ( \sum_{y \in Y} \pi^*_{r} (y | x)^2
    \right )^2 
    \right ). \label{eq:variance_final_statement_main}
\end{align}
\end{restatable}

Applying Chebyshev's inequality to the variance of \Cref{lemma:estimator_variance_main}, MOSAIC's output converges to that of $\pi_r^*$ at a rate of $\sqrt{M}$ (\Cref{lemma:estimator_convergence_rate}).\footnote{In \Cref{sec:app:variation_distance} we establish a more general result,  
showing that for any sequence $y$, its estimated probability converges to the corresponding probability under $\pi_r^*$ at a rate of $1 / \sqrt{M}$. }
However, this convergence rate depends inversely on $ \frac{\pi^*_{r} (y | x)^4}{\pi_{\text{gen}}(y | x)} $. Thus, generating the replies via an LLM $ \pi_{\text{con}}(\cdot | x; c) $ that more closely matches $\pi^*_{r}(\cdot | x)$   reduces the estimator's variance and improves convergence speed.

In our application of integrating advertisers' interests into LLM outputs, $c_i$ is a context-specific description of the $i$-th advertiser. 
These descriptions, supplied by the advertisers themselves, should be easily verifiable and factually accurate, akin to ``MusicMastery: offering online music lessons'', or ``InstaTune: selling musical instruments.''\footnote{A practical way of implementing $\pi_{\text{con}}(\cdot | x ; c)$ given $\pi_{\text{ref}}(\cdot | x)$ is to 
augment the input $x$ to the reference LLM with the advertiser descriptions.
In our example: ``Try to mention \textlangle advertiser x\textrangle, \textlangle advertiser x description\textrangle.'
\label{footnote:context_aware_llm}
} 
This approach is analogous to search engine optimization in sponsored search advertising, where advertisers supply and potentially optimize their own descriptions to influence how they are presented by the auction mechanism.

Our experiments in \Cref{subsection:efficiency_results} demonstrate that, when using reasonable computational resources, using the context-aware LLM to generate candidate replies achieves substantially higher rewards and utility for the advertisers, increased revenue for the auctioneer, and faster convergence.
In the rest of the paper,
we refer to using the reference and context-aware LLMs  as the baseline and context-aware versions of our mechanisms, respectively.

\subsection{Practical Considerations and Computational Efficiency} \label{section:practical_considerations}

\paragraph{Input Methods and Black Box Access to Advertiser LLMs.}
MOSAIC's allocation and payment rules do not depend on the  advertisers' full reward functions, but only on their rewards for the candidate replies. 
Thus, MOSAIC requires only ``API access'' to the involved LLMs without fine-tuning or access to their weights. 
\citet{rafailov2023dpo} established a mapping between an agent's LLM and her implicit reward function, allowing MOSAIC to use as inputs reply probabilities (i.e., LLM inference calls) instead of rewards. For more details, see \Cref{appendix:full_differences}.

\paragraph{Computational Costs.}
Generating a single candidate reply of length \( L \) tokens requires \( L \) forward passes through \( \pi_{\text{gen}} \), as each token is generated sequentially in an autoregressive manner. 
Evaluating each reply involves processing the entire sequence in a single forward pass, as the sequence is already known. 
Thus, with \( M \) candidate replies and \( n \) advertisers (plus the reference LLM for evaluation), the total number of forward passes required by MOSAIC is 
\( M \cdot (L + n + 1) \). 
This underscores the importance of approximating the optimal distribution with relatively few candidate replies.
By comparison, in the only other mechanism for aggregating multiple advertiser LLMs \citep{duetting2023mechanism}, 
all \( n \) advertiser LLMs participate in generating each token, leading to a total of \( L \cdot n \) forward passes. 
\Cref{section:experiments} demonstrates that MOSAIC converges to the optimal distribution using \( M = 20 \) candidate replies in realistic scenarios.
Additionally, \Cref{sec:app:compute_experiments} provides a detailed experimental evaluation showing that MOSAIC requires the compute time of only five LLM queries to converge.

\paragraph{Reward Elicitation and Computational Efficiency.} MOSAIC can be integrated with computationally efficient methods for eliciting the advertisers' reward functions and sampling from $\pi_{\text{con}}$, reducing computational overhead. For instance, \citet{li2024qprobe} demonstrated that a simple linear function on a model's embedding space can approximate an LLM's reward function. 
If advertisers' reward functions are represented this way, MOSAIC would require only $M \cdot L$ forward passes.

\paragraph{{Compatibility with More Efficient Generation Methods.}}
MOSAIC relies solely on LLM forward passes—the core operation that modern architectures are highly optimized to perform. This reliance on standard operations allows MOSAIC to take full advantage of existing optimizations designed for efficient generation. 
As a result, MOSAIC converges to the optimal distribution 
with compute costs equivalent to just five times that of a single LLM query (\Cref{sec:app:compute_experiments}).

\paragraph{Parallelization.}
The generation and evaluation of each candidate sequence are independent processes. This independence allows MOSAIC to be fully parallelized, in which case the response time for a user query is comparable to that of directly querying a \emph{single} LLM.

\paragraph{Inherent Competition.}
 In conventional auctions, competition among bidders is crucial for driving up prices; when the market is thin, prices tend to be low. Auctioneers often use reserve prices to mitigate this problem. However, reserve prices are challenging to calibrate and can lead to lost revenue if not optimized properly. In our design, advertisers inherently compete with the reference LLM, ensuring consistent competition even in thin markets. This built-in competition mitigates the need for reserve prices and ensures a baseline revenue even with only one advertiser.

While MOSAIC introduces a \emph{single} tunable parameter, $\tau$, its behavior differs significantly from reserve prices in conventional auctions. 
A slightly mis-calibrated reserve price can result in a total loss of revenue by excluding all bidders. In contrast, MOSAIC's revenue is a smooth and differentiable function of $\tau$ (see the proof of \Cref{theorem:truthful_payments}), making it more robust to the choice of $\tau$.

{\paragraph{Support for Output Constraints.}
The  candidate replies are generated by $\pi_{\text{gen}}$, which is controlled by the platform. 
Thus, MOSAIC allows the platform to incorporate any desired constraints  on the output replies such as maximum length by adding those constraints to the generation procedure.}

\section{The MOSAIC Mechanism: Payment Rule} \label{section:payments}
In this section, we first show how the allocation rule from \Cref{section:allocation_rule} can be combined with an appropriate payment rule so that the resulting mechanism is strategyproof (\Cref{subsection:payments_through_convexity}). 
Then, we detail how auctions for LLM-generated content differ from standard auctions (\Cref{subsec:differences}). 
Taking those differences into account, we create a payment offset,
so the mechanism is both strategyproof and social welfare aligned (\Cref{section:bidding zero offset}).
We defer all proofs to \Cref{appendix:section:payment_details}.

\subsection{Strategyproof Payments through Cyclic Monotonicity} \label{subsection:payments_through_convexity}

The allocation rule we introduced in \Cref{section:allocation_rule} satisfies \textit{cyclic monotonicity}~\citep[\S~24]{Rockafellar1970}.
 \citet{rochet1987} first proposed the use of cyclic monotonicity in mechanism design as a generalization of the monotonicity property in single-parameter settings \citep{Myerson81Optimal}.
For general settings, cyclic monotonicity of the allocation rule is a sufficient and necessary condition for the existence of a payment rule so that the resulting mechanism is strategyproof. 
By applying these results we can prove:

\begin{theorem} \label{theorem:truthful_payments}

The allocation rule of \Cref{alg:static_mechanism_improved} can be combined with a payment rule  such that for any advertiser $i \in N$ and 
set of candidate replies $\{y_j\}_{j=1}^M$, reporting truthfully is a dominant strategy. 
Advertiser $i$'s expected utility (up to a constant of integration $C$) under truthful reporting is:

\begin{align}
    {\widetilde{U}_{i}}(r_i, \widehat{r}_{-i} ; \pi_{ref}, \pi_{gen} ) 
    & = 
    \tau \log \left ( 
        \sum_{j = 1}^M \exp  \left (  \frac{\sum_{k \in N \setminus \{ i \}} \widehat{r}_k(x, y_j) +  r_i(x,y_j) }{\tau} 
        \right ) 
        + \log \frac{\pi_{\text{ref}}(y_j | x)}{\pi_{\text{gen}}(y_j | x )}
        \right ) 
         + C.
    \label{equation:utility_myerson_spelled_out}
\end{align}
\end{theorem}

Note that, based on \Cref{theorem:truthful_payments}, in our mechanism it is \textit{always} a dominant strategy for an advertiser to report truthfully. 
Crucially, this is not the case for VCG, where truthful reporting would be optimal only if the allocation rule had converged to the optimal distribution as defined in \Cref{eq:opt_distribution}.

\subsection{Differences from Auction Mechanisms for Item Allocation} \label{subsec:differences}
Auction mechanisms designed to sell items or ad slots typically rely on simplifying assumptions that do not apply in a setting with LLM-generated content.
Key differences include:
(i) \textit{Non-Negative Values:} These mechanisms assume agents' values are non-negative due to having zero value for the empty bundle and free disposal. In contrast, in our setting, an advertiser's reward can be negative based on the discrepancy between her LLM and the reference LLM.
(ii) \textit{Advertiser-Specific Allocations:} These mechanisms allocate different item bundles to different agents. Here, a single reply is produced, and advertisers' rewards depend on that reply's probability  with respect to their LLMs.
(iii) \textit{Zero Utility for Non-Participation:} In most auction settings, not participating yields zero utility. 
Here, non-participation can result in negative utility since the produced reply may be unfavorable to non-participating advertisers.
For details, see Appendix \ref{appendix:full_differences}.

\subsection{advertiser-Specific Utility Offset} \label{section:bidding zero offset}
We now modify MOSAIC's payment rule by adding a \textit{payment offset} (and thus a utility offset):  
    $C = - \widetilde{U}_{i}({0}, \widehat{r}_{-i} ; \pi_{ref}, \pi_{con} ) $.
Informally, we additionally
charge each advertiser her utility in \Cref{equation:utility_myerson_spelled_out}
if her reward for all candidate replies was zero.
This offset maintains the key properties of our mechanism (i.e., strategyproofness and convergence to the optimal distribution), while also achieving two additional properties that are critical for the long-term success of a market for LLM aggregation. The first is \textit{``Almost IR:''} An advertiser with weakly positive reward for all candidate replies has weakly positive expected utility for all reports by the other advertisers. In particular, an advertiser with zero reward for all outputs is guaranteed zero utility. 
The second is \textit{``What you give is what you get:''} an advertiser's ex-interim expected utility is monotone in how well-aligned her exponentiated reward for the replies is with the interim allocation rule if she were to not participate.

\subsubsection{Our Mechanism is ``Almost Individually Rational''} \label{subsec:almost_ir}
Individual rationality is important to incentivize agents to participate in the mechanism. 
In \Cref{app:almost_ir}, we explain why the standard notion of individual rationality (i.e., weakly positive utility from participation in the mechanism) encountered in standard auction settings is impossible to achieve in our setting while converging to the optimal distribution and maintaining 
strategyproofness.
Then, we explain how, with our payment offset, our mechanism is ``almost IR:''
In \Cref{lemma:zero_bid_offset_ex_interim_utility_zero} we prove that the ex-interim utility of an advertiser who has zero reward for all candidate replies and bids truthfully is deterministically zero, i.e.,  advertisers that do not contribute to the social welfare (but also do not detract from it) have zero utility. 
Similarly, in \Cref{lemma:zero_bid_offset_ex_interim_utility_positive} we prove that 
if an advertiser's reward for all candidate sequences is (weakly) positive, then her ex-interim utility is (weakly) positive.

\begin{remark}
In \Cref{section:experiments}, we experimentally show that our offset payment rule, coupled with our context-aware allocation rule, results in both high expected rewards and positive expected utility for the advertisers (i.e., ex-ante individual rationality), as well as significant revenue for the auctioneer.
\end{remark}

\subsubsection{``What you give is what you get''} \label{section:what_you_give_is_what_you_get}

Our allocation rule, which is the only one over a finite set of replies that converges to the optimal LLM, is also the (sub)gradient of the utility to ensure truthfulness \citep{rochet1987}. 
Because the allocation rule is the same for all advertisers, their utilities must  also be the same, up to advertiser-specific offsets, as indicated by \Cref{equation:utility_myerson_spelled_out}. However, not all advertisers contribute equally to social welfare.
In \Cref{app:market_unraveling}, we explain why implementing the mechanism without carefully designed offsets would lead to free-riding. Specifically, if an advertiser's utility in \Cref{equation:utility_myerson_spelled_out} is positive, even when her relevance to the user query is low, it incentivizes her participation. This dilutes the utility for all participants, discouraging more relevant advertisers from participating. As a result, the expected rewards for advertisers decline, and the generated replies become less useful for users.
To avoid free-riding, it is crucial to align each advertiser’s utility with her contribution to social welfare.

\begin{lemma} \label{lemma:ex_interim_utility_alignment}
For the offset 
{ $C = - \widetilde{U}_{i}(0, \widehat{r}_{-i} ; \pi_{ref}, \pi_{con} ) $
 advertiser $i$'s expected utility under truthful reporting (which is her dominant strategy), with the expectation taken over the returned reply, is:} 

{
\begin{gather}
   U_i(r_i, \widehat{r}_{-i} ; \pi_{ref}, \pi_{con})  =  \tau \log \left ( 
\sum_{j = 1}^M \exp \left ( \frac{r_i(x, y_j)}{\tau} \right ) 
\pi_{\widehat{r}_{-i} } (y_j | x ) \right ). 
\label{eq:expected_utility_static_mech_zero_bid_offset_2}
\end{gather}}
\end{lemma}

In words, \Cref{lemma:ex_interim_utility_alignment} proves that every advertiser's expected utility is monotone in how well aligned her exponentiated reward function is with MOSAIC's induced LLM $\pi_{\widehat{r}_{-i} }$, had she not participated in the mechanism.

\begin{remark} \label{remark:what_you_give_is_what_you_get}
In \Cref{subsection:fairness_results}, we experimentally show that the offset described in \Cref{section:bidding zero offset} creates a strong positive correlation between an advertiser’s contribution to social welfare and her expected utility gain from participation in the mechanism. Moreover, our experiments show that for the distribution of instances tested, MOSAIC satisfies ex-ante IR. \end{remark}

\section{Experiments} \label{section:experiments}
In this section, we experimentally evaluate the performance of our mechanism. 
We focus on the online advertising domain, the main motivation application for our mechanism.

\subsection{Experiment Setup} \label{sec:experiment_setup}
We create synthetic instances, each consisting of a user query (e.g., ``How to learn a musical instrument online?'') and {two} advertisers  (e.g., ``MusicMastery, offering online music lessons'').

This is consistent with the experimental setup in \citet{duetting2023mechanism} while also highlighting MOSAIC'S performance and revenue even in scenarios with low competition. 
In \Cref{sec:app:experiments_many_advertisers}, we extend our results to settings with more advertisers. 
Crucially, MOSAIC demonstrates similarly rapid convergence to the optimal distribution in these settings, requiring the same number of candidate replies as in the two-advertiser case, while also generating significant value for the advertisers and revenue for the platform.

All synthetic instances are provided in \Cref{sec:all_prompts}.
We use Llama-2-7b-chat-hf as the reference LLM \citep{touvron2023llama}.\footnote{
In \Cref{subsec:testing_with_other_models}, we replicate these experiments using Google’s flan-t5-large model \citep{flan}, observing qualitatively very similar results.}
Following \citet{duetting2023mechanism},  we create the advertisers' LLMs by adding advertising instructions to the reference LLM.
The context-aware LLM is created as described in \Cref{footnote:context_aware_llm}.

Following \citet{rafailov2023dpo}, the advertisers' reward functions are defined as $r_i(x, y) = \log \frac{\pi_i (y | x)}{\pi_{\text{ref}} (y | x)}$. For the auctioneer's objective, we set $\tau = 1$ in \Cref{eq:objective},  balancing advertisers' rewards and divergence from the reference LLM.

We use 50 user queries
and test each query on 25 different random seeds, resulting in 1,250 instances. 
Following \citet{li2024qprobe,rozière2024codellama}, we sample from all LLMs using a temperature of 0.8 and top-p 0.95.
For full experimental details, see Appendix \ref{appendix:full_experiment_setup}.

\subsection{Experimental Results} \label{section:results}
\subsubsection{Evaluating the Allocation Rule} \label{subsection:efficiency_results}

\begin{table}[t!] 
\vskip -0.25cm 
\centering
\begin{tabularx}{\textwidth}{X}
\toprule
\textbf{User query} \\
\midrule
What are effective ways to learn a musical instrument online? \\
\midrule 
\textbf{Advertisers} \\
\midrule
- MusicMastery: offering online music lessons and virtual tutoring

- InstaTune: selling musical instruments and learning aids
\\ 
\midrule
\textbf{Context-Aware Mechanism Output} \\
\midrule
Effective ways to learn a musical instrument online include:

1. Online music lessons: Platforms like \textbf{MusicMastery} offer one-on-one virtual tutoring, group classes, and pre-recorded lessons. These lessons can be tailored to your skill level and learning style.

2. Interactive learning tools: \textbf{InstaTune} offers a range of interactive learning tools, such as apps, software, and games, that make learning fun and engaging.\dots \\ 

\midrule
\textbf{Baseline Mechanism Output} \\
\midrule
- Practice consistently and regularly: Consistent practice is essential to learning any musical instrument, and it's even more important when learning online. Set aside a specific time each day or week to practice, and make sure to stick to it.

- Find a good teacher or mentor: While it's possible to learn some things on your own, working with a good teacher or mentor can help you improve faster and avoid developing bad habits. \dots \\ 
\bottomrule
\end{tabularx}
\caption{Comparison of context-aware and baseline mechanism outputs using 10 candidate sequences. Advertisers highlighted in bold, outputs truncated to conserve space.}
\label{tab:mechanism_output_comparison2}
\vskip -0.45cm 
\end{table}

In this section, we evaluate the effectiveness of our  allocation rule in generating value for advertisers and converging to the optimal distribution. We assess both the baseline and context-aware versions of MOSAIC, i.e., using the reference and context-aware LLMs to generate candidate replies.

To provide some intuition on how the context-aware mechanism enhances the relevance of responses for advertisers, we compare outputs from both mechanisms in \Cref{tab:mechanism_output_comparison2}. Notably, only the context-aware mechanism successfully incorporates advertisers into the replies.

In \Cref{subfig:sequence_log_probability}, we plot the log probability of the replies returned by MOSAIC with respect to the optimal distribution (\Cref{eq:opt_distribution}) 
{against the number $M$ of candidate replies generated.}
We compare the context-aware version of MOSAIC to the baseline version to evaluate the effectiveness of incorporating contextual information.\footnote{
There are no other mechanisms in the literature for aggregating multiple LLMs to compare against. 
The only other mechanism that aggregates multiple LLMs is \citet{duetting2023mechanism}, which is not strategy-proof and does not converge to the platform's objective. The only other potential comparison is VCG, which, as explained in \Cref{section:vcg_problems}, is not applicable since it does not align with the platform’s objective. Hence, the best baseline for comparison is naive MOSAIC, which we consider a strong benchmark given that it provably converges to the platform’s optimal distribution based on \Cref{Corollary:mechanism_follows_optimal_distribution}.
}
 
To benchmark our results, we estimate the log probability of sampling a \textit{single reply} from the optimal distribution.\footnote{Note that the closed-form solution of \Cref{eq:opt_distribution} allows us to evaluate the probability of sentences with respect to the optimal solution, but it does not enable us to sample from that distribution. Sampling from that distribution would require using reinforcement learning to train the optimal LLM on the  advertisers' aggregate reward function, which is computationally infeasible for the number of problem instances we test. Instead, we generate replies from the reference LLM and evaluate them based on the induced probabilities of the reference LLM, for which the reference LLM is the optimal one. This serves as a proxy for the log probabilities we should expect if we were to draw replies from the optimally fine-tuned model for each query.}

We observe that for both the context-aware and baseline versions of MOSAIC, the log probability of the returned reply scales rapidly with the number of candidate replies generated. 
This is in line with our theoretical analysis in \Cref{section:allocation_rule}, where we proved that both versions converge to the optimal distribution. 
We also observe that incorporating context into the mechanism significantly accelerates convergence.
Notably, the context-aware mechanism can achieve higher log probability with respect to the optimal distribution with four candidate replies than the baseline version can with $20$. 
Additionally, with only $20$ generated candidate replies, the context-aware mechanism can almost match our estimate of the log probability of sampling from the optimal distribution. 
In \Cref{sec:usefulness_to_the_user}, we further demonstrate that for the context-aware mechanism, the usefulness to the user also scales rapidly with the number of candidate replies.

\begin{figure}
    \centering
    \begin{subfigure}[b]{0.45\textwidth} %
        \centering
        \includegraphics[scale = 0.2]{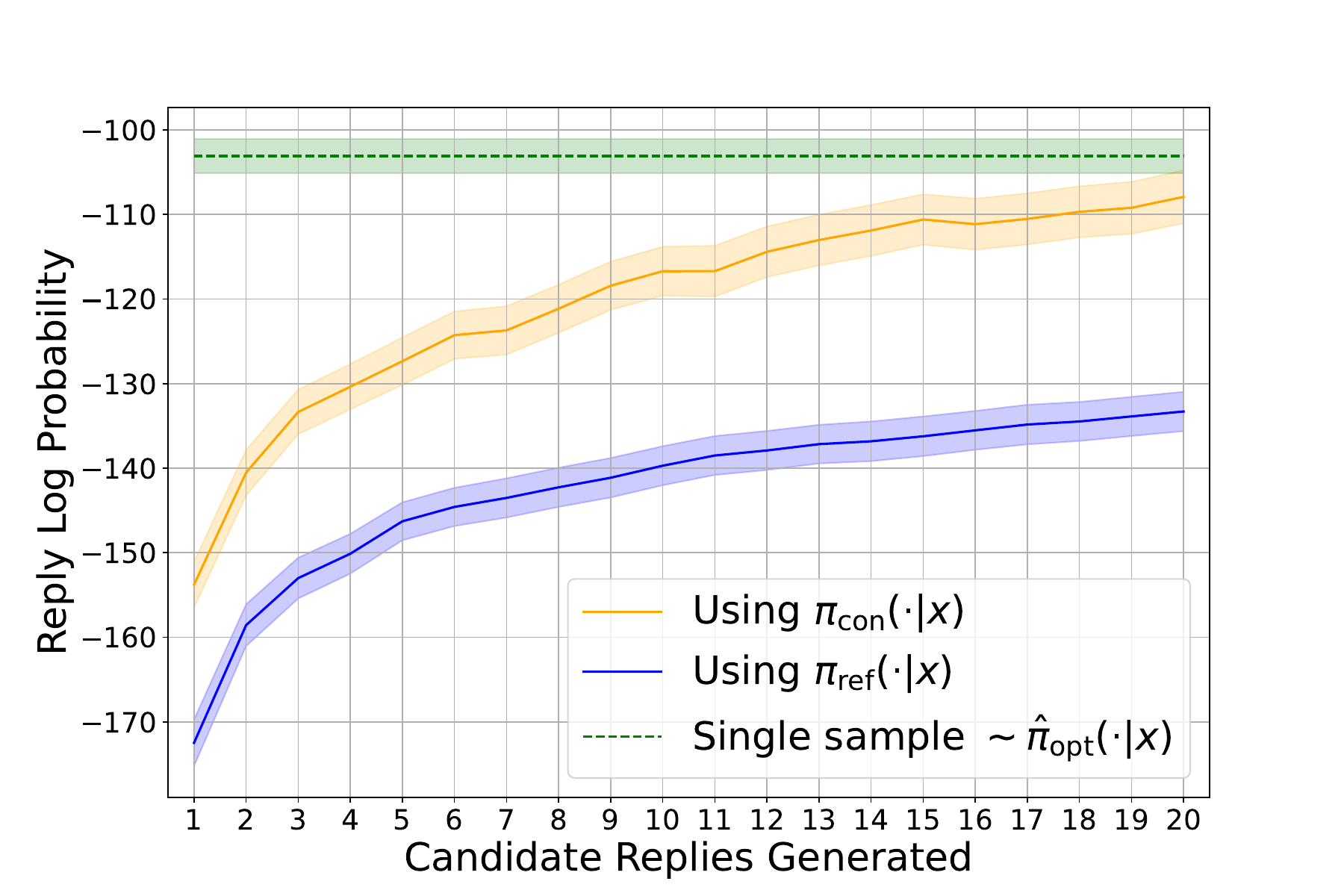}
        \caption{Returned sequence log probability as a function of the number of candidate replies generated using 
        $\pi_{\text{ref}}$ and $\pi_{\text{con}}$.
        We also show a proxy of drawing a single sequence from the optimal distribution.    
        }
        \label{subfig:sequence_log_probability}
    \end{subfigure}
    \hfill %
    \begin{subfigure}[b]{0.45\textwidth} %
        \centering
        \includegraphics[scale=0.2]{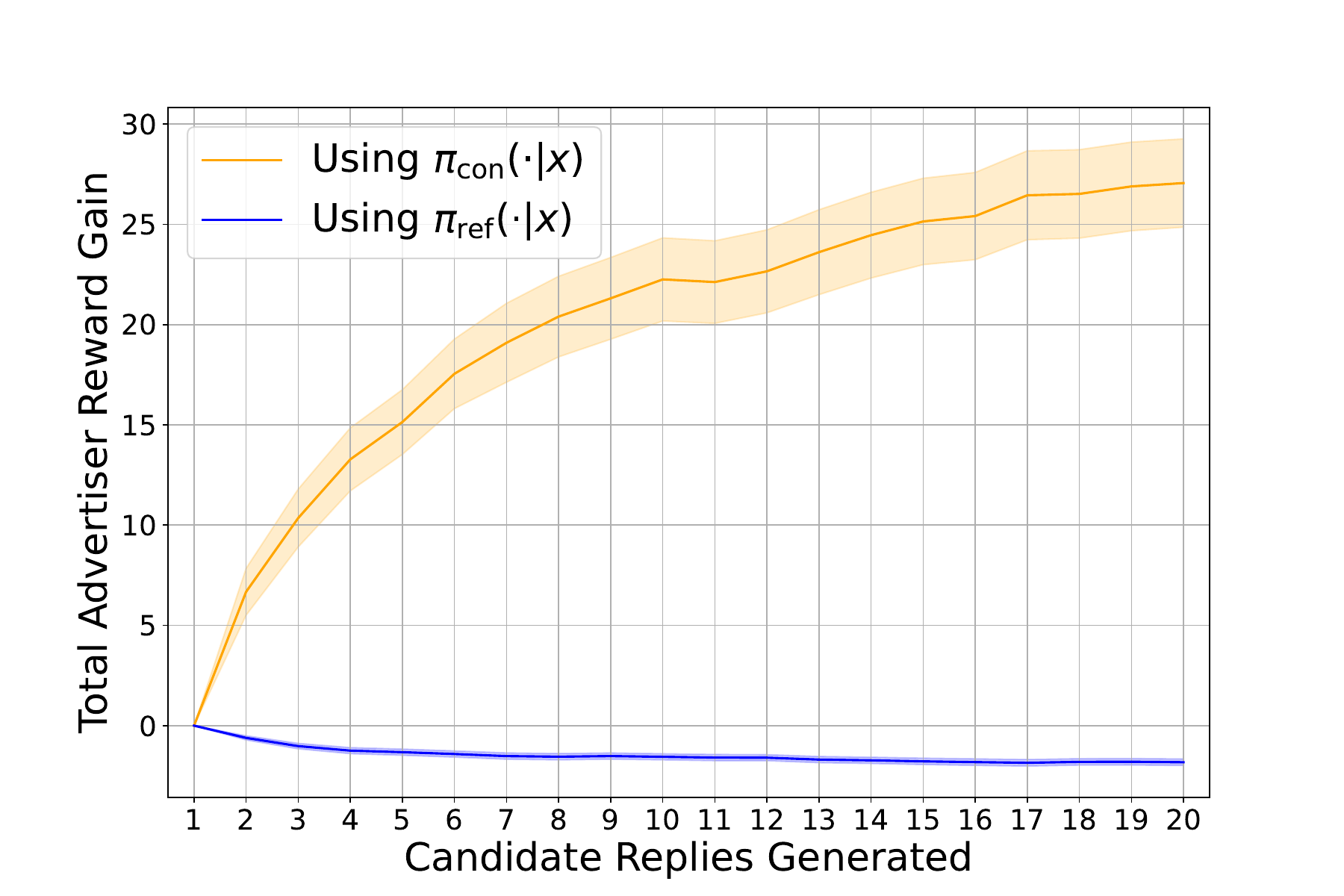}
        \caption{Total advertiser normalized reward as a function of the number of candidate replies generated using 
        $\pi_{\text{ref}}$ and $\pi_{\text{con}}$.\\}
        \label{subfig:advertiser_reward_gain_participating}
    \end{subfigure}
    \caption{
    Reply log probability and total advertiser normalized reward as a function of the number of candidate sequences generated
    using \(\pi_{\text{ref}}\) and \(\pi_{\text{con}}\).
    Averaged over 1250 runs including 95\% CIs. 
    }
    \vspace{-20pt}
    \label{fig:main_results}
\end{figure}

In \Cref{subfig:advertiser_reward_gain_participating}, we plot the total expected reward for advertisers as a function of the number of candidate replies generated by the mechanism. The expectation is taken over the draw of the returned reply from the set of candidate replies. To make the plot more interpretable, we normalize an advertiser's expected reward by her expected reward had she not participated in the mechanism. Specifically, we plot
$
r_{i, \text{norm}}(x) = \sum_{j \in M} r_i(x, y_j) \cdot \pi_{\text{int}}(y_j | x) - \sum_{j \in M} r_i(x, y_{-i, j}) \cdot \pi_{\text{int}}(y_{-i, j} | x),
$
where \( y_{-i} = (y_{-i,1}, \dots, y_{-i,M}) \) represents the candidate replies that would have been generated had the advertiser not participated.\footnote{To reduce computational costs, we estimate an advertiser's reward for not participating based on her expected reward over the already generated replies in which her brand is not mentioned by name, motivated by the fact that if she does not participate, her brand will not be mentioned by name.}
We compare the outcomes for the two versions of our mechanism, highlighting the added value of integrating context.

\Cref{subfig:advertiser_reward_gain_participating} demonstrates that our context-aware mechanism significantly boosts advertiser rewards, with benefits scaling rapidly as the number of candidate replies increases. In contrast, the baseline version fails to improve advertiser rewards within a computationally feasible number of generated replies, and may even cause a slight decrease. This occurs because generating more replies adds marginal positive reward to the reference LLM (see \Cref{sec:usefulness_to_the_user}), but does so at the expense of advertiser rewards. In summary, MOSAIC's support for context-aware LLMs allows it to quickly converge to the theoretically optimal distribution and generate substantial rewards for advertisers.

In \Cref{sec:app:compute_experiments} we provide a detailed experimental evaluation of MOSAIC's compute requirements, showing how it requires the compute cost of just five LLM queries to converge.\footnote{We also discuss the fact that MOSAIC can be fully parallelized so that the user-perceived latency is the same as for a \emph{single} LLM query.}

\subsubsection{Evaluating the Payment Rule} \label{subsection:fairness_results}

\begin{wrapfigure}[13]{r}{0.45\textwidth}
    \centering
    \vspace{-35pt}
    \includegraphics[width=\linewidth]{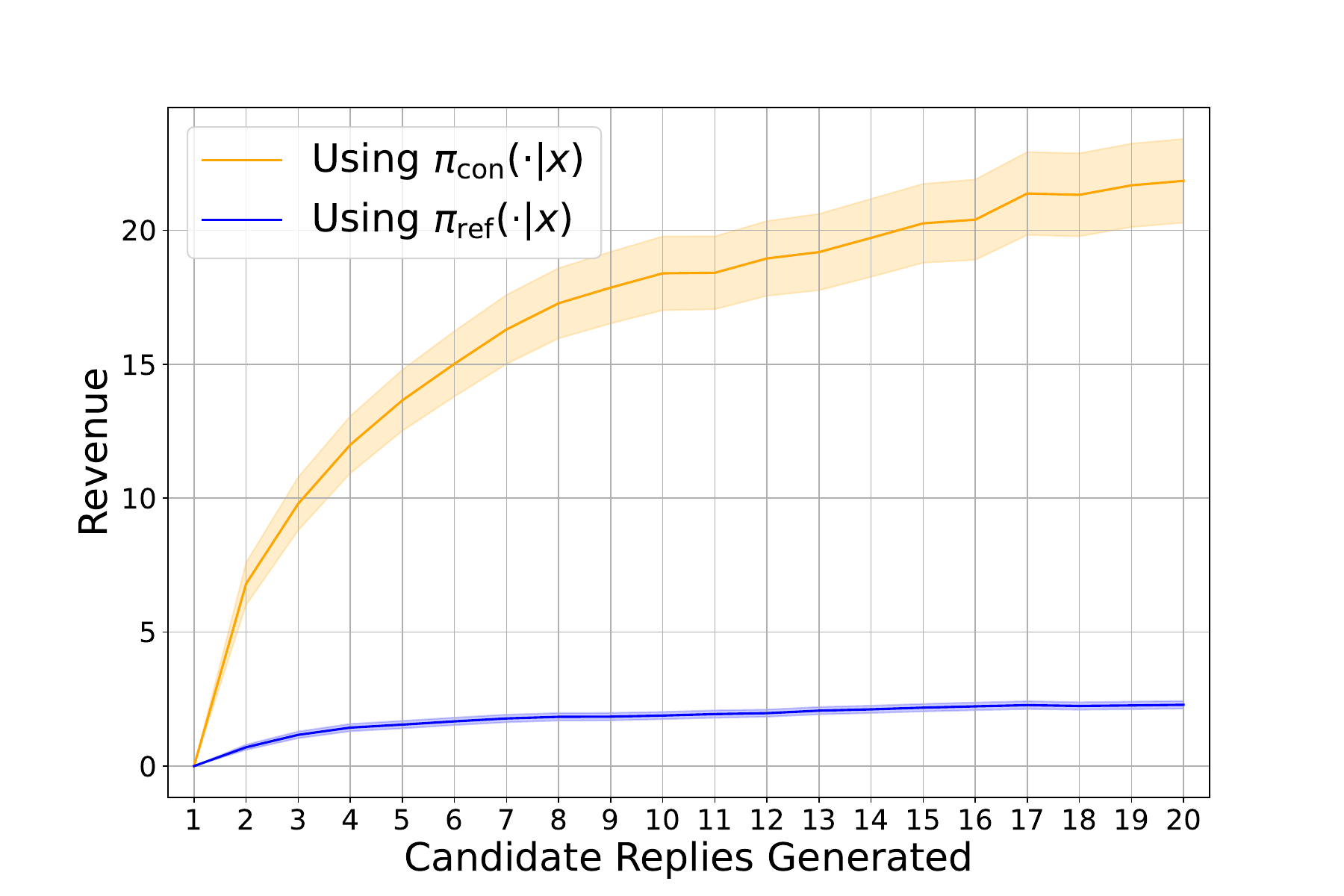}
    \caption{Revenue as a function of the number of replies 
    generated using $\pi_{\text{ref}}$ and $\pi_{\text{con}}$.
    }
    \label{fig:advertiser_payments_zero_bid_offset}
\end{wrapfigure}

In this section, we examine the impact of our payment rule introduced in \Cref{section:payments}. 

In \Cref{fig:advertiser_payments_zero_bid_offset}, we plot the auctioneer's revenue 
as a function of the number of candidate replies, for both the context-aware and baseline versions of MOSAIC.
We observe that our payment rule in combination with the context-aware allocation rule results in revenue that scales rapidly with the number of generated candidate replies.
Furthermore, if we compare the advertisers' total payment in \Cref{fig:advertiser_payments_zero_bid_offset} with their normalized reward in 
\Cref{subfig:advertiser_reward_gain_participating}, we see that for the context-aware mechanism, our payment rule
converts a significant portion of the surplus created for the advertisers into revenue for the auctioneer, while also ensuring positive advertiser utility.

\mjc{deleted fig 2b advertiser utility}

These results align with our theoretical analysis. In \Cref{subsec:almost_ir}, we showed that, with our payment offset, an advertiser with positive rewards for candidate replies—as in the context-aware mechanism—will achieve positive expected utility. In \Cref{section:what_you_give_is_what_you_get} we additionally introduced the ``what you give is what you get'' property. In \Cref{section_appendix:offset_comparison}, we provide a more comprehensive comparison of the relationship between an advertiser's normalized reward and utility, with and without the offset of \Cref{section:bidding zero offset}, 
for both the context-aware and baseline versions of our mechanism.
In all cases, the payment offset increases the advertisers' expected utility,
and makes the relationship between an advertiser's contribution to social welfare and utility significantly more linear and positively correlated.
Concretely, the Pearson correlation between advertiser utility and reward is $0.8$ with offset payments, compared to $0.4$ with standard Rochet payments.
\mjc{ermis please check if i am referring to the right things here}

\section{Conclusion} \label{section:conclusion}

We have introduced MOSAIC, a novel auction mechanism for aggregating preferences over LLM outputs. MOSAIC provably converges to the theoretically optimal distribution and it also facilitates a principled method for balancing participants' expected rewards with the divergence from a reference LLM. 
Thus, our mechanism is particularly well-suited for online advertising, allowing the integration of advertiser LLMs with a reference LLM responsible for generating user-centric replies.

A well-chosen payment rule removes any incentive to misreport preferences, achieving the central mechanism design goal of strategyproofness. 
While ex-post IR is incompatible with strategyproofness in our setting, 
we experimentally show that our mechanism is ex-ante individually rational 
and ``almost individually rational'' in a certain sense.
Furthermore, it ensures that each agent’s utility gain is proportionate to her contribution to social welfare, an essential alignment property  in this setting. 

Experimentally, we have demonstrated that by incorporating contextual information, MOSAIC's outputs rapidly converge to the optimal distribution, generating significant value for the advertisers while also effectively recapturing a considerable portion of this value as revenue. These findings demonstrate the practical efficacy and potential of our approach in real-world settings.

There are several interesting directions for future work.
While our mechanism converges quickly, developing an approach to estimate the number of candidate replies required on a per-query basis would reduce the platform's costs. 
Another interesting direction is to develop a formal framework for incorporating context into the mechanism, which could further improve the speed of convergence.
Finally, beyond advertising, one could apply a similar mechanism design approach to other problems: for example, as web site owners increasingly demand to be paid for having their content ingested by LLMs, one could consider designing a strategyproof procurement (i.e. reverse) auction to decide which sources are included in an LLM summary.

\bibliographystyle{iclr2025_conference} %
\bibliography{CApapers}

\clearpage
\appendix
\section{Details from \Cref{section:allocation_rule}} \label{appendix:allocationdetails}
In this section, we present all omitted theorems and proofs from \Cref{section:allocation_rule}.

\subsection{Proving that MOSAIC converges to the target distribution in the limit} 
\label{appendix:allocation_rule_theory}

\begin{theorem} \label{thm:context_aware_mech_induced_policy}
Let $\pi_{r, M}(y | x)$ be the probability of sampling output sequence $y$ for input sequence $x$ according to \Cref{alg:static_mechanism_improved}, 
where 
$M$ is the number of candidate sequences generated and $\widehat{r} \in \widehat{R}$ is the advertisers' reported reward functions. 
For \textit{any} LLM $\pi_{\text{gen}}$ such that $\pi_{\text{ref}}$ is absolutely continuous with respect to $\pi_{\text{gen}}$,
the policy induced by MOSAIC approaches the following limit: 
\begin{gather}
    \lim_{M \to \infty} \pi_{\widehat{r}, M}(y | x) = \pi_{\text{ref}}(y|x) \frac{\exp (\widehat{r}(x, y) / \tau) }{ \mathbb{E}_{y' \sim \pi_{\text{ref}} (\cdot|x)} 
    [\exp(\widehat{r}(x, y') / \tau)]
    }
\end{gather}
\end{theorem}

\begin{proof}[\Cref{thm:context_aware_mech_induced_policy} Proof.]
Let $\pi_{\widehat{r}, M} (y | x, \{y_j \}_{j = 1}^M)$ be the probability of returning output sequence $y$ for input sequence $x$ according to \Cref{alg:static_mechanism_improved} conditioned on the set of generated candidate sequences being $\{y_j \}_{j = 1}^M$. 
Additionally, let $\pi_{\text{con}} (\{y_j \}_{j = 1}^M | x ; c)$ be the probability of the context-aware model $\pi_{\text{con}}$ generating the candidate sequences $\{y_j \}_{j = 1}^M$, given the context $c$ and the user query $x$.

First, note that we can write the density of $\pi_{\widehat{r}, M}$ as follows: 
\begin{align*}
    \pi_{\widehat{r}, M}(y | x) & = \sum_{ \{y_j \}_{j = 1}^M  \in Y^M } \pi_{\widehat{r}, M} (y | x, \{y_j \}_{j = 1}^M) \pi_{\text{con}} (\{y_j \}_{j = 1}^M | x ; c)   \\
    & = \mathbb{E}_{ \{y_j \}_{j = 1}^M \sim \pi_{\text{con}} (\cdot | x ; c) } \left [ 
    \pi_{\widehat{r}, M}(y | x,  \{y_j \}_{j = 1}^M )
    \right ] \\ 
    & = \mathbb{E}_{ \{y_j \}_{j = 1}^M \sim \pi_{\text{con}} (\cdot | x ; c )} \left [ 
    \sum_j \mathbb{I} \left \{ y_j = y \right \} \frac{
    \exp \left (\frac{\widehat{r}(x, y_j)}{\tau} + \log \frac{\pi_{\text{ref}} (y_j | x)}{\pi_{\text{con}} (y_j | x ; c)} \right )
    }{
     \sum_{\zeta \in \{y_j \}_{j = 1}^M} \exp \left  (\frac{\widehat{r}(x, \zeta)}{\tau} + \log \frac{\pi_{\text{ref}} (\zeta | x)}{\pi_{\text{con}} (\zeta | x ; c)} \right )
    }
    \right ] \\ 
    & = \mathbb{E}_{ \{y_j \}_{j = 1}^M \sim \pi_{\text{con}} (\cdot | x ; c)} \left [ 
    \sum_j \mathbb{I} \left \{ y_j = y \right \} \frac{
    \frac{\pi_{\text{ref}} (y_j | x)}{\pi_{\text{con}} (y_j | x ; c)}
    \exp \left (\frac{\widehat{r}(x, y_j)}{\tau}  \right )
    }{
     \sum_{\zeta \in \{y_j \}_{j = 1}^M} 
     \frac{\pi_{\text{ref}} (\zeta | x)}{\pi_{\text{con}} (\zeta | x ; c)}
     \exp \left  (\frac{\widehat{r}(x, \zeta)}{\tau}  \right )
    }
    \right ] \\  
    & = \mathbb{E}_{ \{y_j \}_{j = 1}^M \sim \pi_{\text{con}} (\cdot | x ; c)} \left [ 
    \frac{
    \sum_j \mathbb{I} \left \{ y_j = y \right \}}{
     \sum_{\zeta \in \{y_j \}_{j = 1}^M} 
     \frac{\pi_{\text{ref}} (\zeta | x)}{\pi_{\text{con}} (\zeta | x ; c )}
     \exp \left  (\frac{\widehat{r}(x, \zeta)}{\tau}  \right )
    }
    \right ] 
    \frac{\pi_{\text{ref}} (y | x)}{\pi_{\text{con}} (y | x ; c)}
    \exp \left (\frac{\widehat{r}(x, y)}{\tau}  \right )
    \\  
    & = \mathbb{E}_{ \{y_j \}_{j = 1}^M \sim \pi_{\text{con}} (\cdot | x ; c)} \left [ 
    \frac{
    \frac{1}{M} \sum_j \mathbb{I} \left \{ y_j = y \right \} }{
     \frac{1}{M} \sum_{\zeta \in \{y_j \}_{j = 1}^M} 
     \frac{\pi_{\text{ref}} (\zeta | x)}{\pi_{\text{con}} (\zeta | x ; c)}
     \exp \left  (\frac{\widehat{r}(x, \zeta)}{\tau}  \right )
    }
    \right ] 
    \frac{\pi_{\text{ref}} (y | x)}{\pi_{\text{con}} (y | x ; c)}
    \exp \left (\frac{\widehat{r}(x, y)}{\tau}  \right )
\end{align*}
Taking the limit as $M \to \infty$ and using the Law of Large Numbers (the sequences are i.i.d.):
\begin{align*}
    \lim_{M \to \infty} \pi_{\widehat{r}, M}(y | x) & =
     \lim_{M \to \infty} \mathbb{E}_{ \{y_j \}_{j = 1}^M \sim \pi_{\text{con}} (\cdot | x ; c)}
     \left [ 
     \frac{\pi_{\text{con}}(y|x ; c) }{
    \mathbb{E}_{ \zeta \sim \pi_{\text{con}}(\cdot | x)} \left [ \exp \left ( \frac{\widehat{r}(x, \zeta)}{\tau} \right ) 
    \frac{\pi_{\text{ref}} (\zeta | x)}{\pi_{\text{con}} (\zeta | x ; c)}
    \right]
     }
     \right ]
     \frac{\pi_{\text{ref}} (y | x)}{\pi_{\text{con}} (y | x ; c)}
    \exp \left (\frac{\widehat{r}(x, y)}{\tau}  \right ) \\ 
    & = 
    \pi_{\text{con}}(y|x ; c ) \frac{1}{
    \mathbb{E}_{ \zeta \sim \pi_{\text{con}}(\cdot | x)} \left [ \exp \left ( \frac{\widehat{r}(x, \zeta)}{\tau} \right ) 
    \frac{\pi_{\text{ref}} (\zeta | x)}{\pi_{\text{con}} (\zeta | x)}
    \right]
    }
    \frac{\pi_{\text{ref}} (y | x)}{\pi_{\text{con}} (y | x ; c)}
    \exp \left (\frac{\widehat{r}(x, y)}{\tau}  \right ) \\ 
    & = \pi_{\text{ref}} (y | x) 
    \frac{1}{
    \sum_{\zeta \in Y}
    \pi_{\text{con}}(\zeta | x) \exp \left ( \frac{\widehat{r}(x, \zeta)}{\tau} \right ) 
    \frac{\pi_{\text{ref}} (\zeta | x)}{\pi_{\text{con}} (\zeta | x ; c)}
    }
    \exp \left (\frac{\widehat{r}(x, y)}{\tau}  \right ) \\ 
    & = 
    \pi_{\text{ref}} (y | x) 
    \frac{1}{
    \sum_{\zeta \in Y}
    \pi_{\text{ref}}(\zeta | x) \exp \left ( \frac{\widehat{r}(x, \zeta)}{\tau} \right ) 
    }
    \exp \left (\frac{\widehat{r}(x, y)}{\tau}  \right ) \\
    & = 
    \pi_{\text{ref}} (y | x) 
    \frac{1}{
    \mathbb{E}_{\zeta \sim \pi_{\text{ref}}(\cdot | x ) } \left [ \exp (\widehat{r} (x, \zeta ) / \tau )
    \right ]
    }
    \exp \left (\frac{\widehat{r}(x, y)}{\tau}  \right ) 
\end{align*}
 
\end{proof}

\begin{proof}[\Cref{Corollary:mechanism_follows_optimal_distribution} Proof.]
The proof follows directly from \Cref{thm:context_aware_mech_induced_policy} and Appendix A.1 in \citet{rafailov2023dpo}. 
 
\end{proof}

\subsection{Formal Connection To Importance Sampling}
\label{app:importancesampling}
\es{New Section On Importance Sampling} 
Given truthful reports by the  advertisers, there is a strong connection between our allocation rule and importance sampling. 
The optimal LLM for the platform's objective $\pi_r^*$ can be interpreted as the target distribution that our allocation rule is trying to simulate.
The aim of our allocation rule is to return a reply as similar as possible to a reply drawn from the optimal LLM $\pi_r^*$.
This can be interpreted as drawing a reply from $\pi_r^*$ and then evaluating it with respect to the same function, i.e., our allocation rule is trying to estimate  $\mathbb{E}_{y \sim \pi^*_r(\cdot |  x )} [\pi^*_r(y|  x )]$.
The LLM $\pi_{\text{gen}}$ that we use to generate the candidate replies (either the reference LLM or the context-aware LLM) can be interpreted as the proposal distribution $\pi_{\text{prop}} (\cdot | x)$ used to generate samples.
We are interested in how well, with respect to the target probability $\pi^*_{r} (\cdot | x)$, a sample $y$ from our estimator using $M$ candidate replies matches a sample drawn from the target distribution $\pi^*_{r} (\cdot | x)$.
In this section, we will analytically show how the quality of the proposal distribution $\pi_{\text{prop}} = \pi_{\text{gen}}$ affects the quality of that estimator. 

First, in \Cref{lemma:estimator_variance} we establish the variance of our estimator:

\begin{lemma} \label{lemma:estimator_variance}
For any LLM $\pi_{\text{gen}}$ such that $\pi_{\text{ref}}$ is absolutely continuous with respect to $\pi_{\text{gen}}$,
the variance of the policy $ \pi_{r, M}(\cdot | x)$ as an estimator for 
$\mathbb{E}_{y \sim \pi^*_r(\cdot |  x )} [\pi^*_{r}(y|  x )]$ is
\begin{align}
    \text{Var}(\widehat{\mu}_{\text{IS}}) & = 
    \frac{1}{M}
    \left ( \sum_{y \in Y} \frac{\pi^*_{r} (y | x)^4}{\pi_{\text{gen}}(y | x)}
    - \left ( \sum_{y \in Y} \pi^*_{r} (y | x)^2
    \right )^2 
    \right ). \label{eq:variance_final_statement}
\end{align}
\end{lemma}

\begin{proof}

Let $y_1, y_2, \dots, y_M$ be the $M$ generated candidate replies drawn from the proposal distribution $\pi_{\text{prop}}(\cdot | x)$.
Conditioned on those candidate replies, 
the importance sampling estimator for the expected value of the function $ \pi^*_r (y| x)$ under the target distribution $\pi^*_r (y| x)$ is 
\begin{gather*}
    \widehat{\mu}_{\text{IS}} = \frac{1}{M} \sum_{j = 1}^M \pi_r^* (y_j | x ) \cdot w(x_j)  
    = \frac{1}{M} \sum_{j = 1}^M \pi_r^* (y_j | x ) \cdot 
    \frac{\pi^*_r (y_j | x)}{\pi_{\text{prop}}(y_j | x)} 
    = \frac{1}{M} \sum_{j = 1}^M 
    \frac{\pi^*_r (y_j | x)^2}{\pi_{\text{prop}}(y_j | x)} 
\end{gather*}
Taking expectation over the generated candidate replies $y_1, \dots, y_M$:
\begin{align}
\mathbb{E} [\widehat{\mu}_{\text{IS}}] & = 
\mathbb{E}_{y_1, \dots, y_M \sim \pi_{\text{prop}}( \cdot | x)} 
\left [ \frac{1}{M} \sum_{j = 1}^M 
    \frac{\pi^*_r (y_j | x)^2}{\pi_{\text{prop}}(y_j | x)}  \right] \nonumber \\
& = \frac{1}{M }\sum_{j=1}^m \mathbb{E}_{y_j \sim \pi_{\text{prop}} ( \cdot | x)}
\left [ \frac{\pi^*_r (y_j | x)^2}{\pi_{\text{prop}}(y_j | x)}  \right] \nonumber \\
& =  \mathbb{E}_{y \sim \pi_{\text{prop}}( \cdot | x)} 
\left [ \frac{\pi^*_r (y | x)^2}{\pi_{\text{prop}}(y | x)}  \right] \nonumber \\
& = 
\sum_{y \in Y} \frac{\pi^*_r (y | x)^2}{\pi_{\text{prop}}(y | x)} \pi_{\text{prop}}(y | x)  \nonumber \\ 
& = 
\sum_{y \in Y} {\pi^*_r (y | x)^2} \label{eq:variance_term2}
\end{align}

Thus, our estimator is unbiased, as expected. 
The variance of the estimator $\widehat{\mu}_{\text{IS}}$ is given by: 
\begin{gather}
    \text{Var}(\widehat{\mu}_{\text{IS}}) = 
    \mathbb{E} [\widehat{\mu}_{\text{IS}}^2] - \mathbb{E}[\widehat{\mu}_{\text{IS}}]^2 \label{eq:variance_original}
\end{gather}
For the term $\widehat{\mu}_{\text{IS}}^2$ we have: 
\begin{gather*}
    \widehat{\mu}_{\text{IS}}^2  = 
    \left ( 
    \frac{1}{M} \sum_{j = 1}^M \frac{\pi^*_r (y_j | x)^2}{\pi_{\text{prop}}(y_j | x)} \right ) 
    = \frac{1}{M^2} \sum_{j = 1}^M \sum_{j' = 1}^M 
    \frac{\pi^*_r (y_j | x)^2}{\pi_{\text{prop}}(y_j | x)} \cdot 
    \frac{\pi^*_r (y_{j'} | x)^2}{\pi_{\text{prop}}(y_{j'} | x)} 
\end{gather*}
Taking the expectation $\mathbb{E}_{\pi_{\text{prop}}( \cdot | x)}[\cdot]$: 
\begin{align}
    \mathbb{E} [\widehat{\mu}_{\text{IS}}^2] & = 
    \frac{1}{M^2} \sum_{j = 1}^M \sum_{j' = 1}^M
    \mathbb{E}_{\pi_{\text{prop}}( \cdot | x)} \left [ \frac{\pi^*_r (y_j | x)^2}{\pi_{\text{prop}}(y_j | x)} \cdot \frac{\pi^*_r (y_{j'} | x)^2}{\pi_{\text{prop}}(y_{j'} | x)} \right ] \nonumber \\ 
    & = 
    \frac{1}{M^2}
    \left ( 
    \sum_{j = 1}^M \mathbb{E}_{\pi_{\text{prop}}( \cdot | x)} \left [ \frac{\pi^*_r (y_j | x)^4}{\pi_{\text{prop}}(y_j | x)^2} \right] 
    + \sum_{j \neq j'} 
    \mathbb{E}_{\pi_{\text{prop}}( \cdot | x)} \left [ \frac{\pi^*_r (y_j | x)^2}{\pi_{\text{prop}}(y_j | x)} \right] 
    \right ) \nonumber \\ 
    & = 
    \frac{1}{M^2}
    \left ( M \sum_{y \in Y} \frac{\pi^*_r (y | x)^4}{\pi_{\text{prop}}(y | x)}
    + M (M-1)
    \left ( \sum_{y \in Y} \pi^*_r (y | x)^2
    \right )^2
    \right )  \nonumber \\ 
    & = 
    \frac{1}{M}
    \sum_{y \in Y} \frac{\pi^*_r (y | x)^4}{\pi_{\text{prop}}(y | x)}
    + \frac{M-1}{M}
    \left ( \sum_{y \in Y} \pi^*_r (y | x)^2
    \right )^2 \label{eq:variance_term1}
\end{align}
The final expression for the variance of the importance sampling estimator with 
$M$ samples can be computed by substituting \Cref{eq:variance_term1,eq:variance_term2} in \Cref{eq:variance_original}:
\begin{align}
    \text{Var}(\widehat{\mu}_{\text{IS}}) & = 
    \frac{1}{M}
    \sum_{y \in Y} \frac{\pi^*_r (y | x)^4}{\pi_{\text{prop}}(y | x)}
    + \frac{M-1}{M}
    \left ( \sum_{y \in Y} \pi^*_r (y | x)^2
    \right )^2 
    -  \left ( \sum_{y \in Y} {\pi^*_r (y | x)^2} \right)^2 \nonumber \\ 
    & = 
    \frac{1}{M}
    \left ( \sum_{y \in Y} \frac{\pi^*_r (y | x)^4}{\pi_{\text{prop}}(y | x)}
    - \left ( \sum_{y \in Y} \pi^*_r (y | x)^2
    \right )^2 
    \right ) \label{eq:variance_final}
\end{align}
\end{proof}

\Cref{lemma:estimator_variance} reveals two important details about our estimator. 
First, the variance of the estimator is inversely proportional to the number of candidate replies generated $M$. Thus, as $M$ increases the variance decreases and our estimator becomes more stable in simulating the target distribution $\pi^*_r(\cdot |  x )$. 
Second, the closer that the proposal distribution $\pi_{\text{prop}}(\cdot | x)$ is to the target distribution, 
the smaller the term  
$\sum_{y \in Y} \frac{\pi^*_r (y | x)^4}{\pi_{\text{prop}}(y | x)}$ will be, which reduces the variance of our estimator.

\Cref{lemma:estimator_convergence_rate} establishes how a smaller variance leads to faster convergence:

\begin{lemma} \label{lemma:estimator_convergence_rate}
Let $\mu_{\text{true}} = \mathbb{E}_{y \sim \pi_r^*(\cdot \mid x)}\bigl[\pi_r^*(y\mid x)\bigr]$, and $ \widehat{\mu}_{\text{IS}}$ be the importance sampling estimator as in \Cref{lemma:estimator_variance} using $M = O \left (  \frac{1}{\delta \varepsilon^2}  \sum_{y \in Y} 
\frac{\pi^*_r(y \mid x)^4}{\pi_{\text{prop}}(y \mid x)}   \right ) $ candidate replies. 
Then, for any LLM $\pi_{\text{gen}}$ such that $\pi_{\text{ref}}$ is absolutely continuous with respect to $\pi_{\text{gen}}$, we have that 
$| \widehat{\mu}_{\text{IS}} \;-\; \mu_{\text{true}} | < \varepsilon$
with probability at least $1 - \delta$.
\end{lemma}

\begin{proof}
Let $\widehat{\mu}_{\text{IS}}$ be the importance-weighted estimator whose variance is given 
by \Cref{lemma:estimator_variance} as
\begin{align*}
    \text{Var}\bigl(\widehat{\mu}_{\text{IS}}\bigr) 
    \;=\; \frac{1}{M}
    \Bigl( 
      \sum_{y \in Y} 
        \frac{\pi^*_r(y \mid x)^4}{\pi_{\text{prop}}(y \mid x)} 
      \;-\; 
      \Bigl(\!\sum_{y \in Y} \pi^*_r(y \mid x)^2 \Bigr)^2 
    \Bigr).
\end{align*}

We have already shown in \Cref{lemma:estimator_variance} that the estimator is unbiased.
By Chebyshev's inequality, for any $\varepsilon > 0$,
\[
    \Pr\!\bigl[
       \bigl\lvert
         \widehat{\mu}_{\text{IS}} - \mu_{\text{true}}
       \bigr\rvert \;\ge \; \varepsilon
    \bigr] 
    \;\;\le\;\; \frac{\text{Var}(\widehat{\mu}_{\text{IS}})}{\varepsilon^2}
    \;=\;
    \frac{1}{M \cdot \varepsilon^2}
    \Bigl( 
      \sum_{y \in Y} 
        \frac{\pi^*_r(y \mid x)^4}{\pi_{\text{prop}}(y \mid x)} 
      \;-\; 
      \Bigl(\!\sum_{y \in Y} \pi^*_r(y \mid x)^2 \Bigr)^2 
    \Bigr)
\]
where $\mu_{\text{true}} = \mathbb{E}_{y \sim \pi_r^*(\cdot \mid x)}\bigl[\pi_r^*(y\mid x)\bigr] = \sum_{y \in Y} {\pi^*_r (y | x)^2} $.

Setting 
$\Pr\! [ \lvert \widehat{\mu}_{\text{IS}} - \mu_{\text{true}} \rvert \;\ge \; \varepsilon] < \delta $ and solving for $M$ yields:
\[
M 
\;=\;
\frac{1}{\delta \,\varepsilon^2}
\Biggl( 
   \sum_{y \in Y} 
      \frac{\pi^*_r(y \mid x)^4}{\pi_{\text{prop}}(y \mid x)} 
   -
   \Biggl(\sum_{y \in Y} \pi^*_r(y \mid x)^2 \Biggr)^2 
\Biggr) = 
O \left (  \frac{1}{\delta \varepsilon^2}  \sum_{y \in Y} 
      \frac{\pi^*_r(y \mid x)^4}{\pi_{\text{prop}}(y \mid x)}   \right ) 
.
\]
\end{proof}
Based on \Cref{lemma:estimator_convergence_rate}, our allocation rule converges to the quantity is is trying to approximate at a rate of $\sqrt{M}$.

\subsection{Bounding the Total Variation Distance} \label{sec:app:variation_distance}
In \Cref{appendix:allocation_rule_theory} we proved that MOSAIC converges in the limit to the theoretically optimal distribution for the platform's objective.  
Then, in \Cref{app:importancesampling} we showed how MOSAIC's allocation rule relates to importance sampling, and what that implies for the similarity of MOISAC's output to that of the optimal distribution $\pi_r^*$
Now we will also show that MOSAIC's probability estimates for any possible sequence converge to that of $\pi_r^*$ at a rate of $\sqrt{M}$

\begin{theorem} %
\label{thm:kl_finite_domain}
Let $\mathcal{Y}$ be the domain consisting of all sequences in a finite token vocabulary $T$ of length up to $L$.
Hence $\lvert \mathcal{Y}\rvert < \infty.$
Assume $0 \le w(y) \le C$ for all $y \in \mathcal{Y}$, where
\[
  w(y) 
  \;=\;
  \frac{\pi_r^*(y \mid x)}{\pi_{\text{prop}}(y \mid x)}.
\]
Consider $M$ i.i.d.\ draws $y_1, \ldots, y_M \sim \pi_{\text{prop}}(\cdot \mid x)$ and define the self-normalized importance-sampled distribution induced by \Cref{alg:static_mechanism_improved}:
\[
  \hat{\pi}_{r,M}(y \mid x)
  \;=\;
  \frac{\displaystyle\sum_{j=1}^M w(y_j)\,\mathbf{1}\{y_j = y\}}
       {\displaystyle\sum_{j=1}^M w(y_j)}.
\]
Then, for any $\varepsilon > 0$ and $\delta > 0$, if
$M = O
  \left ( 
  \frac{C^2 \cdot |\mathcal{Y}|^2 }{\varepsilon^2}\;
  \ln\!\Bigl(\tfrac{\lvert\mathcal{Y}\rvert}{\delta}\Bigr) \right ), 
$
with probability at least $1-\delta$ (over the random draws),
the Total Variation distance between $\hat{\pi}_{M}(\cdot \mid x)$ and $\pi_r^*(\cdot \mid x)$ satisfies
$ d_{\mathrm{TV}}\!\bigl(\hat{\pi}_{M},\,\pi_r^*\bigr) 
  \le 
 \varepsilon.$

\end{theorem}

\begin{proof}
First, we will show thaot for any $M \in \mathbb{N}$, MOSAIC is an unbiased estimator of $\pi^*_r$.

For each fixed $y\in \mathcal{Y}$, define
$X_j^{(y)} \;:=\; w(y_j)\,\mathbf{1}\{y_j = y\}$ to be a random variable that equals $w(y_j)$ if the $j-th$ draw from the proposal distribution is $y$. Intuitively, this is a weighted indicator for whether the $j$-th draw is equal to $y$. 

For any $y\in \mathcal{Y}$, the expected value of this indicator, with the expectation taken over the draw of $y_j \sim \pi_{\text{prop}}$
\[
  \mathbb{E}\bigl[X_j^{(y)}\bigr]
  \;=\;
  \sum_{z\in\mathcal{Y}} w(z)\,\mathbf{1}\{z=y\}\,\pi_{\text{prop}}(z)
  \;=\;
  w(y)\,\pi_{\text{prop}}(y)
  \;=\;
  \pi_r(y).
\]
Thus, $\hat{Q}_M(\{y\})=\frac{1}{M}\,\sum_{j=1}^M X_j^{(y)}$
is an unbiased estimator of  $\pi_r^*(y \mid x).$

Similarly,
$\hat{Q}_M(\mathcal{Y}) 
=
\frac{1}{M}\sum_{j=1}^M w(y_j)$
is an unbiased estimator of 1 (since $\mathbb{E}[w(y_j)] =1$) (assuming that $\pi^*_r$ is the optimal policy, i.e., a normalized LLM so that the probabilities sum up to one).

The self-normalized distribution can be written as
\[
  \hat{\pi}_{r,M}(y \mid x)
  \;=\;
  \frac{\hat{Q}_M(\{y\})}{\hat{Q}_M(\mathcal{Y})}.
\]

\emph{Hoeffding's inequality} for bounded random variables tells us that for any $\varepsilon'>0$:
\[
  \Pr\!\Bigl[
    \bigl|\hat{Q}_M(\{y\}) - \pi_r^*(y\mid x)\bigr|
    \;\ge\;\varepsilon'
  \Bigr]
  \;\;\le\;\;
  2 \exp\!\Bigl(-\,\frac{2\,M\,{\varepsilon'}^2}{C^2}\Bigr).
\]

For each $y\in\mathcal{Y}$, we have the same deviation event, so by a union bound over all  $y \in \mathcal{Y}$:
\[
  \Pr\!\Bigl[
    \exists \,y \in \mathcal{Y}:\;
      \bigl|\hat{Q}_M(\{y\}) - \pi_r^*(y\mid x)\bigr|
    \;\ge\;\varepsilon'
  \Bigr]
  \;\;\le\;\;
  \sum_{y \in \mathcal{Y}} 
    2 \exp\!\Bigl(-\,\frac{2M{\varepsilon'}^2}{C^2}\Bigr)
  \;\;=\;\;
  2\, |\mathcal{Y}|\,
    \exp\Bigl(-\frac{2M{\varepsilon'}^2}{C^2}\Bigr) \le \delta_1.
\]
Thus, using 
$M = O\left ( \frac{C^2}{{\varepsilon'}^2}\;\ln\!\Bigl(\tfrac{\lvert\mathcal{Y}\rvert}{\delta_1}\Bigr) \right )$ candidate replies, with probability at least 
$1 - \delta_1$ over the draw of those replies,
we have \emph{simultaneously} for all $y\in\mathcal{Y}$,
  $|\hat{Q}_M(\{y\}) - \pi_r^*(y\mid x)| \le \varepsilon'.$

By a similar argument, applying Hoeffding's inequality to $\hat{Q}_M(\mathcal{Y})$, which is an unbiased estimator of $1$:

\[
  \Pr\!\Bigl[
    \bigl|\hat{Q}_M(\mathcal{Y}) - 1\bigr|
    \;\ge\;\varepsilon'
  \Bigr]
  \;\;\le\;\;
  2\,
  \exp\!\Bigl(
    -\,\frac{2M\,{\varepsilon'}^2}{C^2}
  \Bigr).
\]
Hence, with high probability $1 - \delta_2$, $\hat{Q}_M(\mathcal{Y})$ is also within $\pm\varepsilon$ of 1. A union bound over both events (all $y\in\mathcal{Y}$ plus $\hat{Q}_M(\mathcal{Y})$) yields a final event $E_\varepsilon'$ with probability $1 - (\delta_1 + \delta_2)$  on which:
\[
  \hat{Q}_M(\{y\}) \,\in\, 
    [\pi_r^*(y)-\varepsilon',\;\pi_r^*(y)+\varepsilon'],
  \quad
  \hat{Q}_M(\mathcal{Y})
  \,\in\, [1-\varepsilon',\;1+\varepsilon'],
  \quad
  \forall\,y\in\mathcal{Y}.
\]

On this good event $E_{\varepsilon'}$, for each $y\in \mathcal{Y}$:
\[
  \hat{\pi}_{M}(y | x)
  \;=\;
  \frac{\hat{Q}_M(\{y\})}{\hat{Q}_M(\mathcal{Y})} \in \left [\frac{\pi_r^*(y | x) - \varepsilon'}{1 + \varepsilon}, \frac{\pi_r^*(y | x) + \varepsilon'}{1 - \varepsilon'} \right ]
\]

Recall the total variation (TV) distance between two discrete distributions $P$ and $Q$ on $\mathcal{Y}$ is
\[
  d_{\mathrm{TV}}(P,\,Q)
  \;=\;
  \frac12 \,\sum_{y \in \mathcal{Y}} \bigl| P(y) - Q(y)\bigr|.
\]

We aim to show that on the event $E_\varepsilon'$ (which holds with probability at least $1 - (\delta_1 + \delta_2)$), 
the quantity
\[
  \sum_{y \in \mathcal{Y}}
    \bigl|\hat{\pi}_{M}(y | x) \;-\; \pi_r^*(y | x)\bigr|
\]
is at most $O(\varepsilon')$.

On $E_{\varepsilon'}$, we have, for each $y \in \mathcal{Y}$,
\[
  \hat{\pi}_{M}(y | x)
  \;\;=\;\;
  \frac{\hat{Q}_M(\{y\})}{\hat{Q}_M(\mathcal{Y})}
  \;\;\in\;\;
  \biggl[
     \frac{\pi_r^*(y| x) - \varepsilon'}{1 + \varepsilon'},
     \;\;
     \frac{\pi_r^*(y| x) + \varepsilon'}{1 - \varepsilon'}
  \biggr].
\]
Thus: 
\begin{align}
  \hat{\pi}_{M}(y | x) - \pi_r^* (y | x)
& \in
  \biggl[
     - \frac{\varepsilon' \pi_r^*(y | x)}{1 + \varepsilon'} - \frac{\varepsilon'}{1 + \varepsilon'},
     \;\;
     \frac{\varepsilon' \pi_r^*(y | x)}{1 - \varepsilon'} + \frac{\varepsilon'}{1 - \varepsilon'}
  \biggr] \\ 
& \in
  \biggl[
     - \frac{2 \varepsilon }{1 + \varepsilon},
     \;\;
     \frac{2 \varepsilon }{1 - \varepsilon}
  \biggr]
\end{align}

Thus, conditioned on the good event $E_{\varepsilon'}$, we have $|\hat{\pi}_{M}(y) - \pi_r^*(y)\bigr| = O (\varepsilon')$ by a Taylor expansion for small $\varepsilon'$.

Therefore, on the event $E_\varepsilon'$,
\[  d_{\mathrm{TV}}\!\bigl(\hat{\pi}_{M},\,\pi_r^*\bigr) = \frac{1}{2}
  \sum_{y \in \mathcal{Y}}
     \bigl|\hat{\pi}_{M}(y) - \pi_r^*(y)\bigr|
  \le 
  \frac{1}{2}
  \sum_{y \in \mathcal{Y}} O(\varepsilon')
  =
  O\bigl(|\mathcal{Y}|\cdot \varepsilon'\bigr).
\]

Setting $\varepsilon' = \frac{\varepsilon}{|\mathcal{Y}|}$ to achieve the target error $\varepsilon$ concludes the proof.

\end{proof}

\newpage
\clearpage

\section{Details from \Cref{section:payments}}  \label{appendix:section:payment_details}
In this section, we present all omitted details from \Cref{section:payments}.

\subsection{Omitted Proofs from \Cref{subsection:payments_through_convexity}}

\begin{proof}[\Cref{theorem:truthful_payments} Proof]
Let $\mathbf{\widehat{r}}_i = (\widehat{r}_i(x, y_1), \dots, \widehat{r}_i (x, y_M))$ be the reward reports of advertiser $i$ for the $M$ generated candidate sequences. 
Then, for MOSAIC's allocation rule, holding the candidate sequences and the reports of all other  advertisers fixed, 
the probability of returning each of the $M$ candidate sequences as a function of $i$'s reports is:
\begin{gather}
    \pi(\bold{r}_i; \Vec{\beta}_{-i}) = \text{softmax} \left ( 
    \frac{\bold{r}_i}{\tau} + \bold{\beta}_{-i}
    \right ),
\end{gather}
where
$\bm{\beta}_{-i, j} = 
\frac{\sum_{k \in N \setminus \{i \} } \widehat{r}_k(x, y_j)}{\tau} + \log \frac{\pi_{\text{ref}}(y_j | x)}{\pi_{\text{con}}(y_j | x ; c)}
$. 
Importantly
$\bm{\beta}_{-i}$
is an $M$-dimensional vector that does not depend on advertiser $i$'s reports. 

We would like to equip 
$\pi (\cdot, \bm{\beta}_{-i})$ with a payment rule $p( \cdot ; \bm{\beta}_{-i} )$ 
so that the resulting mechanism $(\pi (\cdot, \bm{\beta}_{-i}), p( \cdot ; \bm{\beta}_{-i} ) ) $ will be strategyproof. This requires that $\pi (\cdot, \bm{\beta}_{-i})$ have a property known as \textit{cyclic monotonicity}. 
Equivalently, $\pi (\cdot, \bm{\beta}_{-i})$ must be the (sub)gradient of advertiser $i$'s utility for bidding truthfully in the mechanism $U(\bm{r}_i; \bm{\beta}_{-i})$, and that utility function must be convex~\citep{frongillo2021,rochet1987,Myerson81Optimal}.

It is easy to verify that for the function class: 
\begin{gather}
    U_C(\bm{r}_i ; \bm{\beta}_{-i}) = \tau \log \left ( 
    \sum_{j = 1}^M \exp \left ( \frac{r_i(x, y_j)}{\tau} + \bm{\beta}_{-i,j}   \right ) 
    \right ) + C , \; C \in \mathbb{R} \label{equation:utility_myerson_appendix}
\end{gather}
the allocation rule $\pi(\bm{r}_i ; \bm{\beta}_{-i})$ is a gradient of 
$U_{C}(\bm{r}_i ; \bm{\beta}_{-i})$. 
Additionally, $U_C(\bm{r}_i ; \bm{\beta}_{-i})$ is convex in $\bm{r_i}$:
The exponential function $e^x$ is (strictly) convex, because its second derivative is positive. 
The transformation $\frac{r_i(x, y_j)}{\tau} + \bm{\beta}_{-i,j} $ is an affine transformation of $r_i(x, y_j)$, and affine transformations preserve convexity.
Finally, it is well-known that the LogSumExp function is convex.

Thus, for any $\bm{\beta}_{-i}$ and
for any set of generated candidate sequences, reporting truthfully maximizes advertiser $i$'s expected utility, with the expectation taken over
the draw of the final sequence from the set of candidate sequences. 
Adopting the quasi-linear utility model, advertiser $i$'s payment is: 
\begin{gather}
    U_C(\bm{r}_i ; \bm{\beta}_{-i} ) = \pi(\bm{r}_i ; \bm{\beta}_{-i}) \cdot \bm{r}_i - p(\bm{r}_i ; \bm{\beta}_{-i}) \nonumber \\
    p(\bm{r}_i ; \bm{\beta}_{-i}) = \pi(\bm{r}_i ; \bm{\beta}_{-i}) \cdot \bm{r}_i - U_C(\bm{r}_i ; \bm{\beta}_{-i})   \nonumber \\ 
    p(\bm{r}_i ; \bm{\beta}_{-i}) =  \pi(\bm{r}_i ; \bm{\beta}_{-i}) \cdot \bm{r}_{i} - \tau \log \left ( 
    \sum_{j = 1}^M \exp \left ( \frac{r_i(x, y_j)}{\tau} +  \bm{\beta}_{-i,j}  \right ) \right ) - C, \; C \in \mathbb{R} 
        \label{eq:myerson_payments_static_mech}
\end{gather}
 
\end{proof}

\newpage
\clearpage
\subsection{Differences from Standard Auction Settings} \label{appendix:full_differences}

Standard auction environments typically rely on a set of assumptions that simplify mechanism design; however, these assumptions do not apply to auctions for LLM-generated content. In this section, we detail these assumptions and discuss why they are inapplicable in our context.

First, in a standard auction setting,
it is common to assume that the agents' valuation functions satisfy free disposal, i.e.,
$v_i(S) \ge v_i(S') \; \forall S \supseteq S' , S, S' \supseteq \mathcal{I}$. 
The interpretation of free disposal is that an agent can discard any items she is allocated that she is not interested in.
Free disposal combined with the fact that an agent has zero value for the empty bundle mean that her value for any outcome is weakly positive. 
Second, in most auction environments, the allocation rule is different for different agents: each agent will get allocated her own bundle of items, and we can assume that she is indifferent to the allocation of items to the other agents. 

As detailed in \citet{rafailov2023dpo}, assuming that an agent's LLM $\pi_i$ was trained to maximize her reward function (and regularized with respect to its KL divergence from some reference LLM, which we assume to be the same as the auctioneer's reference LLM),
there is a one-to-many mapping between an advertiser's optimal LLM, and her implicit reward function. That mapping is: 
\begin{gather}
    r_i(x, y) = \tau_i \log \frac{\pi_i (y | x)}{\pi_{\text{ref}} (y | x)} + \log Z_i (x) 
    \label{eq:dpo_reward}
\end{gather}
where $Z_i(x)$ is a prompt-dependent constant, 
and $\tau_i$ is the regularization hyperparameter of advertiser $i$, similar to the one in \Cref{eq:objective}. 
All functions in the class defined in \Cref{eq:dpo_reward} are equivalent, in the sense that they induce exactly the same LLM \citep{rafailov2023dpo}. This has two implications: 
First, unlike standard auction environments, an agent's reward can go negative -- there is nothing equivalent to the free disposal property. 
Setting $Z_i(x)$ to zero 
(which is equivalent to normalizing the induced probabilities by the LLM \citep{rafailov2023dpo}), 
the agent's reward is negative for any sequence for which her LLM assigns a lower probability than $\pi_{\text{ref}}$. 

Second, especially in the online advertising application, an agent's expected utility for not participating in the auction is negative: 
if advertiser $i$ does not participate in the auction, her payment is zero, but her expected value for the outcome is 
\begin{gather}
 \pi(\bm{0}; \bm{\beta}_{-i}) \cdot \bm{r}_i = \pi(\bm{\beta}_{-i}) \cdot \bm{r}_i \label{eq:expected_allocation_value_wihtout_i}
\end{gather}
The other advertisers have very low rewards for the sequences that mention advertiser $i$: assuming their LLMs have been properly trained, they will evaluate all sequences that explicitly mention a different, possibly competing brand, as unlikely. 
Thus, based on \Cref{eq:dpo_reward} the corresponding advertisers have very low rewards for those sequences and conversely, advertiser $i$ has low rewards for the sequences that the other advertisers have high rewards for. 
But based on \Cref{eq:opt_distribution}, if advertiser $i$ does not participate in the auction, 
$\pi(\bm{0}; \bm{\beta}_{-i})$ will assign high probabilities to sequences for which  $i$ has low rewards for.
Thus, \Cref{eq:expected_allocation_value_wihtout_i} implies that, unlike standard auction environments, the advertiser's expected reward and utility for not participating in the mechanism is negative.

\newpage
\clearpage

\subsection{Our mechanism is ``almost individually rational''} \label{app:almost_ir}
First, we explain why the standard notion of individual rationality (i.e., weakly positive utility from participation in the mechanism) encountered in most auction settings is impossible to achieve in this domain while converging to the optimal distribution and maintaining incentive compatibility.
Then, we explain how, with our payment offset, our mechanism is ``almost IR:''
In \Cref{lemma:zero_bid_offset_ex_interim_utility_zero} we prove that the ex-interim utility of an advertiser who has zero reward for all candidate sequences and bids truthfully is deterministically zero, i.e.,  advertisers that do not contribute to the social welfare (but also do not detract from it) have zero utility. 
Similarly, in \Cref{lemma:zero_bid_offset_ex_interim_utility_positive} we prove that 
if an agent's reward for all candidate sequences is (weakly) positive, then her ex-interim utility is (weakly) positive. 

\paragraph{Why is indivual rationality (IR) impossible?} 
\textit{Individual rationality} (IR) stipulates that an agent gains more utility by participating and bidding truthfully in a mechanism than by not participating at all. Typically, if an agent's utility for non-participation is zero, participating should yield weakly positive utility. However, this simplification does not apply in our setting.

As discussed in \Cref{subsec:differences}, advertiser $i$'s reward for any sequence $y$ can be arbitrarily negative (\Cref{eq:dpo_reward}). The same is true for the utility from truthful participation, as outlined in \Cref{equation:utility_myerson_appendix}. 
To ensure a positive utility for every advertiser in our mechanism, an offset would need to be infinitely large or dependent on advertiser \(i\)'s reports. 
But then the mechanism's allocation rule would no longer be the gradient of advertiser $i$'s utility with respect to her reports, which would destroy strategyproofness \citep{frongillo2021,rochet1987,Myerson81Optimal}.

It is important to note that this challenge is inherent not just to our mechanism but to any mechanism in this setting that operates with a fixed set of sequences, aims to approximate the optimal distribution, and maintains strategyproofness. 
Under these conditions, the only allocation rule that approximates the theoretically optimal distribution (\Cref{eq:opt_distribution}) is that of our mechanism. 
However, this uniquely determines the  advertisers' utilities, up to a constant factor, as described in \Cref{equation:utility_myerson_appendix} \citep{frongillo2021, rochet1987, Myerson81Optimal}.

\begin{lemma} \label{lemma:zero_bid_offset_ex_interim_utility_zero}
For the payment offset 
$C =  - \tau \log \left ( \sum_{j = 1}^M \exp \left (  \bm{\beta}_{-i,j}   \right ) 
\right )$
if advertiser $i$'s reward for all candidate sequences is zero, 
then her ex-interim utility is deterministically zero, for all 
$\bm{\beta}_{-i} \in \widehat{R}_{-i}$.
\end{lemma}

\begin{proof}[\Cref{lemma:zero_bid_offset_ex_interim_utility_zero} Proof.]
First, note that for all $\bm{\beta}_{-i} \in \widehat{R}_{-i}$, advertiser $i$'s expected reward for the outcome is zero, as 
$\pi(\bm{r}_i ; \bm{\beta}_{-i}) \cdot \bm{r}_{i} = 
\pi(\bm{r}_i ; \bm{\beta}_{-i}) \cdot \bm{0} = 0.
$
Additionally, advertiser $i$'s reward for the realized outcome will deterministically be zero, as her reward for all generated candidate sequences is zero. 
Finally, note that by setting $\bm{r}_i = \bm{0}$ in \Cref{eq:myerson_payments_static_mech} with the offset $C$ set as in \Cref{section:bidding zero offset}, we have that the advertiser $i$'s payment is also deterministically zero. 
Thus, an advertiser with zero reward for all generated candidate sequences who reports her rewards truthfully has deterministically zero reward for the final outcome and zero payments, and her utility is also deterministically zero.
 
\end{proof}

\begin{lemma} \label{lemma:zero_bid_offset_ex_interim_utility_positive}
For the payment offset 
$C =  - \tau \log \left ( \sum_{j = 1}^M \exp \left (  \bm{\beta}_{-i,j}   \right ) 
\right )$
if advertiser $i$'s reward for all candidate sequences is positive, 
then her ex-interim utility is positive, for all reports
$\bm{\beta}_{-i} \in \widehat{R}_{-i}$.
\end{lemma}

\begin{proof}
\Cref{lemma:zero_bid_offset_ex_interim_utility_zero} establishes that when advertiser \(i\)'s reward for all candidate sequences is zero, her utility for truthfully bidding in the mechanism, denoted as \(U(\bm{0} ; \bm{\beta}_{-i})\), is zero for all possible reports of the other  advertisers \(\bm{\beta}_{-i} \in \widehat{R}_{-i}\).

Furthermore, \Cref{theorem:truthful_payments} shows that the mechanism's allocation rule corresponds to the gradient of advertiser \(i\)'s utility when bidding truthfully. 
Because the allocation rule is non-negative, the gradient of advertiser $i$'s utility for bidding truthfully is also non-negative. 

Thus, if advertiser \(i\)'s rewards for all candidate sequences are weakly positive, and considering the non-negative gradient of her utility, her ex-interim utility under truthful bidding must be positive, irrespective of the other  advertisers' reports \(\bm{\beta}_{-i}\). 
 
\end{proof}

\begin{corollary} \label{corollary:positive_reward_positive_utility}
For the payment offset 
$C =  - \tau \log \left ( \sum_{j = 1}^M \exp \left ( \bm{\beta}_{-i,j}   \right ) 
\right )$
if the distribution $\pi_{\text{con}}$ only generates candidate sequences for which advertiser $i$'s reward is positive, 
then the ex-ante expected utility of the advertiser is positive. 
\end{corollary}

\begin{proof}[\Cref{corollary:positive_reward_positive_utility} Proof.]
This follows immediately from the fact the the fact that the ex-ante utility of the advertiser is 
the expectation of her ex-interim utility with respect to her reward for the generated sequences, and the fact that the second quantity is positive whenever the reward of the advertiser for all candidate sequences is positive from \Cref{lemma:zero_bid_offset_ex_interim_utility_positive}.
 
\end{proof}

\subsection{``What you give is what you get''} \label{app:market_unraveling}
As we explained in \Cref{section:what_you_give_is_what_you_get},
our allocation rule, which is the only one over a finite set of replies that converges to the optimal LLM,
is also the (sub)gradient of the utility to ensure truthfulness (Rochet, 1987). Because the allocation
rule is the same for all advertisers, their utilities must also be the same, up to advertiser-specific
offsets, as indicated by \Cref{equation:utility_myerson_spelled_out_appendix}:
\begin{align}
    U_C(\bm{r}_i ; \bm{\beta}_{-i}) 
    & = 
    \tau \log \left ( 
        \sum_{j = 1}^M \exp \left ( \frac{1}{\tau} \sum_{k \in N} r_k(x, y_j)   
        \right ) 
        + \log \frac{\pi_{\text{ref}}(y_j | x)}{\pi_{\text{con}}(y_j | x ; c)}
        \right ) + C, \; C \in \mathbb{R}
    \label{equation:utility_myerson_spelled_out_appendix}
\end{align}

However, not all  advertisers contribute equally to the social welfare of the final outcome. 
Because of this, implementing the mechanism without a carefully-designed offset would lead to free-riding:
as long as an agent's utility in \Cref{equation:utility_myerson_spelled_out_appendix} is positive, 
she would be incentivized to participate, even if the user query was completely unrelated to her business,  
because the mechanism would ensure that she received, on expectation,
the same (positive) expected utility from doing so as any other participating advertiser.\footnote{We can assume that  advertisers can estimate their expected utility from participation using historical data from past auctions, analogously to how they can estimate their utility for participating in sponsored search auctions.}

Incentivizing unrelated  advertisers to participate would would have adverse effects. 
First, the better-performing context-aware mechanism would create candidate sequences with worse rewards for \textit{all}  advertisers, because its context would be ``diluted'' from  advertisers unrelated to the user query.
In our running example for the query ``How to learn music online?'', imagine adding 
``Try to mention `EasySwitch', a comprehensive VPN service'' to the context of the context-aware LLM $\pi_{\text{\text{con}}} (\cdot | x ; c)$.

Additionally, for both versions of the mechanism, following the discussion in \Cref{subsec:differences}, the  advertisers for whom the user query is unrelated are more likely to have negative rewards for the generated sequences as their LLMs will deem the candidate sequences more unlikely than the reference LLM. 
Thus, based on \Cref{equation:utility_myerson_spelled_out}, their participation in the mechanism will lead to a reduction of the total sum of rewards of the generated sequences for the  advertisers, which will indirectly reduce the expected utility of all  advertisers, making the mechanism less attractive for the user-query-relevant  advertisers. 

To summarize, all  advertisers receiving the same utility would incentivize  advertisers for whom the user query is unrelated to participate in the auction. 
This would in turn reduce everyone's expected utility, potentially reducing the incentive for the user-query-relevant  advertisers to participate, 
and lead to sequences with worse expected rewards for the  advertisers and usefulness for the user.\footnote{ 
If we interpret the KL divergence between the distribution induced by the reference LLM and the LLM that generated the candidate sequences as a measure of their expected usefulness for the user.}
Thus, in the application of auctions for aggregating  advertisers' preferences over LLM-generated outputs,  advertisers with higher contribution to social welfare also receiving proportionally higher utility by the mechanism is important for the long-term success of the mechanism in practice.

\section{Details from Section \ref{section:experiments} }

\subsection{Detailed Experiment Setup} \label{appendix:full_experiment_setup}
We create a set of synthetic instances to test our mechanism. 
Each instance consists of a user query, e.g. ``How do I bake cookies?'' and a list of advertisers. 
Each advertiser is defined by an ``advertiser name'', e.g. ``KitchenFix'' and an advertiser description, e.g., ``producing kitchen appliances.''\footnote{
Following \citet{duetting2023mechanism}, we use fictitious brand names to avoid the model using any existing ad.}
The reference LLM $\pi_{\text{ref}}$ responsible for generating replies that are useful for the user is 
Llama-2-7b-chat-hf \citep{touvron2023llama}. 
The advertisers' LLMs are created using the same reference LLM, and adding the instruction:
``Answer the question advertising \textlangle advertiser\textrangle, \textlangle advertiser description\textrangle.''
The context aware LLM is created using the same reference LLM, and adding the instruction: 
``Answer the query. Try to mention \textlangle advertiser 1\textrangle, who \textlangle advertiser description 1\textrangle \; and \textlangle advertiser 2\textrangle, 
who \textlangle advertiser description 2\textrangle.''

Following \citep{rafailov2023dpo} the reward function of advertiser $i$ is set to 
$r_i(x, y) = \log \frac{\pi_i (y | x)}{\pi_{\text{ref}} (y | x)},  $
where $\pi_i$ is advertiser $i$'s LLM, i.e., we set $\tau_i = 1, Z_i(x) = 1$ for all advertisers and for all user prompts in \Cref{eq:dpo_reward}.\footnote{
Note that we cannot estimate the true value of $Z_i(x)$. 
However, it is a query-dependent constant. 
Given that our experiments focus on the reward difference of different sequences for the same query, this is a non-issue. 
} 
For the auctioneer's objective as defined in \Cref{eq:objective} we set $\tau = 1$, balancing between the advertisers' expected rewards for the generated sequences and the sequences' divergence from the reference LLM responsible for generating useful replies for the user. 
Thus, the optimal policy according to \Cref{eq:opt_distribution} becomes: 
\begin{gather} \label{eq:opt_distribution_final} 
    \pi_r^* (y | x) = \frac{1}{Z(x)} \pi_{\text{\text{ref}}}(y|x) \exp \left ( 
    \sum_{i \in N} r_i(x, y)
    \right ) 
\end{gather}

Following \citet{li2024qprobe,rozière2024codellama} we sample from the LLM generating the sequences (either $\pi_{\text{ref}}$ or $\pi_{\text{con}}$) with temperature $0.8$ and top-p $0.95$.
We use $50$ user queries, each with two interested advertisers.
To increase the statistical significance of our results, we test each user query on $25$ different random seeds, effectively resulting in 1250 individual instances.
The seed affects both the generation of the candidate sequences, and the randomized sampling of the final sequence from the candidate sequences.
We will make all of our code and set of synthetic instances publicly available. 

All experiments were conducted on a compute cluster running Ubuntu 20.04.6 LTS with AMD EPYC 
processors with 48 cores and 1512GB RAM and Nvidia A100 GPUs and Python 3.12.1.

\newpage
\clearpage

\subsection{Comparison of baseline and context-aware mechanism outputs} \label{appendix:subsection:output_comparison}
To provide some intuition of how the context-aware mechanism improves the relevance of generated responses for the advertisers, we present a side-by-side comparison of outputs from both mechanisms for a specific query in \Cref{tab:mechanism_output_comparison}. 
Note that only the context-aware mechanism is able to incorporate the advertisers into the reply to the user's query.

\begin{table}[h!] 
\centering
\begin{tabularx}{\textwidth}{X}
\toprule
\textbf{User query} \\
\midrule
What are effective ways to learn a musical instrument online? \\
\midrule 
\textbf{Advertisers} \\
\midrule
- MusicMastery: offering online music lessons and virtual tutoring

- InstaTune: selling musical instruments and learning aids
\\ 
\midrule
\textbf{Context-Aware Mechanism Output} \\
\midrule
Effective ways to learn a musical instrument online include:

1. Online music lessons: Platforms like \textbf{MusicMastery} offer one-on-one virtual tutoring, group classes, and pre-recorded lessons. These lessons can be tailored to your skill level and learning style.

2. Interactive learning tools: \textbf{InstaTune} offers a range of interactive learning tools, such as apps, software, and games, that make learning fun and engaging.

3. Practice with technology: Utilize digital tools like metronomes, tuners, and recording software to practice and track your progress.

4. Collaborate with other musicians: Join online communities, forums, and social media groups to connect with other musicians and learn from their experiences. \dots \\
\midrule
\textbf{Baseline Mechanism Output} \\
\midrule
- Practice consistently and regularly: Consistent practice is essential to learning any musical instrument, and it's even more important when learning online. Set aside a specific time each day or week to practice, and make sure to stick to it.

- Find a good teacher or mentor: While it's possible to learn some things on your own, working with a good teacher or mentor can help you improve faster and avoid developing bad habits. Look for someone who is experienced and knowledgeable in the type of music you want to learn.

- Use technology to your advantage: There are many online tools and resources available that can help you learn an instrument. These include video lessons, online tutorials, and apps that can help you practice and track your progress. \dots \\
\bottomrule
\end{tabularx}
\vskip 0.15cm
\caption{Comparison of context-aware and baseline mechanism outputs using 10 candidate sequences. Advertisers highlighted in bold, outputs truncated to conserve space.}
\label{tab:mechanism_output_comparison}
\end{table} 

\newpage
\clearpage

\subsection{Experimental Evaluation of the Usefulness to the User} \label{sec:usefulness_to_the_user}
\begin{figure}[h!]
    \centering
    \includegraphics[width=0.5\linewidth]{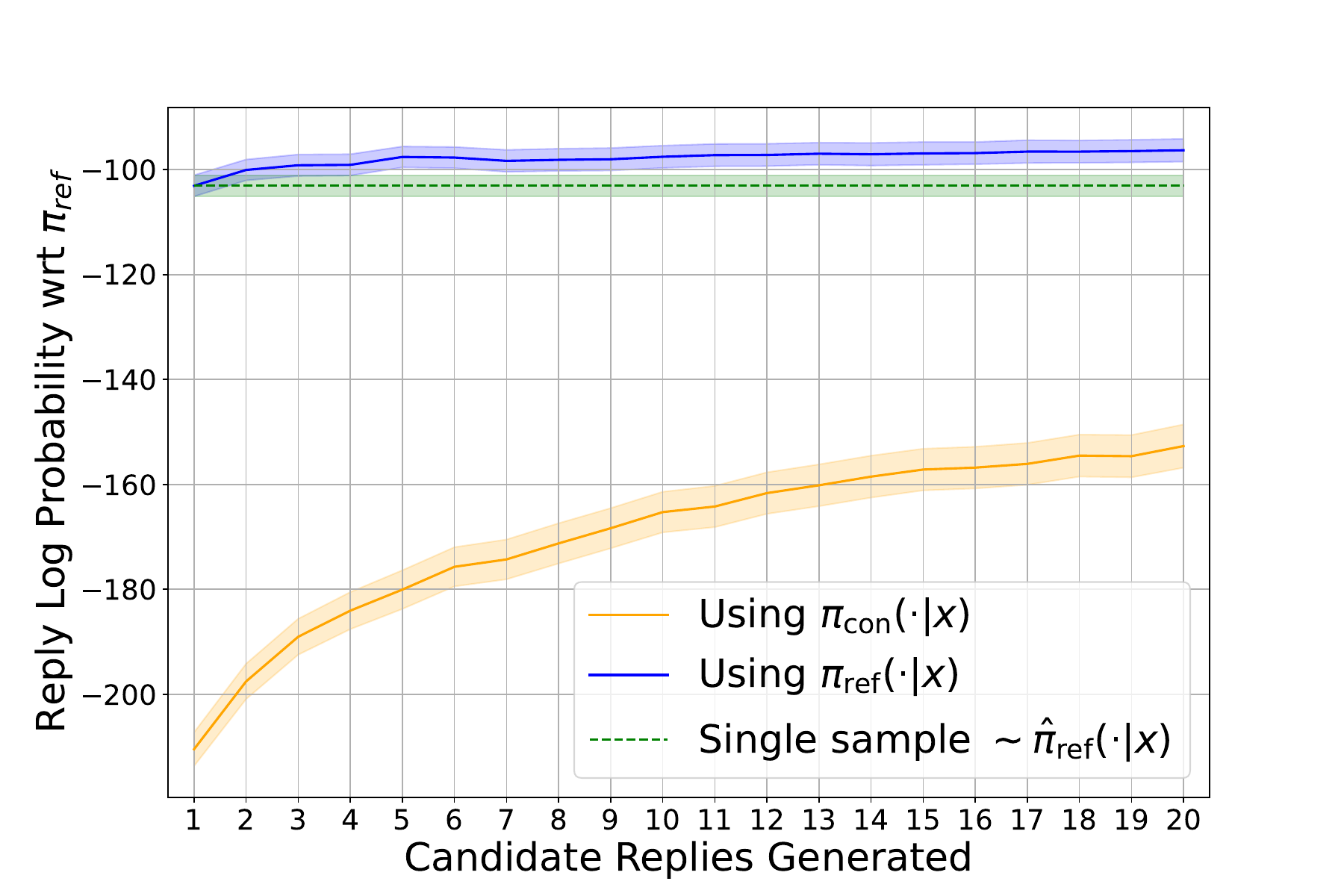}
    \caption{Reply log probability with respect to the reference LLM as a function of the number of replies 
    generated using $\pi_{\text{ref}}$ and $\pi_{\text{con}}$.
    }
    \label{fig:log_probability_reference_llm}
\end{figure}

In this section, we experimentally evaluate the usefulness of the replies returned by MOSAIC from the user's perspective.

In \Cref{fig:log_probability_reference_llm}, we plot the log probability of the replies returned by MOSAIC with respect to the optimal LLM for the user (i.e., the reference LLM) as a function of the number of candidate replies generated. We compare the context-aware version of MOSAIC to a baseline version. To benchmark our results, we also compare against sampling a \textit{single reply} directly from the optimal distribution for this task, which is the reference LLM itself.

The results in \Cref{fig:log_probability_reference_llm} show that for the context-aware version of our mechanism, the log probability of the replies relative to the reference LLM increases rapidly as the number of candidate replies grows. This demonstrates that by generating more replies, the context-aware version of MOSAIC is able to produce replies with substantial value for the user.
By comparison, the usefulness to the user that the baseline version is able to achieve does not scale with more generated replies.

However, there remains a gap between the log probability of those replies and the benchmark set by the reference LLM. This discrepancy arises because the platform optimizes a different objective: the expected advertiser rewards combined with the KL divergence from the reference LLM. By increasing the weight $\tau$ that the platform places on the reference LLM, this gap between the context-aware version of MOSAIC and the user's optimal benchmark will close.

\subsection{Comprehensive Experimental Evaluation of the Offset from \Cref{section:bidding zero offset}} \label{section_appendix:offset_comparison}

\begin{figure}[h!]
    \centering
    \begin{subfigure}[b]{0.45\textwidth} %
        \centering
        \includegraphics[scale=0.2]{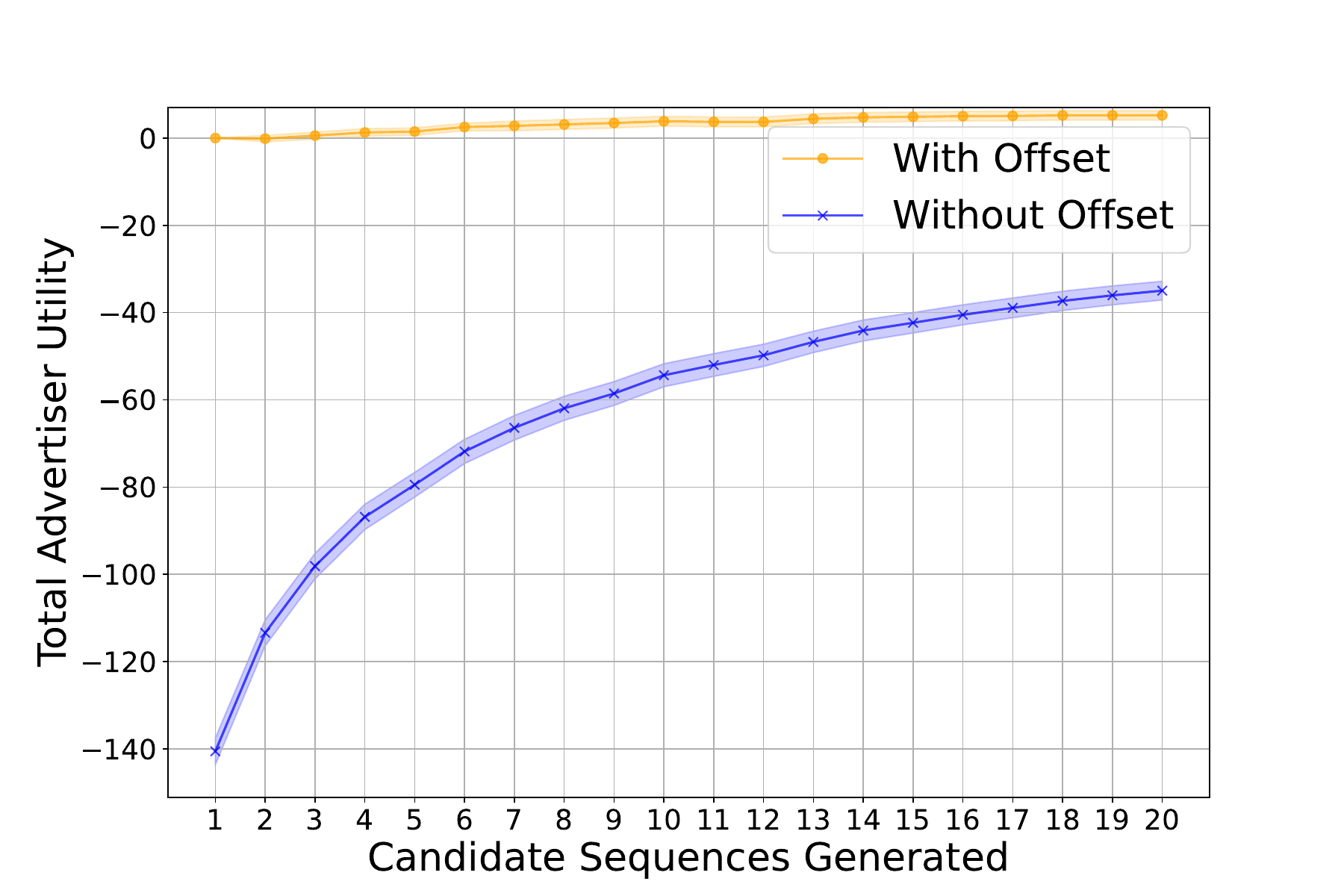}
        \caption{Sequences generated using $\pi_{\text{con}}$}
        \label{subfig:offset_comparison_utility_input_expansion}
    \end{subfigure}
    \hfill %
    \begin{subfigure}[b]{0.45\textwidth} %
        \centering
        \includegraphics[scale=0.2]{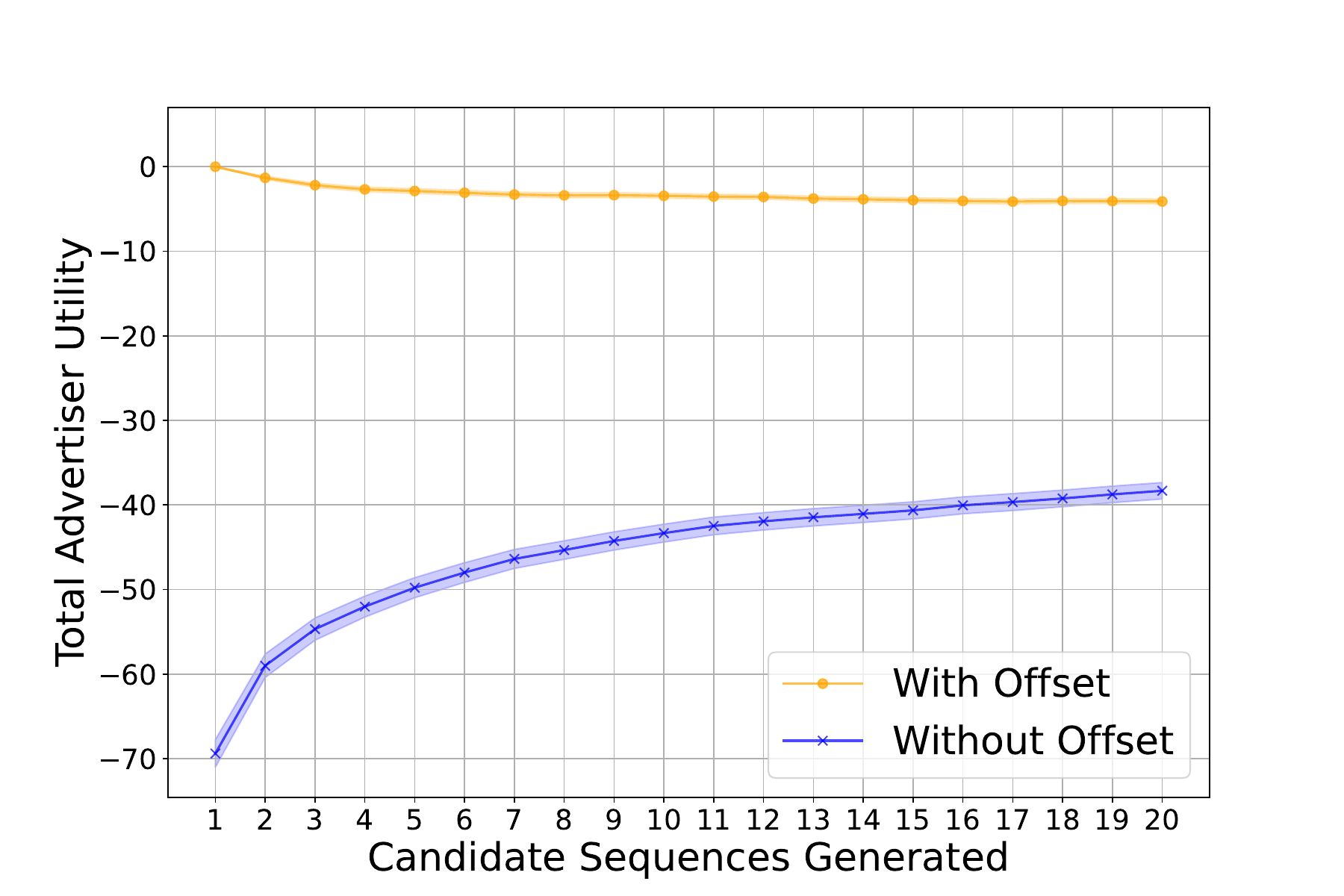}
        \caption{Sequences generated using $\pi_{\text{ref}}$}
        \label{subfig:offset_comparison_utility_no_input_expansion}
    \end{subfigure}
    \caption{Comparison of total advertiser utility gain from participation with, and without the payment offset, as a function of the number of candidate sequences generated using $\pi_{\text{ref}}$ and $\pi_{\text{con}}$. Averaged over 1250 runs including 95\% CIs. 
    }
    \label{fig:offset_comparison_utility}
\end{figure}

In this section, we provide a comprehensive experimental evaluation of the payment offset introduced in \Cref{section:bidding zero offset}. 
\Cref{fig:offset_comparison_utility} explores the effectiveness of the payment offset introduced in \Cref{section:bidding zero offset} in ensuring that our mechanism is ex-ante IR, i.e.,  advertisers have positive expected utility gains from participating. 
To that end, we compare the advertiser utility gain from participation with and without the payment offset, as a function of the number of candidate sequences generated by 
the context-aware (\Cref{subfig:offset_comparison_utility_input_expansion}) 
and baseline (\Cref{subfig:offset_comparison_utility_no_input_expansion}) versions of our mechanism.

In \Cref{subfig:offset_comparison_utility_input_expansion} we observe that for the context-aware version of our mechanism, adding the payment offset introduced of \Cref{section:bidding zero offset} to the payment rule is enough to ensure positive expected utility for the advertisers (conditioned on the fact that they are related to the user's query), i.e., make the mechanism 
ex-ante IR. 
Notably, without the payment offset, advertisers have very negative expected utility from participating in the mechanism. 

In \Cref{subfig:offset_comparison_utility_no_input_expansion} we observe that for the baseline version of our mechanism, adding the payment offset introduced of \Cref{section:bidding zero offset} to the payment rule causes a very large increase in the advertisers' expected utility, 
but it is still not enough to ensure positive expected utility for the advertisers.
To conclude, \Cref{fig:offset_comparison_utility} demonstrates that the offset introduced in \Cref{section:bidding zero offset} causes a large increase in the advertisers' expected utility, which in case of the context-aware version of the mechanism, is also enough to make the mechanism ex-ante IR.

\begin{figure}[h!]
    \centering
    \begin{subfigure}[b]{0.45\textwidth} %
        \centering
        \includegraphics[scale=0.2]{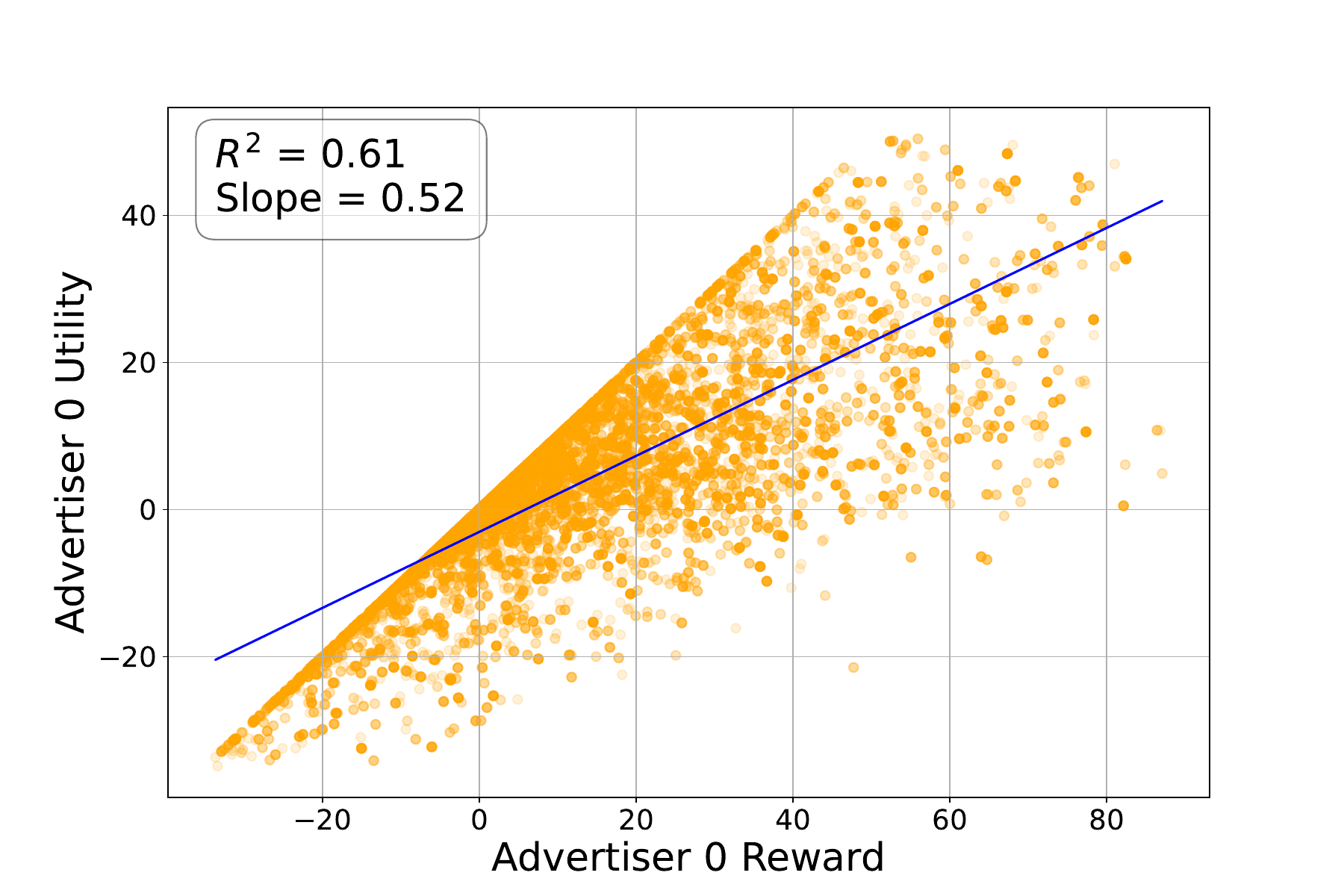}
        \caption{With the payment offset.}
        \label{subfig:scatter_plot_utilities_rewards_1}
    \end{subfigure}
    \hfill %
    \begin{subfigure}[b]{0.45\textwidth} %
        \centering
        \includegraphics[scale=0.2]{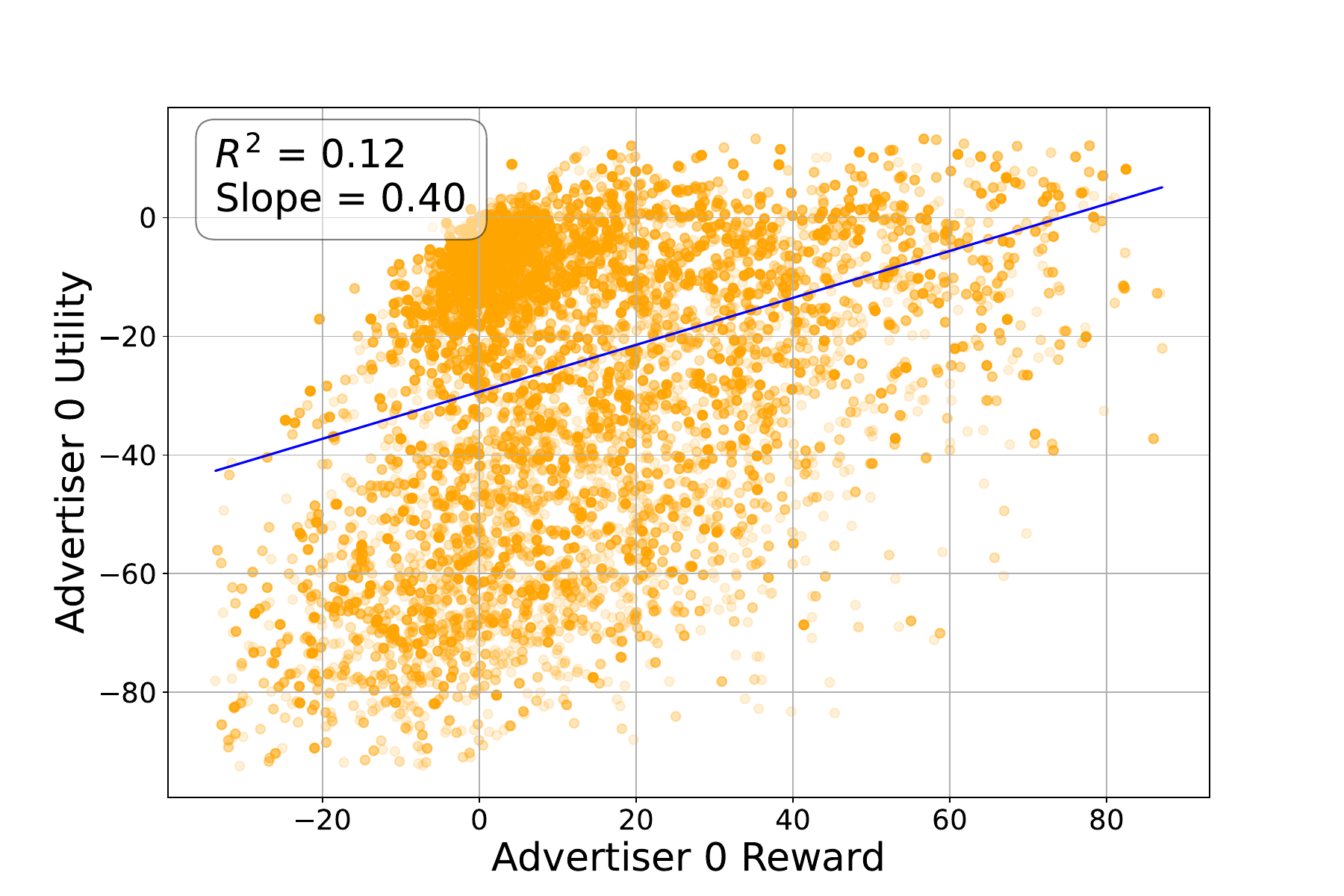}
        \caption{Without the payment offset.}
        \label{subfig:scatter_plot_utilities_rewards_2}
    \end{subfigure}
    \caption{Comparative scatter plots of advertiser reward and utility gain from participation, with and without the payment offset of \Cref{section:bidding zero offset} for candidate sequences generated by the context-aware LLM $\pi_{\text{gem}}$.
    We additionally show a linear regressor fit to that data, its slope and its $R^2$.}
    \label{fig_appendix:scatter_plots}
\end{figure}

\begin{figure}[h!]
    \centering
    \begin{subfigure}[b]{0.45\textwidth} %
        \centering
        \includegraphics[scale=0.2]{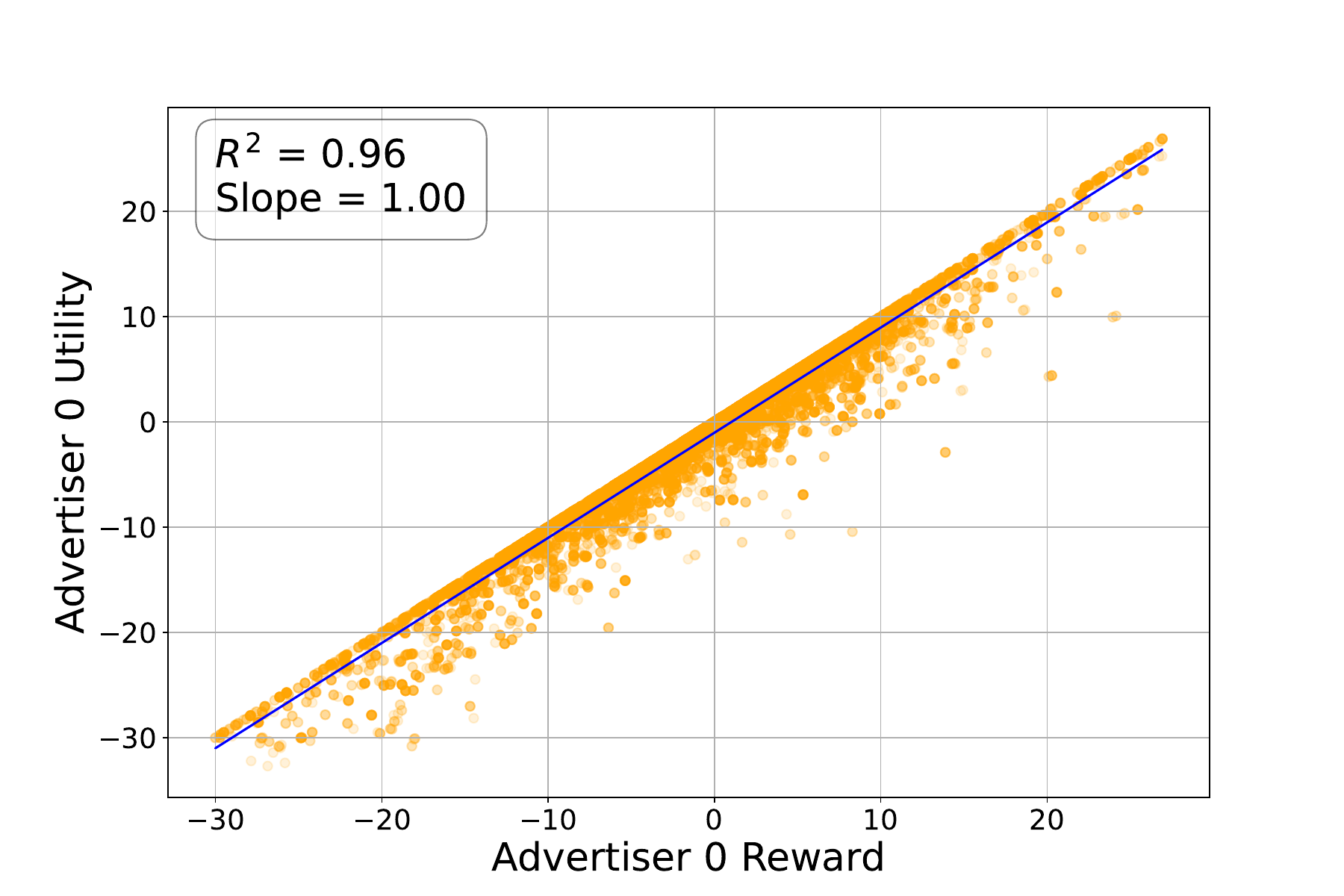}
        \caption{With the payment offset.}
        \label{subfig:scatter_plot_utilities_rewards_3}
    \end{subfigure}
    \hfill %
    \begin{subfigure}[b]{0.45\textwidth} %
        \centering
        \includegraphics[scale=0.2]{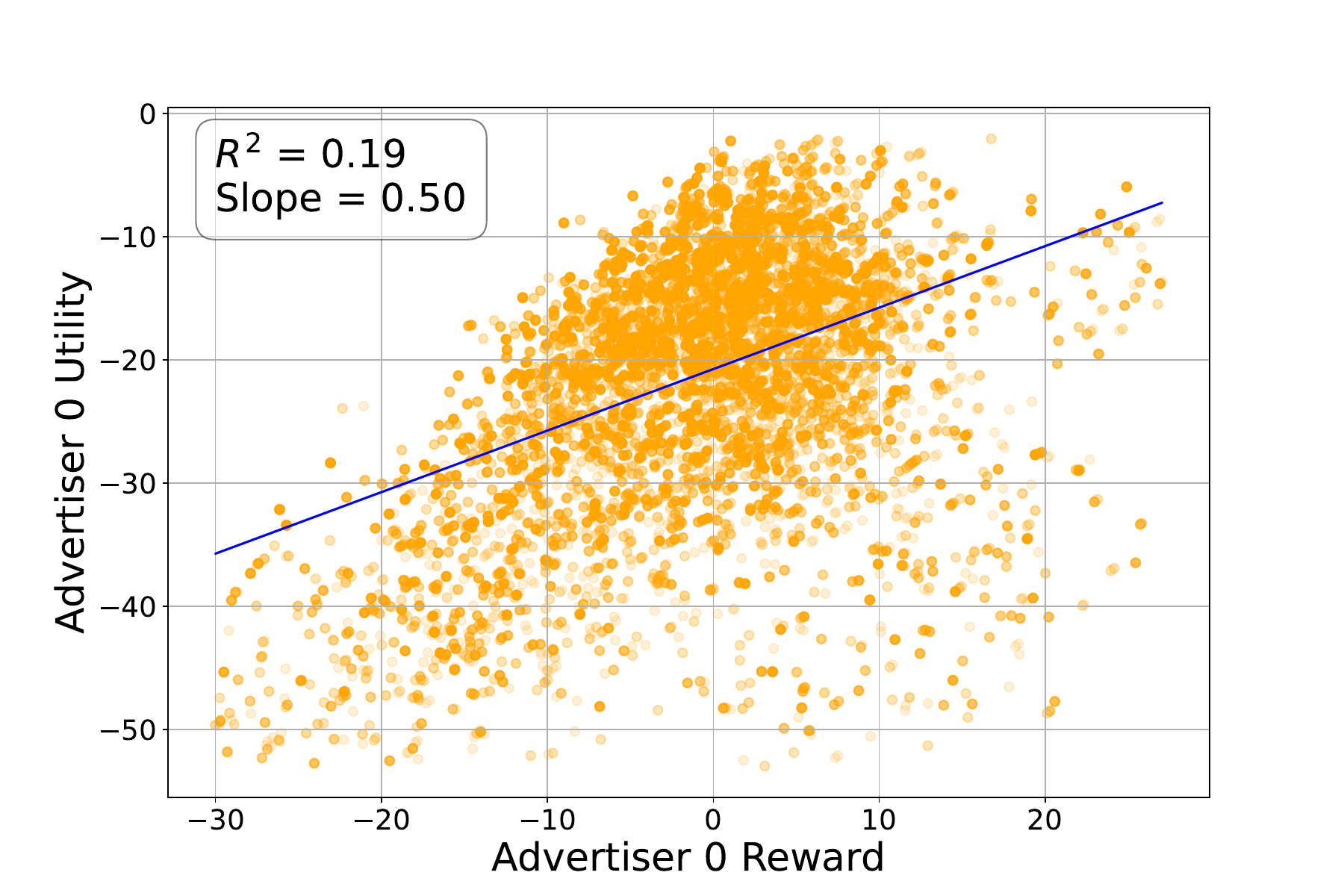}
        \caption{Without the payment offset.}
        \label{subfig:scatter_plot_utilities_rewards_4}
    \end{subfigure}
    \caption{Comparative scatter plots of advertiser reward and utility gain from participation, with and without the payment offset of \Cref{section:bidding zero offset} for candidate sequences generated by the reference $\pi_{\text{\text{ref}}}$.
    We additionally show a linear regressor fit to that data, its slope and its $R^2$.}
    \label{fig_appendix:scatter_plots_without_context}
\end{figure}

\Cref{fig_appendix:scatter_plots,fig_appendix:scatter_plots_without_context} 
explore the effectiveness of the payment offset introduced in \Cref{section:bidding zero offset} in aligning an advertiser's utility with her contribution to the social welfare. 
In \Cref{fig_appendix:scatter_plots} we compare the scatter plots of the advertiser reward and utility gain from participation in the mechanism, with and without the payment offset introduced in \Cref{section:bidding zero offset} for candidate sequences generated using the context-aware LLM $\pi_{\text{\text{ref}}}.$
Additionally, for both subfigures, we show a linear regressor fitted to the data, as well as its slope and coefficient of determination.
Comparing the two subfigures, it is immediately obvious that adding the offset to the payments makes the relationship between advertiser reward and utility gain far more linear. 
This is confirmed by the coefficient of determination of the linear regressors fit to each dataset. 
The coefficient of determination of the linear regressor is far larger when we use the offset. 
Without the payment offset, the coefficient of determination is almost $0$, indicating that, without our payment offset, reward gain is not a predictive measure of an agent's utility. 
Additionally, the slope of the linear regressor is also higher for the scatter plot with the payment offset.

In \Cref{fig_appendix:scatter_plots_without_context} we make the same comparison, but for candidate sequences generated using the reference LLM $\pi_{\text{\text{ref}}}$. 
The results are now even more pronounced. 
In \Cref{subfig:scatter_plot_utilities_rewards_3} we observe the relationship between advertiser utility and reward gain with our payment offset is almost perfectly linear, as suggested by the linear regressor fitted to the data having a slope of $1.00$ and an extremely high coefficient of determination of $0.96$, indicating that it can almost perfectly fit the data. 
Without our payment offset however, in \Cref{subfig:scatter_plot_utilities_rewards_4} we can see that the relationship between the two metrics is again both less linear, and less positively correlated, as the slope of the linear regressor is $0.5$ and its coefficient of determination is only $0.19$. 

\begin{figure}
    \centering
    \resizebox{0.6\textwidth}{!}{
    \includegraphics[trim = 0 0 0 60, clip]{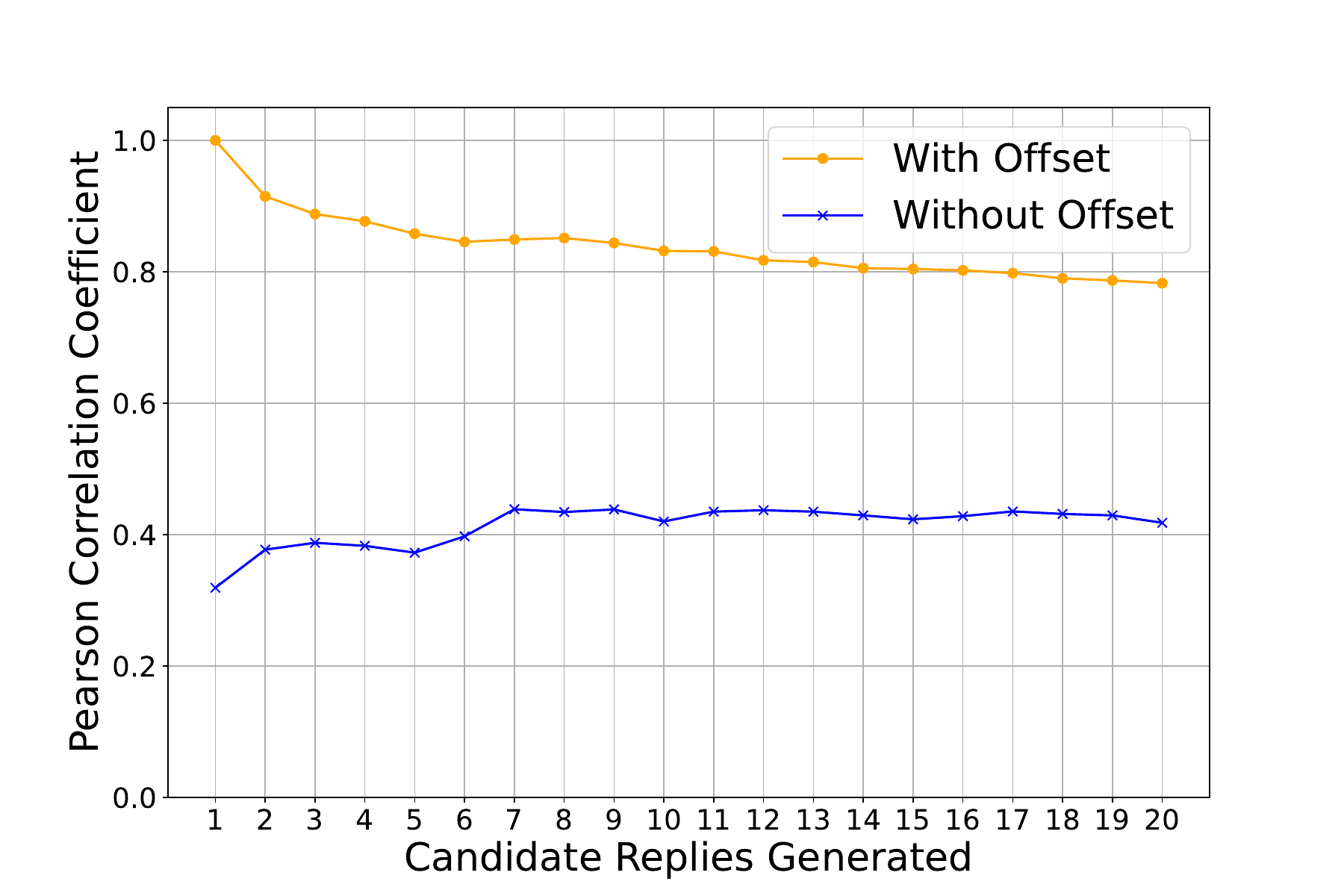}
    }
    \caption{Pearson correlation between advertiser normalized utility and reward.
    }
    \label{fig:pearson_correlation}
\end{figure}

To conclude, in all cases tested, 
the use of the advertiser-specific offset introduced in \Cref{section:bidding zero offset} 
increases an advertiser's expected utility, 
makes the relationship between an advertiser's contribution to social welfare and her utility both more linear and more positively correlated, 
and, for the context-aware version of our mechanism, can ensure positive expected utility for the advertisers.

\newpage
\clearpage

\subsection{Experiments in Settings with Many Advertisers} \label{sec:app:experiments_many_advertisers}
In this section, we test MOSAIC's performance in settings with $n = 5$ and $n = 10$ advertisers, compared to the results of \Cref{section:experiments} where the number of advertisers was set to $2$. 
To maintain the same balance between advertisers and the usefulness to the user as in the experiments of the main paper, we set $\tau = n / 2$.\footnote{Note that in the main experiments of \Cref{section:experiments} we had $n = 2$ and $\tau = 1$}
We use the same $50$ user queries, but now each query has five and ten advertisers, and test each query on $20$ different random seeds, resulting in $1000$ instances. 

In Figures \ref{subfig:log_probability_n5} to \ref{subfig:revenue_n10}, we present the 
log probability of the returned reply with respect to the optimal LLM for the platform's objective $\pi^*_r$ and the reference LLM $\pi_{\text{ref}}$, the total advertiser value and  reward gain from participating in MOSAIC and the platform's revenue. 
All plots are with respect to the number of candidate replies (i.e., LLM queries) that MOSAIC used. 
We compare MOSAIC's baseline and context-aware versions.

First, in \Cref{subfig:log_probability_n5,subfig:log_probability_n10} we observe that the context-aware version of MOSAIC quickly converges to the platform's objective. 
By comparison, MOSAIC's naive version fails to do so.
It is noteworthy that for very low numbers of candidate replies, the naive version of the mechanism performs better than the context-aware version. 
That is because in these experiments, more weight is placed towards the reference LLM, as we have set a higher $\tau$ value. 
However, the superior convergence rate of the context-aware version of the mechanism recovers that difference very quickly.
At the same time, \Cref{subfig:log_probability_ref_n5,subfig:log_probability_ref_n10} show that the context-aware version of the mechanism produces sequences that closely match the distribution of the reference LLM, thus maintaining the usefulness to the user. 

Finally, the context-aware version of MOSAIC is able to generate significant value (\Cref{subfig:value_gain_n5,subfig:value_gain_n10}) and utility (\Cref{subfig:utility_gain_n5,subfig:utility_gain_n10}) for the advertisers, while also recapturing a significant portion as revenue for the platform (\Cref{subfig:revenue_n5,subfig:revenue_n10}).

Taken together, similar to our results in the main paper, we have shown that even with a large number of advertisers, MOSAIC quickly converges to the theoretically optimal distribution, generating significant value and utility for the advertisers and revenue for the platform, while also maintaining its  usefulness to the user.

\begin{figure}
    \centering
    \vspace{-35pt}
    \begin{subfigure}[b]{0.45\textwidth}
        \centering
        \includegraphics[scale=0.2]{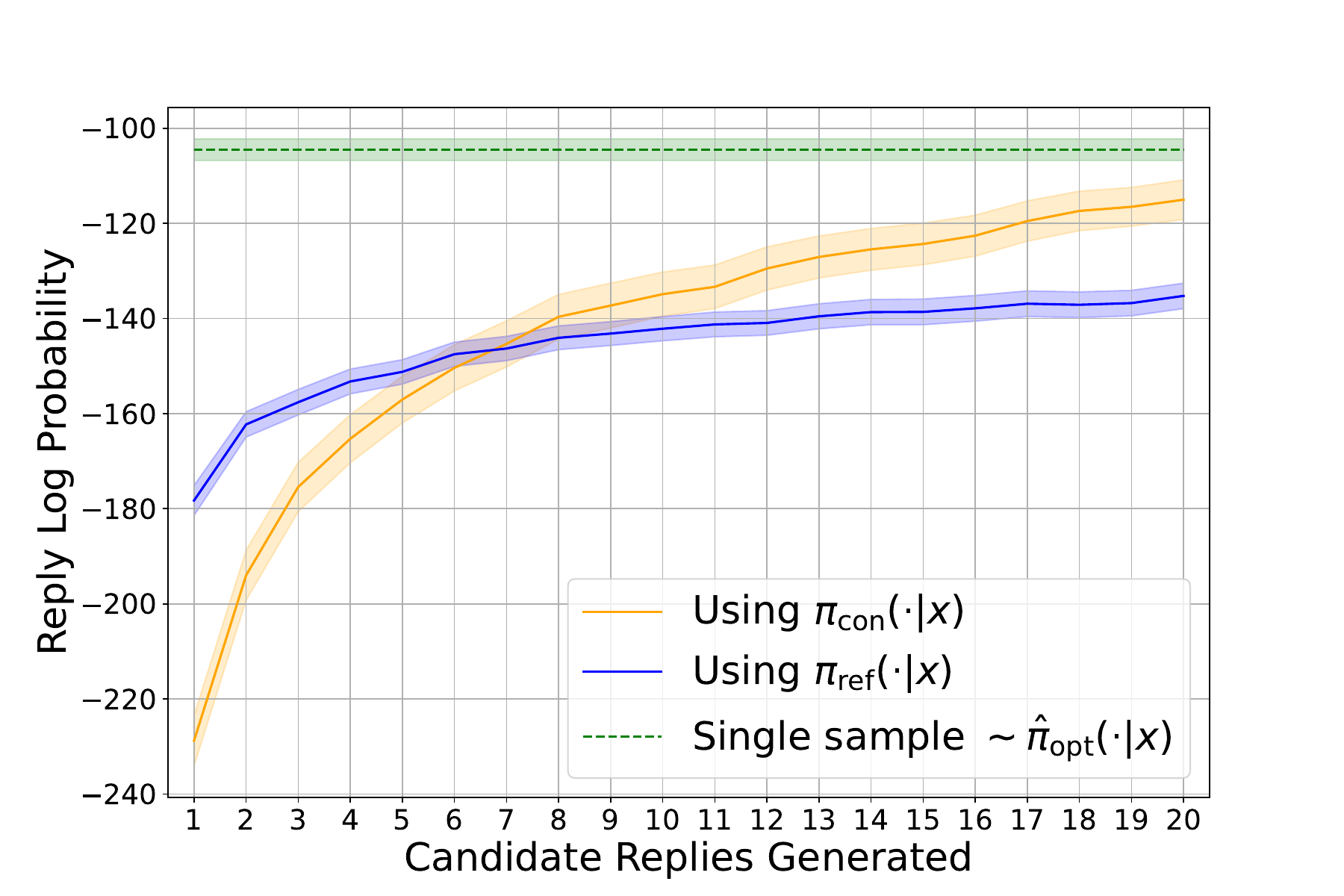}
        \caption{Log probability w.r.t. $\pi^*_{r}$ for 5 advertisers.}
        \label{subfig:log_probability_n5}
    \end{subfigure}
    \hfill
    \begin{subfigure}[b]{0.45\textwidth}
        \centering
        \includegraphics[scale=0.2]{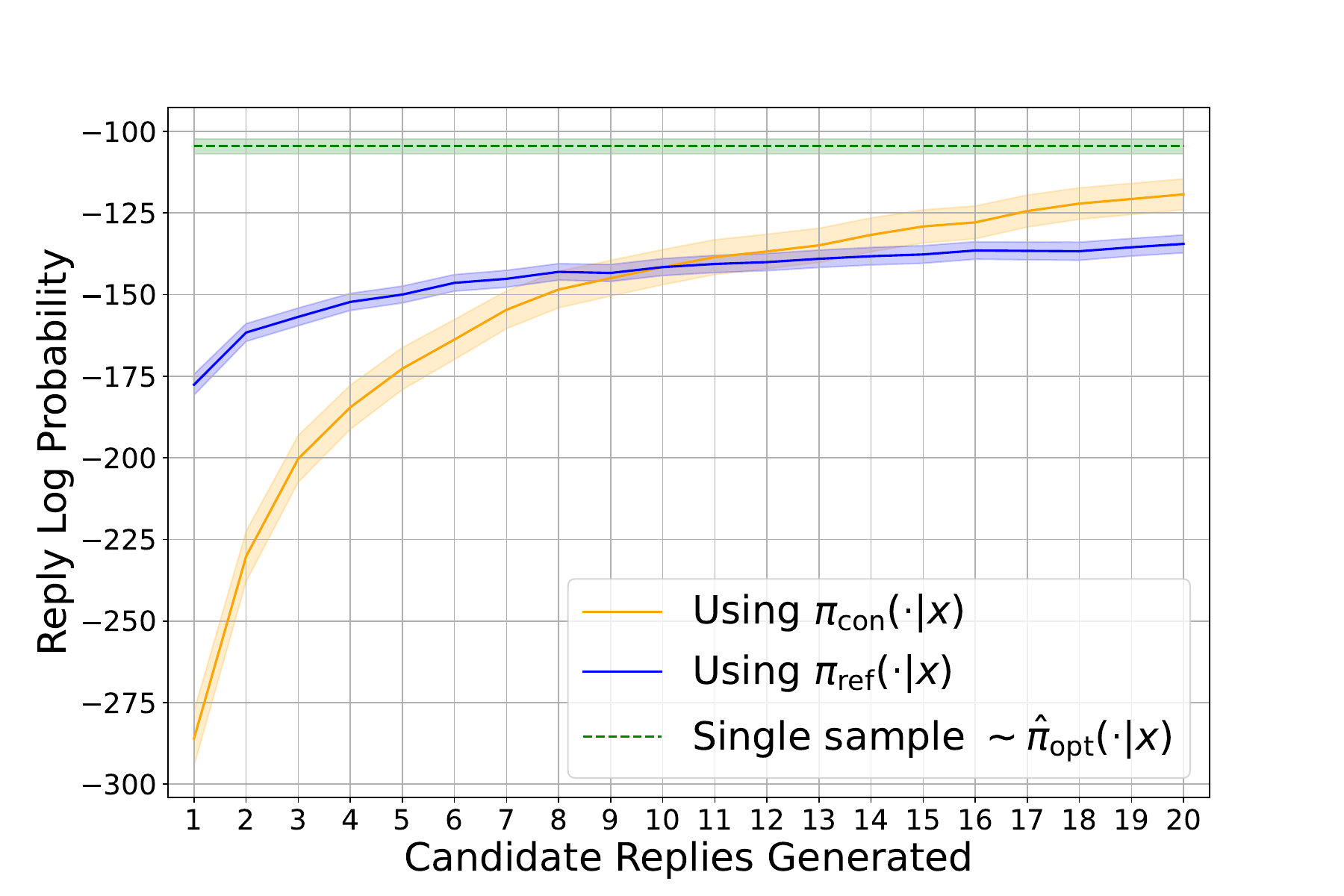}
        \caption{Log probability w.r.t. $\pi^*_{r}$ for 10 advertisers.}
        \label{subfig:log_probability_n10}
    \end{subfigure}

    \begin{subfigure}[b]{0.45\textwidth}
        \centering
        \includegraphics[scale=0.2]{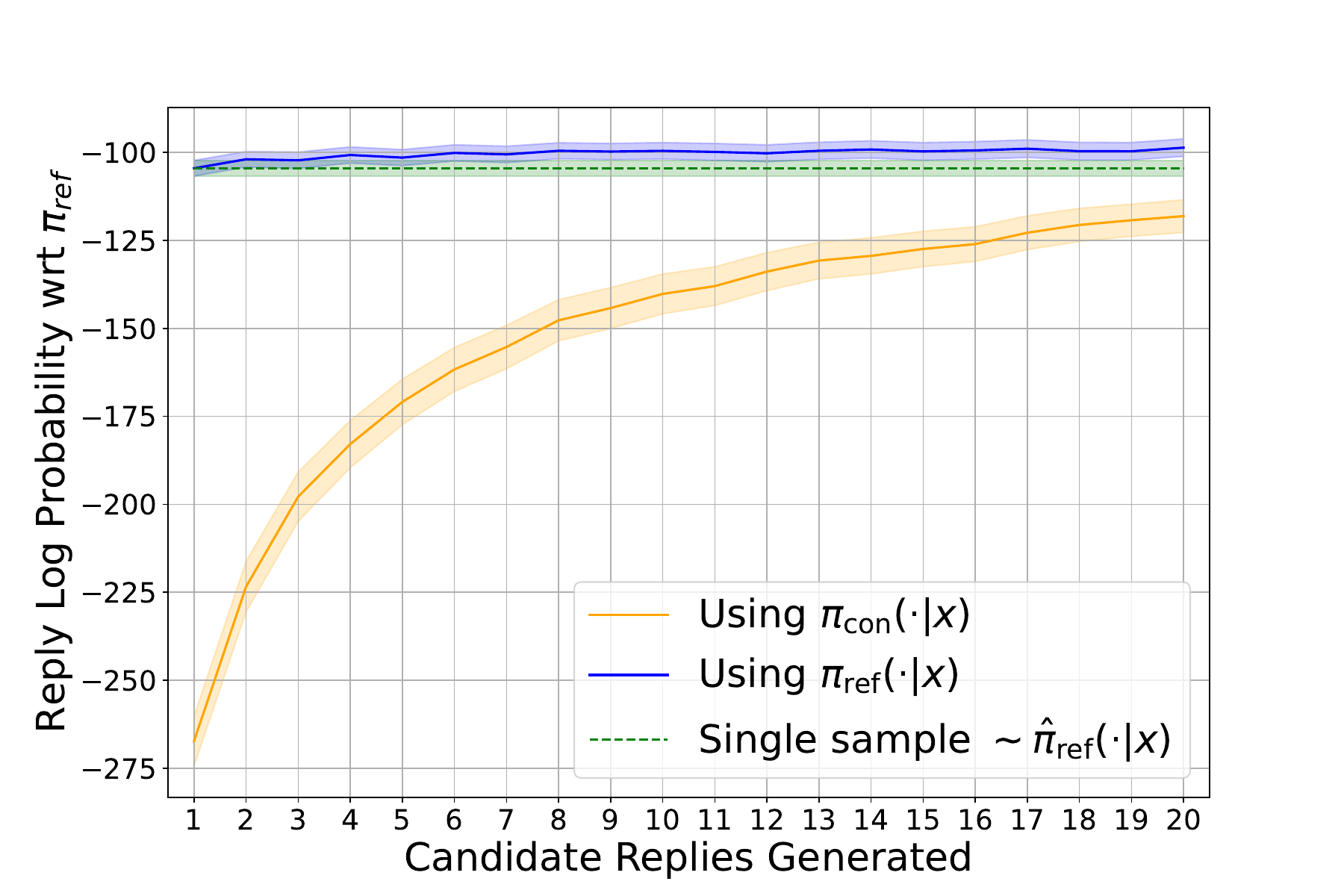}
        \caption{Log probability w.r.t. $\pi_{\text{ref}}$ for 5 advertisers.}
        \label{subfig:log_probability_ref_n5}
    \end{subfigure}
    \hfill
    \begin{subfigure}[b]{0.45\textwidth}
        \centering
        \includegraphics[scale=0.2]{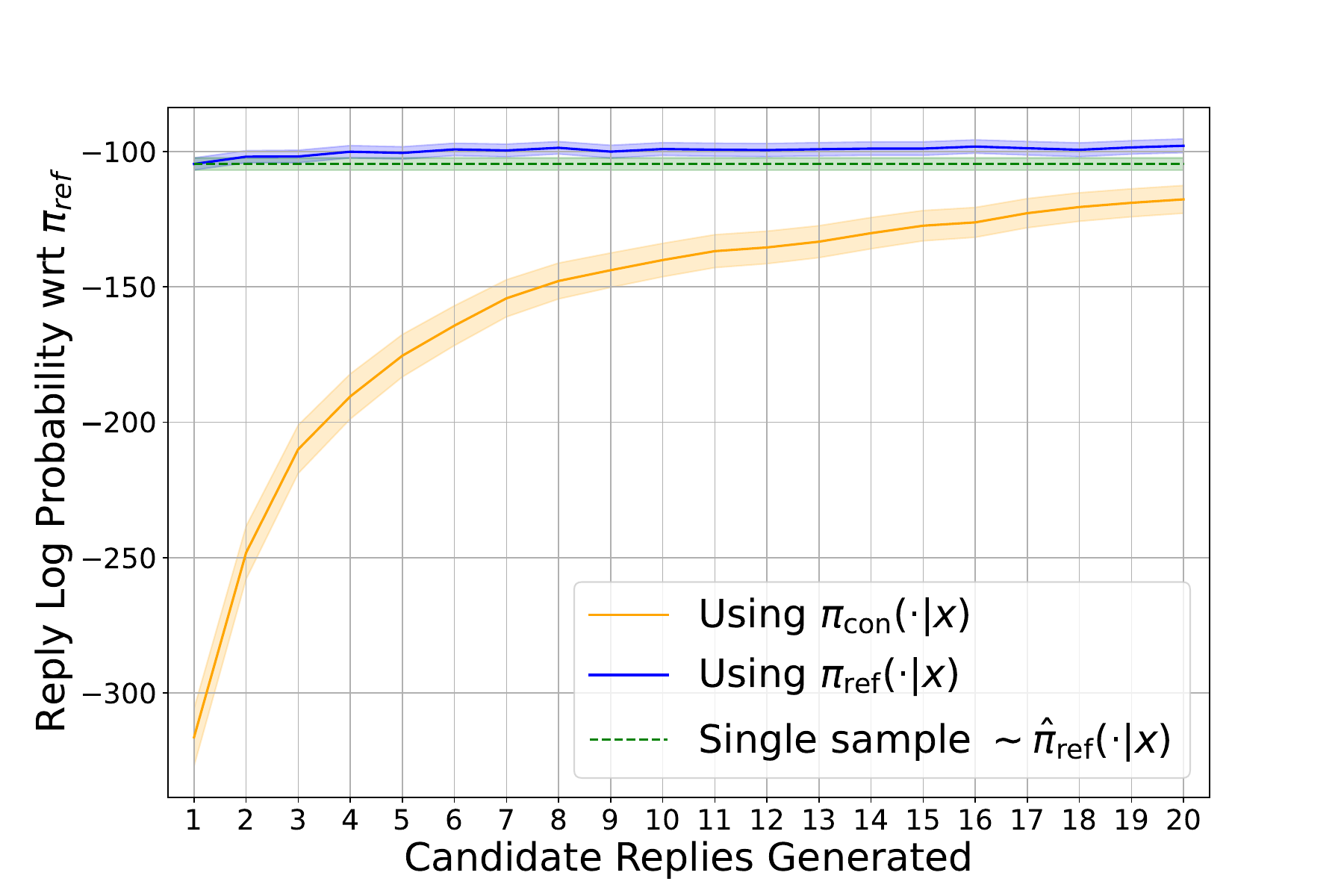}
        \caption{Log probability w.r.t. $\pi_{\text{ref}}$ for 10 advertisers.}
        \label{subfig:log_probability_ref_n10}
    \end{subfigure}
    
    \begin{subfigure}[b]{0.45\textwidth}
        \centering
        \includegraphics[scale=0.2]{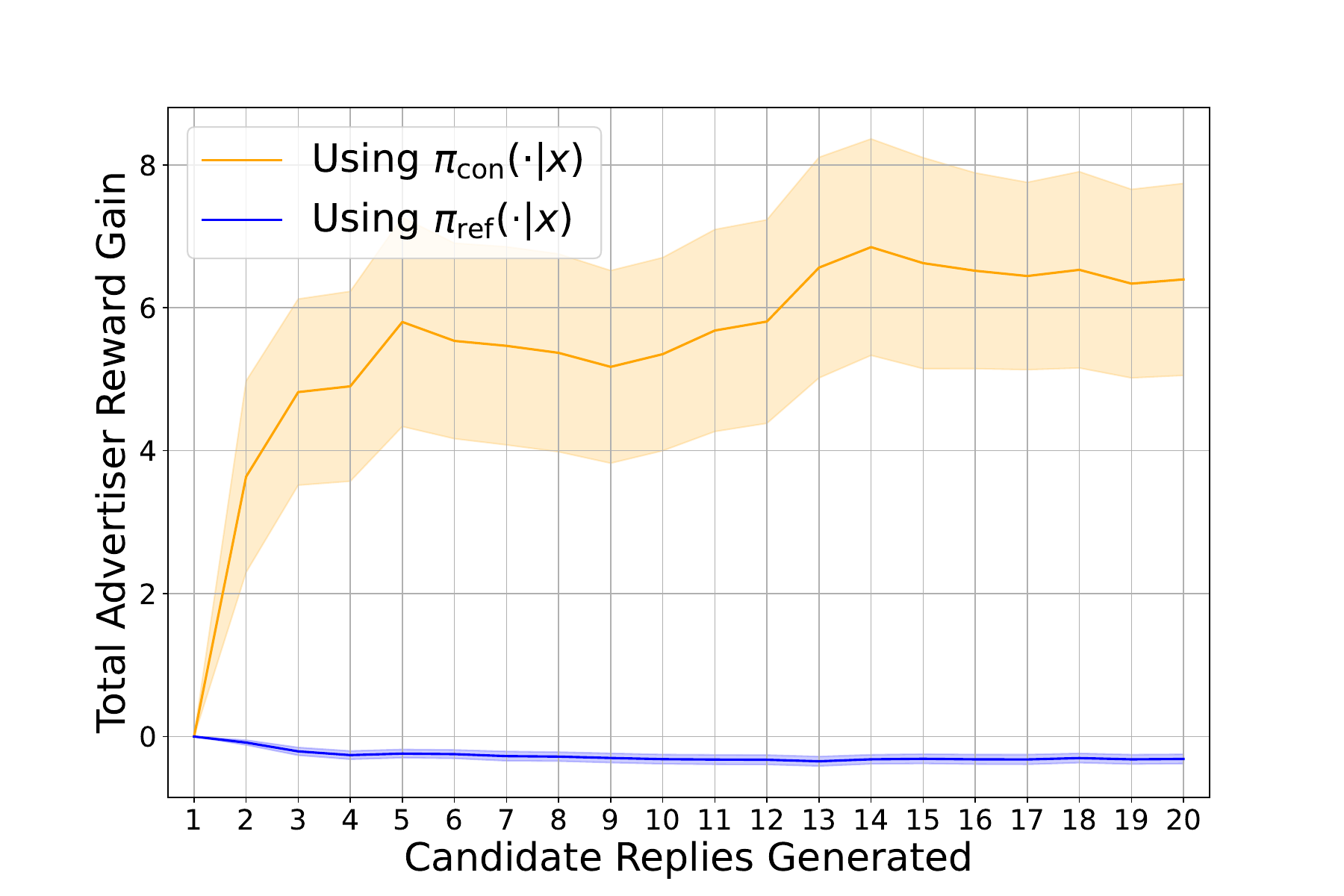}
        \caption{Advertiser value gain for 5 advertisers.}
        \label{subfig:value_gain_n5}
    \end{subfigure}
    \hfill
    \begin{subfigure}[b]{0.45\textwidth}
        \centering
        \includegraphics[scale=0.2]{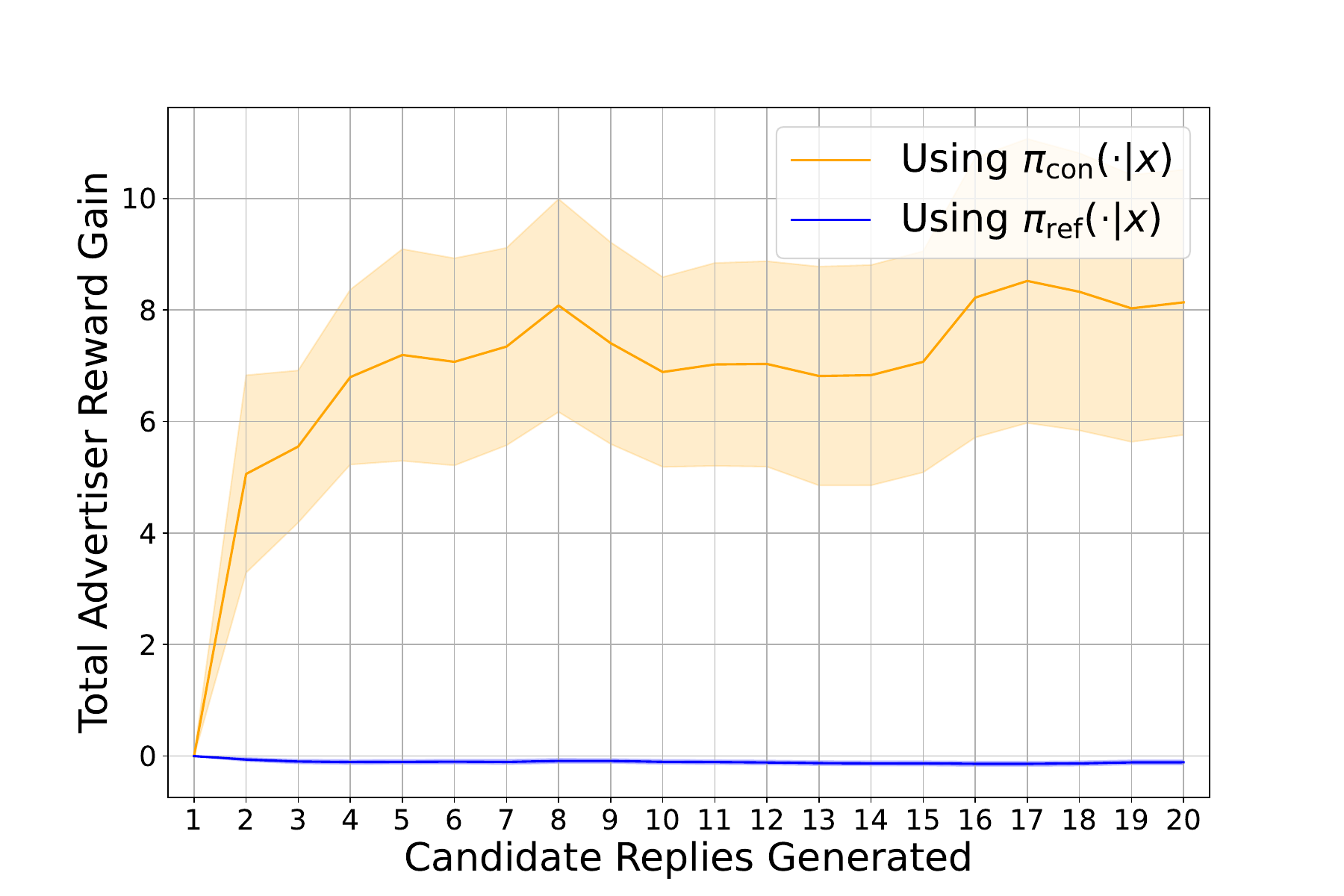}
        \caption{Advertiser value gain for 10 advertisers.}
        \label{subfig:value_gain_n10}
    \end{subfigure}

    \begin{subfigure}[b]{0.45\textwidth}
        \centering
        \includegraphics[scale=0.2]{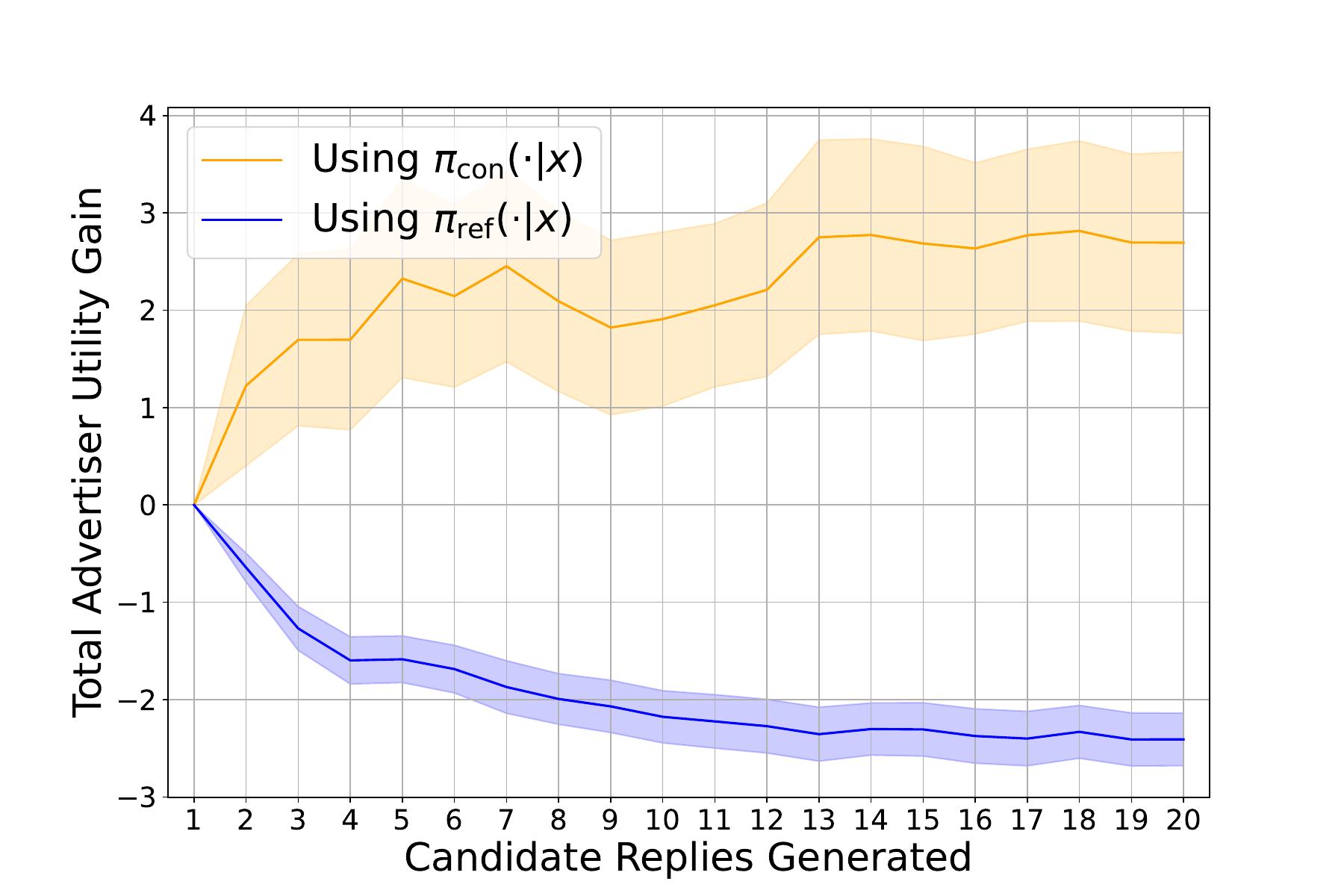}
        \caption{Advertiser utility gain for 5 advertisers.}
        \label{subfig:utility_gain_n5}
    \end{subfigure}
    \hfill
    \begin{subfigure}[b]{0.45\textwidth}
        \centering
        \includegraphics[scale=0.2]{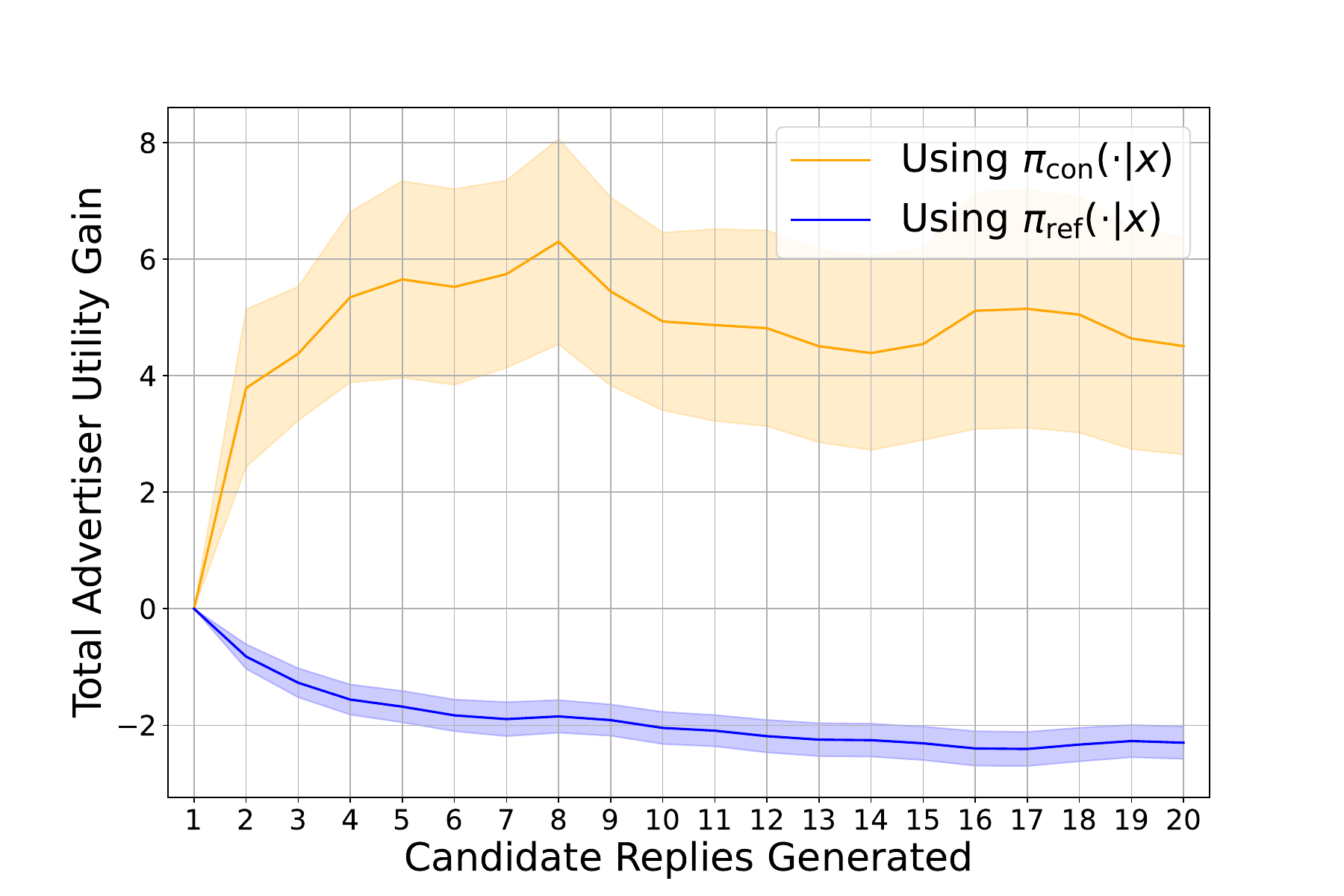}
        \caption{Advertiser utility gain for 10 advertisers.}
        \label{subfig:utility_gain_n10}
    \end{subfigure}
    
    \begin{subfigure}[b]{0.45\textwidth}
        \centering
        \includegraphics[scale=0.2]{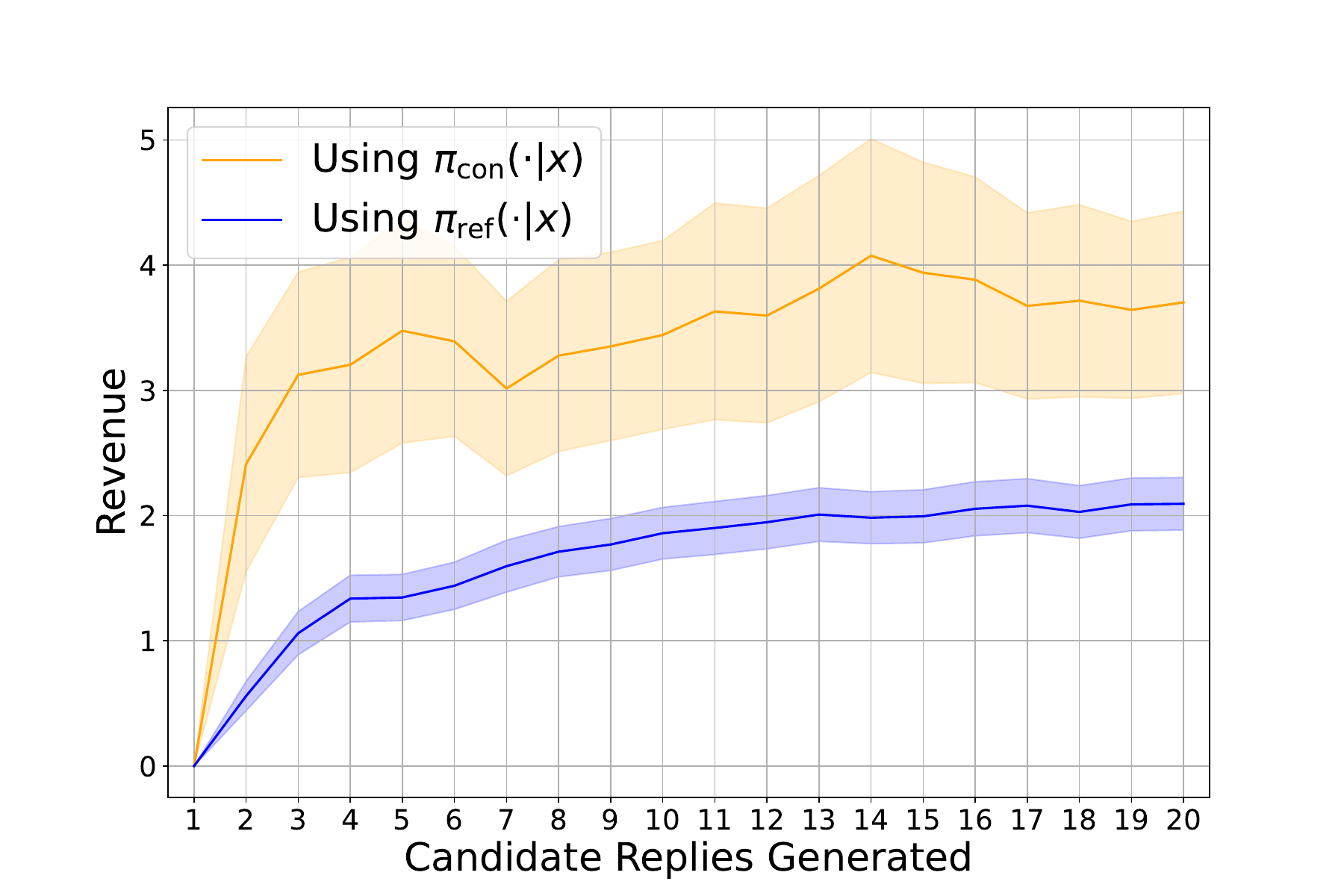}
        \caption{Revenue for 5 advertisers.}
        \label{subfig:revenue_n5}
    \end{subfigure}
    \hfill
    \begin{subfigure}[b]{0.45\textwidth}
        \centering
        \includegraphics[scale=0.2]{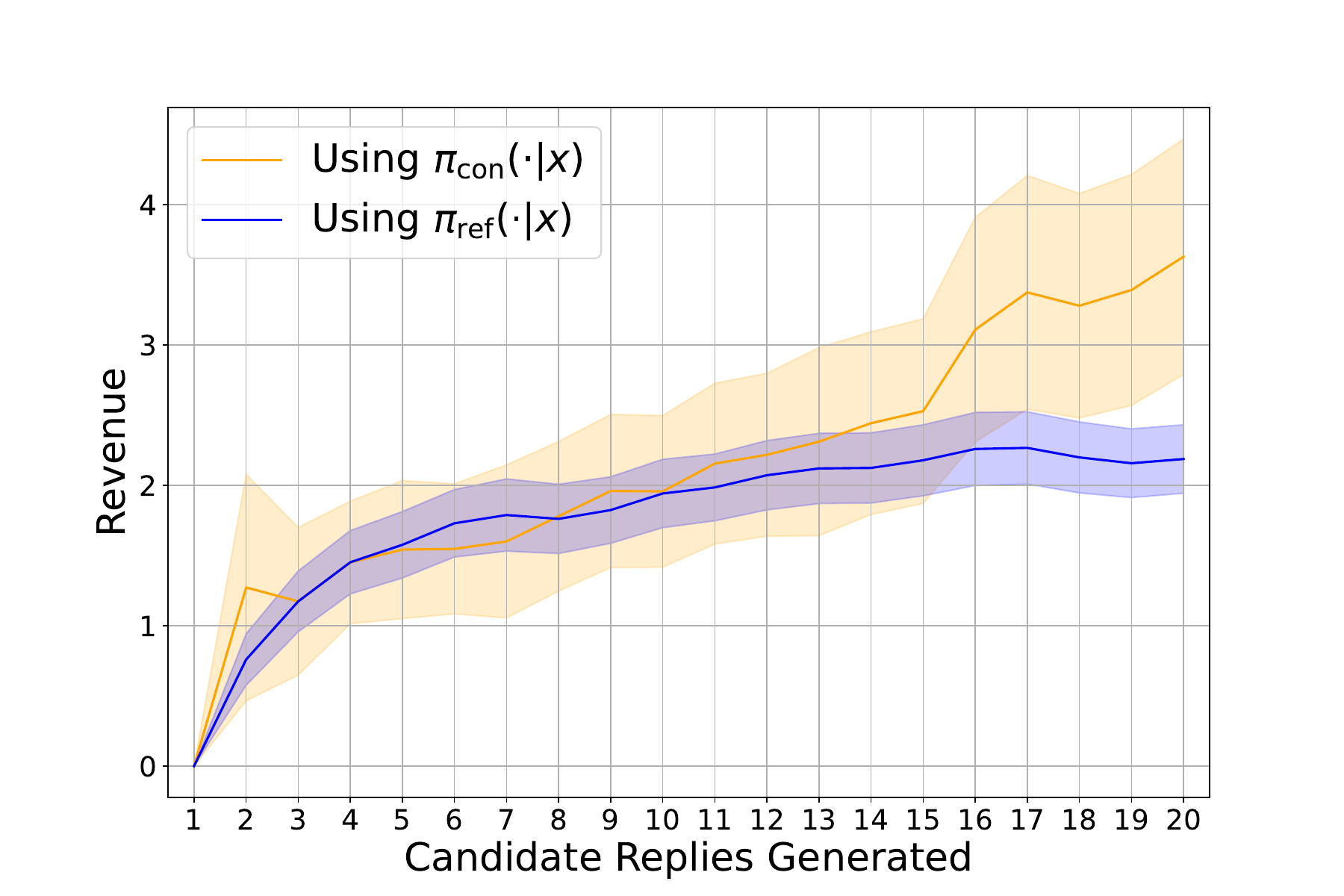}
        \caption{Revenue for 10 advertisers.}
        \label{subfig:revenue_n10}
    \end{subfigure}
    
    \caption{Comparison of results for 5 advertisers (left) and 10 advertisers (right). Each row corresponds to a specific metric: log probability with respect to the optimal LLM $\pi^*_{r}$, the reference LLM $\pi_{\text{ref}}$, advertiser value and utility gain, and revenue.
    Shown are averages over $1000$ instances including $95$\% CIs.
    }
    \vspace{-20pt}
    \label{fig:comparison_5_10_advertisers}
\end{figure}

\newpage
\clearpage

\subsection{Compute Experiments} \label{sec:app:compute_experiments}

In this section, we experimentally evaluate MOSAIC's computational requirements. Using the setup described in \Cref{sec:experiment_setup}, we measure the total wall time required by MOSAIC on a \emph{single} A100 GPU, focusing on the number of candidate replies generated for evaluation. This total time includes:
\begin{enumerate}
    \item Generating candidate replies using the context-aware LLM.
    \item Evaluating these replies with the advertisers' LLMs implemented as discussed in \Cref{sec:experiment_setup}.
    \item Calculating the Rochet payments, as detailed in \Cref{section:payments}. In practice the time to calculate payments is negligible (under $10$ milliseconds), and need not increase user-perceived latency because they can be calculated after showing outputs to the user.
\end{enumerate}

\Cref{fig:total_time} presents the total time taken by MOSAIC, which includes both the generation and evaluation of candidate replies. These components are further broken down in \Cref{fig:generation_time,fig:evaluation_time}, respectively.

As previously demonstrated in \Cref{section:experiments} and \Cref{sec:app:experiments_many_advertisers}, MOSAIC converges to the optimal distribution using only 20 candidate replies, regardless of the number of advertisers. In \Cref{fig:generation_time}, we observe that generating 10 and 20 candidate replies on a single A100 GPU takes $30.5$ and $60.5$ seconds, respectively. This represents a 2.5- and 5-fold increase over the $12.0$ seconds required to generate a single reply from the same LLM, i.e., the latency perceived by the user when directly querying the LLM. 

To reduce user-perceived latency, MOSAIC can parallelize candidate generation across multiple GPUs. For example, distributing the generation of 20 replies across two GPUs (each generating 10 replies) reduces the total time to $30.5$ seconds, assuming efficient evaluation techniques as outlined in \Cref{section:practical_considerations}. This approach results in MOSAIC requiring approximately $5$ times the compute time of generating a single reply to converge, but with a perceived latency to the user of only $2.5$ times that of directly querying an LLM (of the same architecture and similar size).
Although our experiments were run on a single GPU, note that MOSAIC can always be parallelized further so that the user-perceived latency is the same as directly querying a \emph{single} LLM, as explained in \Cref{section:practical_considerations}.

\begin{figure}[h!]
    \centering
    \includegraphics[width=0.7\linewidth]{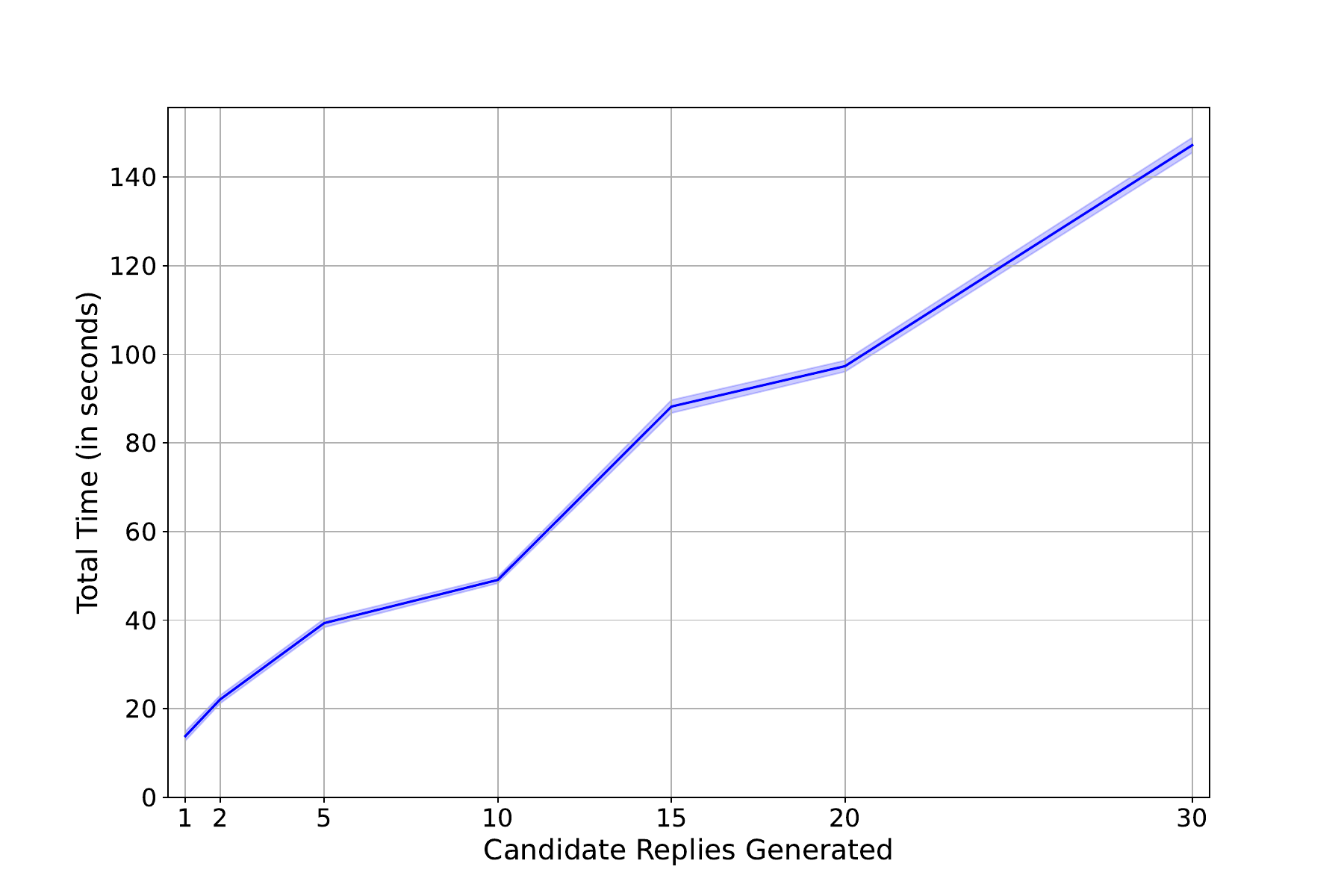}
    \caption{MOSAIC total wall time (seconds) as a function of the number of candidate replies generated using $\pi_{\text{con}}$. Shown are averages over $50$ instances including $95$ \% CIs.
    }
    \label{fig:total_time}
\end{figure}

\begin{figure}[h!]
    \centering
    \includegraphics[width=0.7\linewidth]{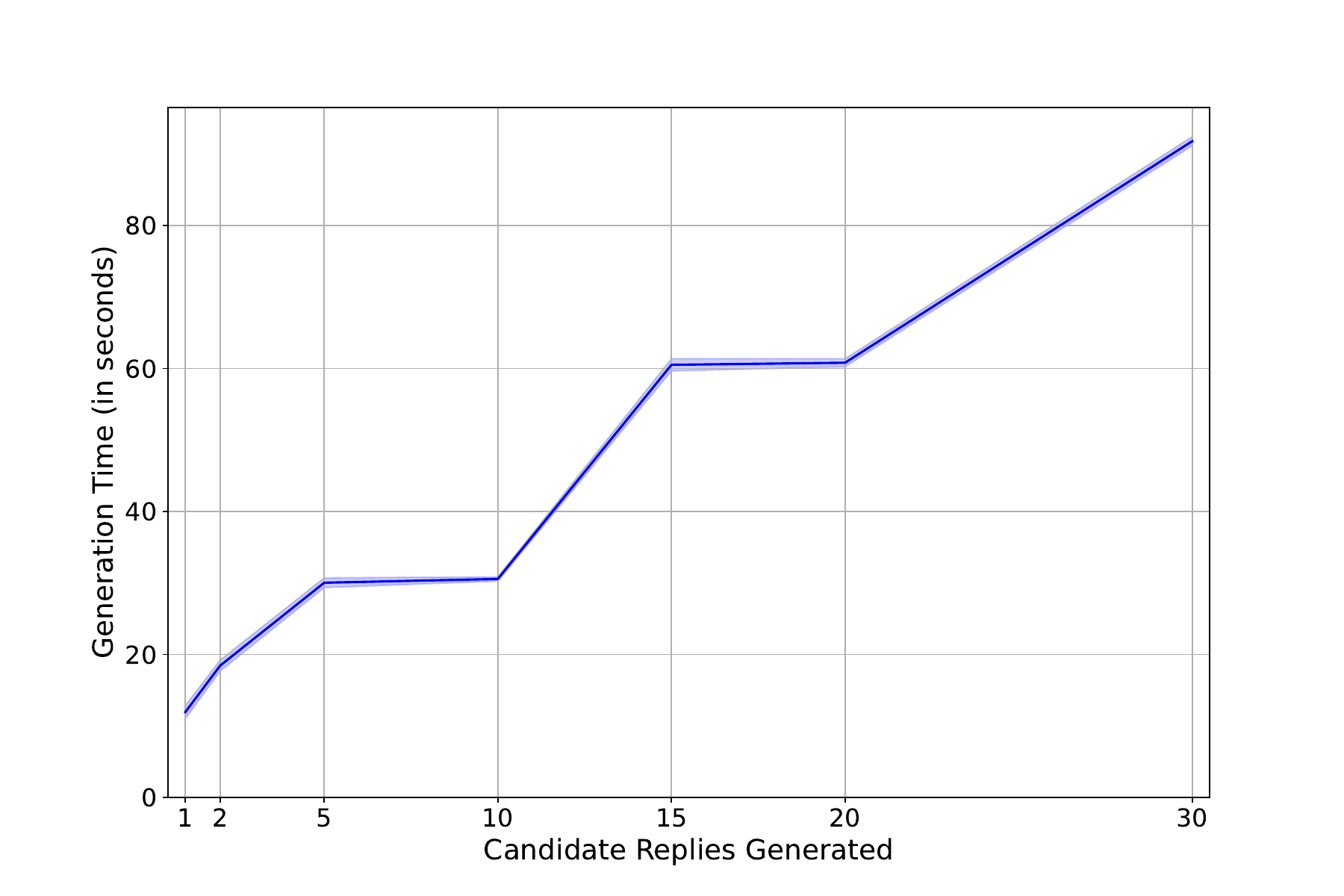}
    \caption{MOSAIC generation time as a function of the number of candidate replies generated using $\pi_{\text{con}}$. Shown are averages over $50$ instances including $95$ \% CIs.
    }
    \label{fig:generation_time}
\end{figure}

\begin{figure}[h!]
    \centering
    \includegraphics[width=0.7\linewidth]{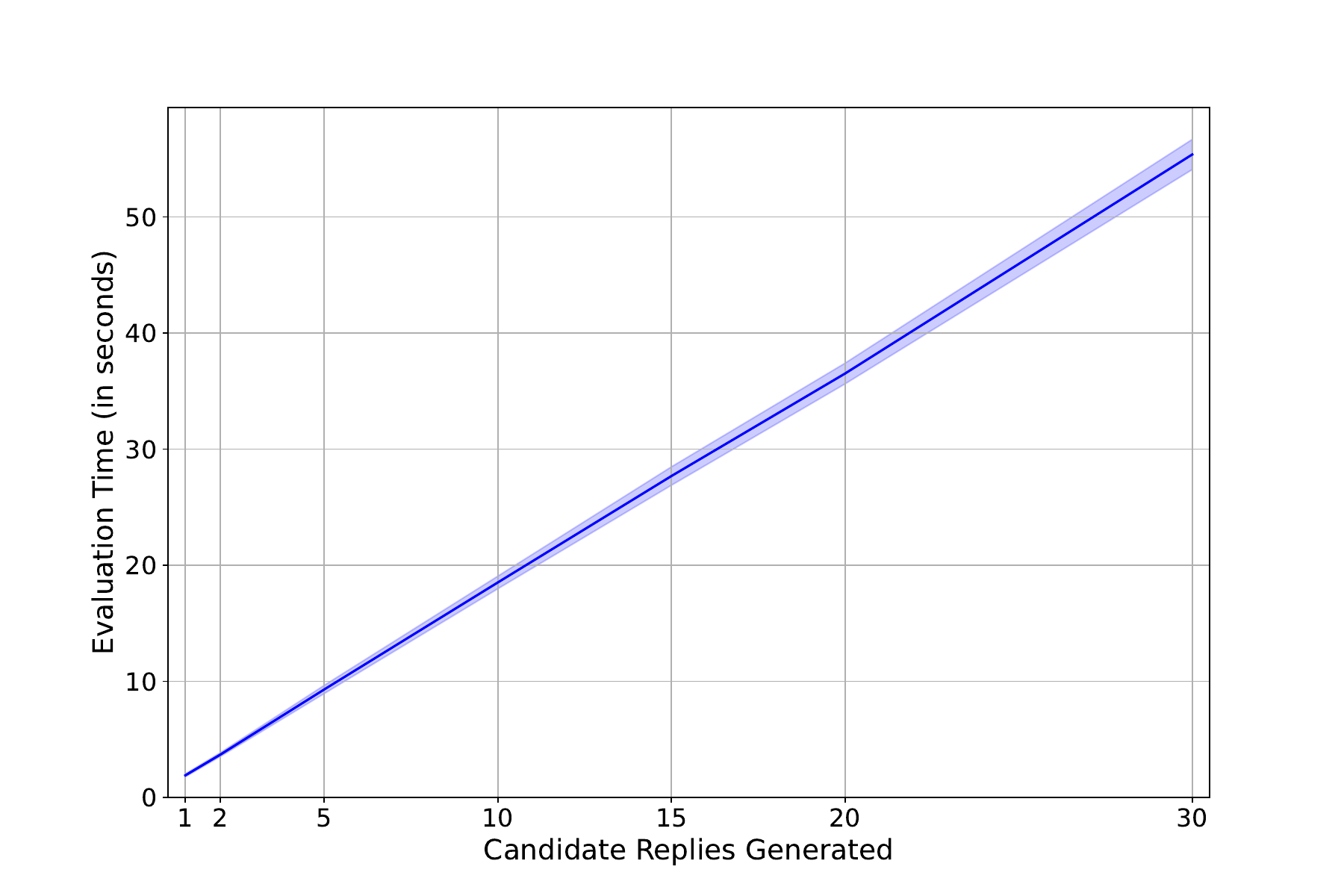}
    \caption{MOSAIC evaluation time as a function of the number of candidate replies. Shown are averages over $50$ instances including $95$ \% CIs.
    }
    \label{fig:evaluation_time}
\end{figure}

\newpage
\clearpage

\subsection{Conflicts between advertisers}
The output of MOSAIC is a single text response to the user.
In trying to reflect the interests of multiple advertisers, there is a risk of a single response being incoherent, just mentioning as many advertisers as possible without being a good promotion for any of them.
However, the platform's objective, which takes into account both advertiser utility and usefulness for the user, should account for this: if advertisers are happy to be cross-promoted, those outcomes should be more likely, but incoherent advertisements or those that simultaneously promote competitors should be less likely.

In \Cref{fig:advertisers_mentioned} we show the fraction of advertisers mentioned by the context-aware version of MOSAIC as a function of the candidate number of replies used. 
We observe that initially, the fraction of advertisers mentioned is very high, but drops significantly once the mechanism has converged. 
The reason is precisely the fact that mentioning more advertisers in this case would cause 
either a significant drop in advertiser utility, or a drop in the usefulness of the reply to the user. 
By placing more weight towards the advertisers (i.e., decreasing the value of $\tau$ in \Cref{eq:objective}), the platform can increase the number of advertisers shown when the mechanism has converged. 

\begin{figure}[h!]
    \centering
    \resizebox{0.6\textwidth}{!}{
    \includegraphics[trim = 0 0 0 60, clip]{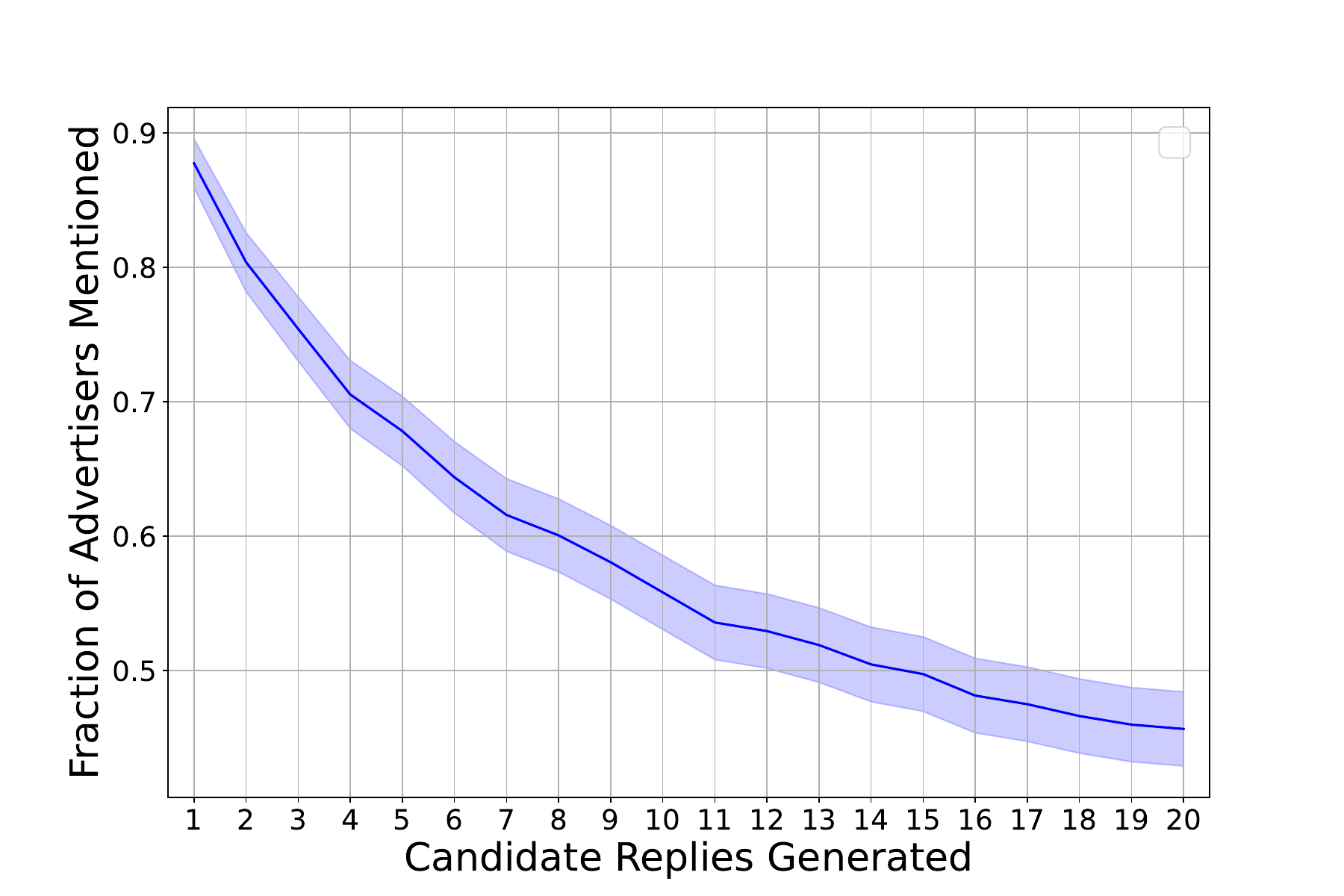}
    }
    \caption{Fraction of advertisers mentioned in the output reply as a function of the number of candidate replies. 
    }
    \label{fig:advertisers_mentioned}
\end{figure}

\newpage
\clearpage

\subsection{Testing with additional language models} \label{subsec:testing_with_other_models}
In this section, we test MOSAIC's performance using Google's flan-t5-large model \citep{flan}. 

We repeat exactly the same procedure as detailed in \Cref{sec:experiment_setup}, but using flan-t5-large as the basis for both $\pi_{\text{gen}}$, the advertisers' LLMs, and the reference LLM. 
We use the same 50 user queries, each with two advertisers, and test each query on $10$ different random seeds, resulting in $500$ instances. 

In \Cref{fig:advertiser_participating-value_gain_llm_t5} to \Cref{fig:log_probability_reference_llm_t5}, we present the total advertiser value gain, reward gain from participating in MOSAIC, platform revenue, and the usefulness of the returned reply to users as functions of the number of candidate replies generated.

Consistent with the main paper’s results using Llama-2-7b-chat-hf \citep{touvron2023llama}, the context-aware version of MOSAIC demonstrates significant value and utility generation for advertisers and successfully captures a substantial portion as platform revenue. In contrast, though, the naive version of MOSAIC fails to generate value or utility for advertisers, resulting in negligible revenue for the platform. 
Similar to the results of the main paper, despite both mechanisms converging in the limit to the platform's objective, as shown in \Cref{Corollary:mechanism_follows_optimal_distribution}, their performance for a computationally tractable number of generated replies diverges. 

Similar to the findings in the main paper, the naive version retains the reference LLM’s usefulness to the user (which is optimal for the task), while the context-aware version prioritizes the platform’s objective, leading to lower user usefulness. Overall, our results using flan-t5-large agree closely with those obtained using Llama-2-7b-chat-hf, demonstrating consistency across models. The main difference is that the context-aware mechanism is even more important when using flan-t5-large.

\begin{figure}[h!]
    \centering
    \includegraphics[width=0.65\linewidth]{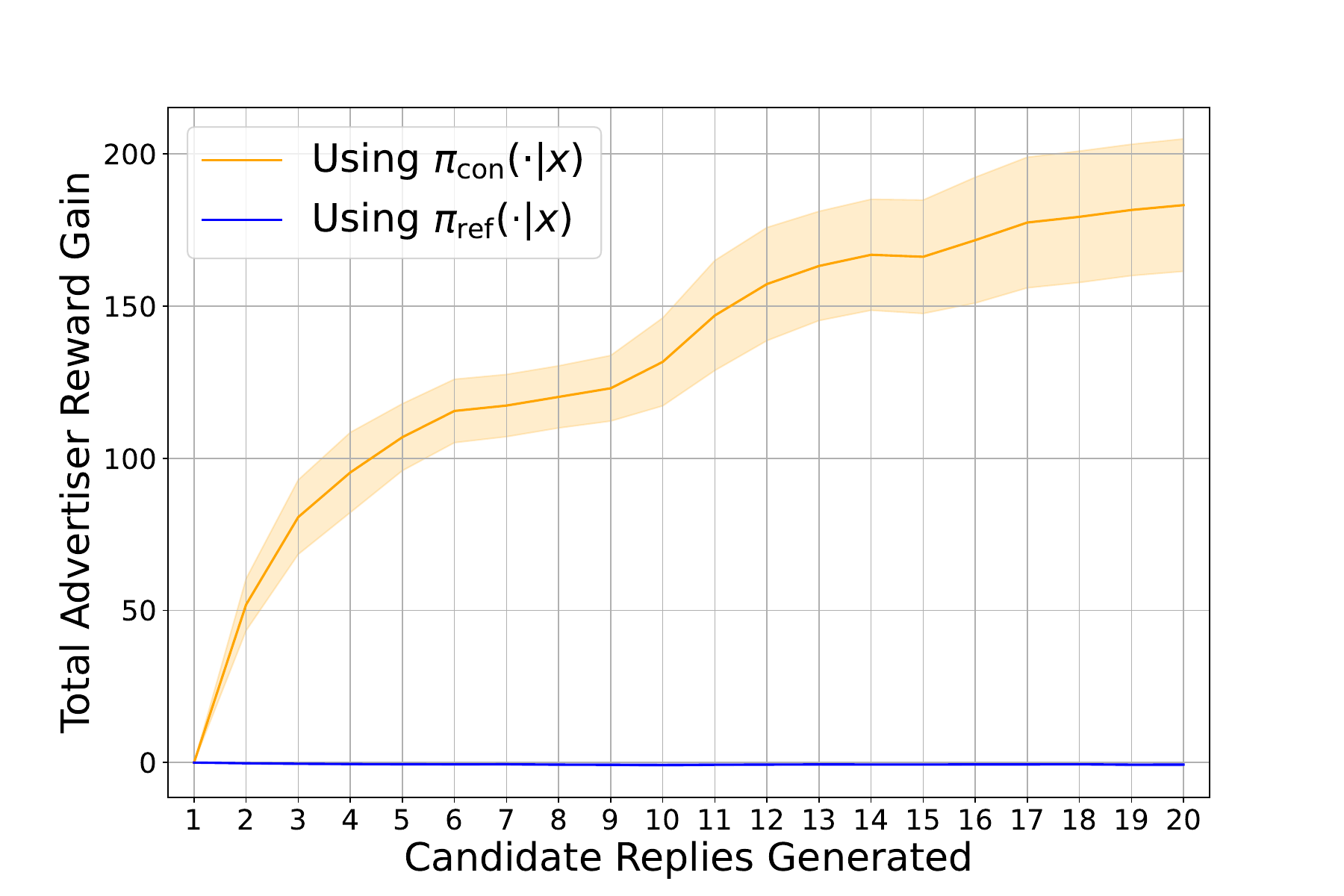}
    \caption{Advertiser value gain as a function of the number of candidate replies generated using $\pi_{\text{ref}}$ and $\pi_{\text{con}}$ (flan-t5-large model).
    }
    \label{fig:advertiser_participating-value_gain_llm_t5}
\end{figure}

\begin{figure}[h!]
    \centering
    \includegraphics[width=0.65\linewidth]{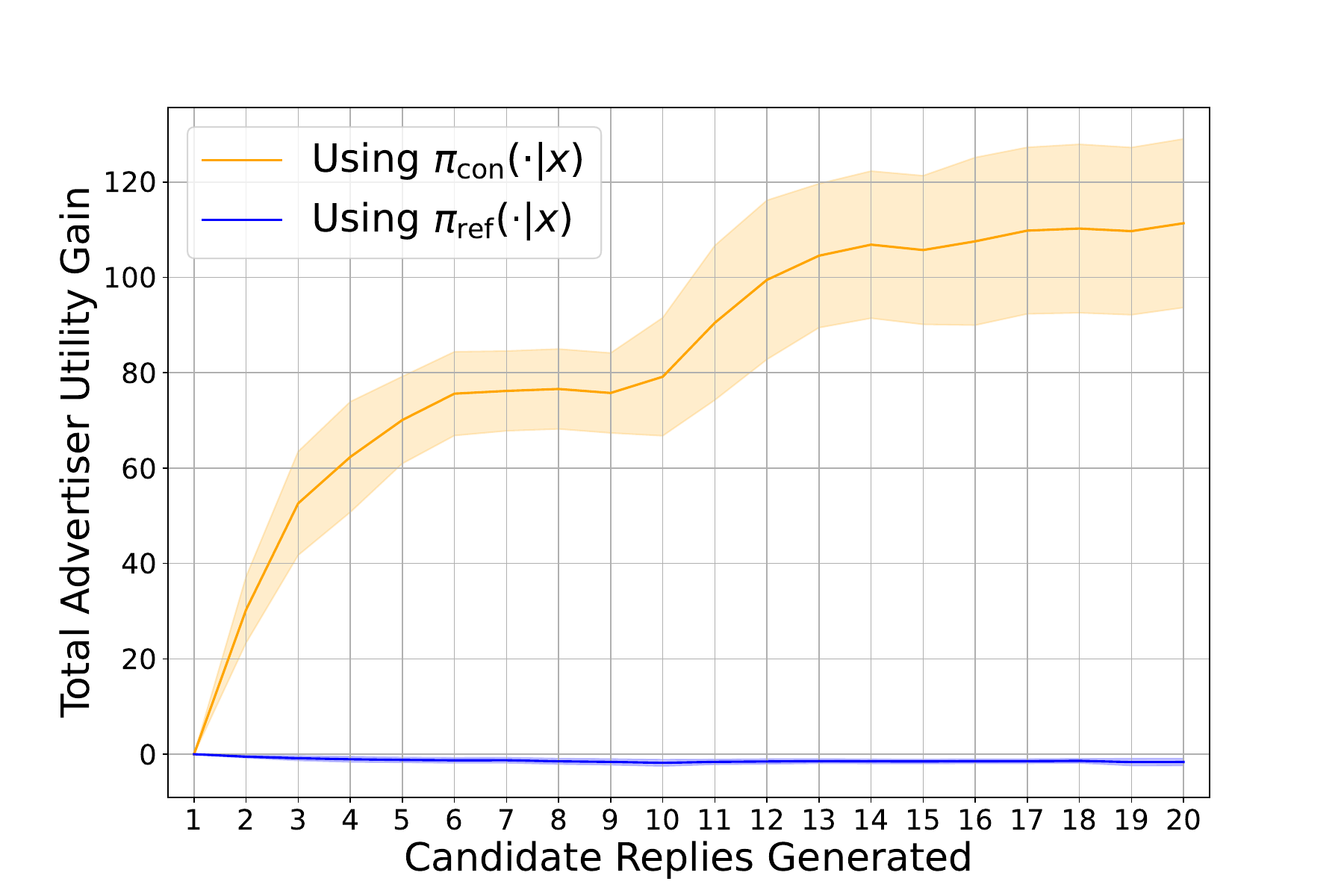}
    \caption{Advertiser value gain as a function of the number of candidate replies generated using $\pi_{\text{ref}}$ and $\pi_{\text{con}}$ (flan-t5-large model).
    }
    \label{fig:advertiser_participating-utility_gain_llm_t5}
\end{figure}

\begin{figure}[h!]
    \centering
    \includegraphics[width=0.65\linewidth]{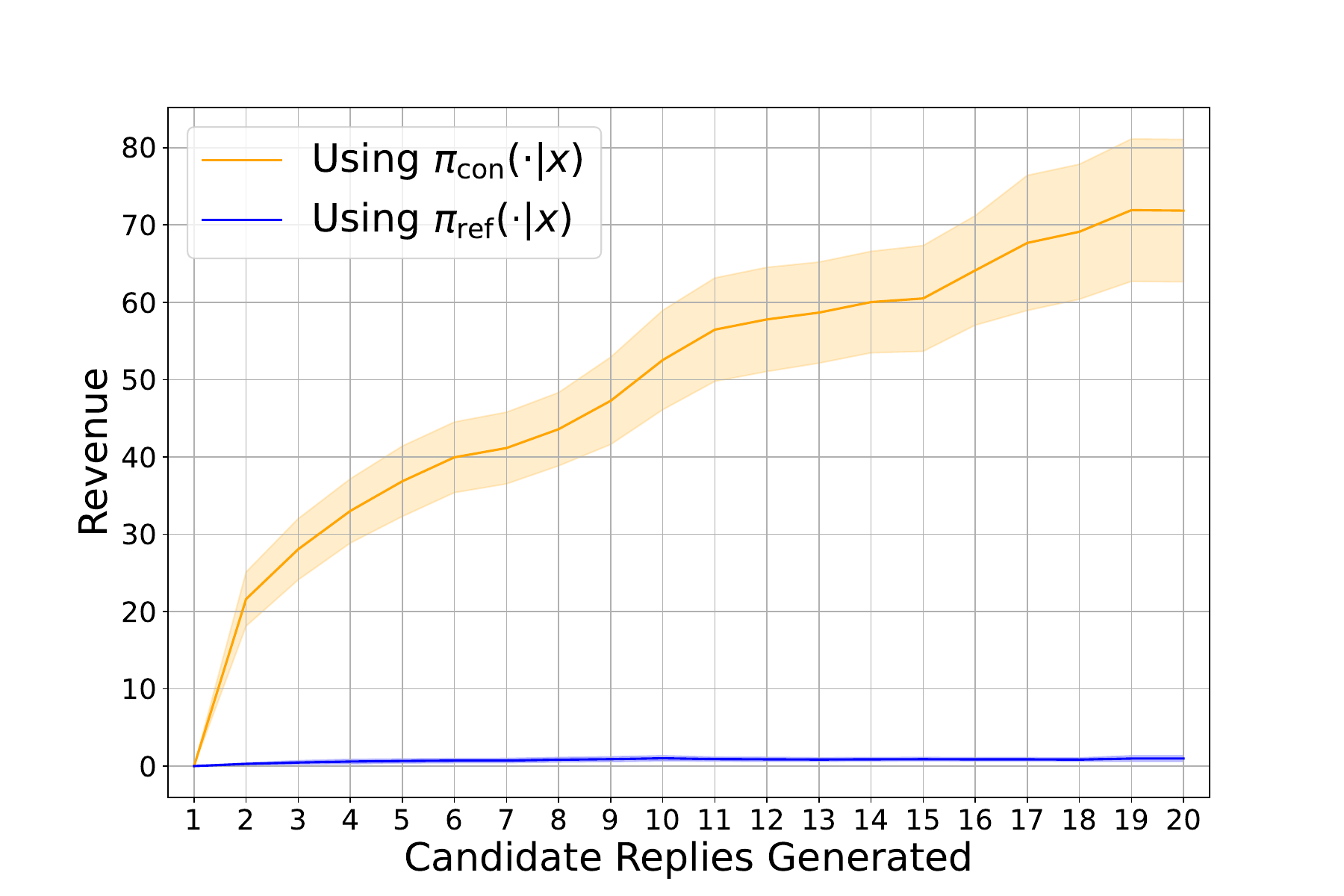}
    \caption{Revenue as a function of the number of candidate replies generated using $\pi_{\text{ref}}$ and $\pi_{\text{con}}$ (flan-t5-large model).
    }
    \label{fig:revenue_llm_t5}
\end{figure}

\begin{figure}[h!]
    \centering
    \includegraphics[width=0.65\linewidth]{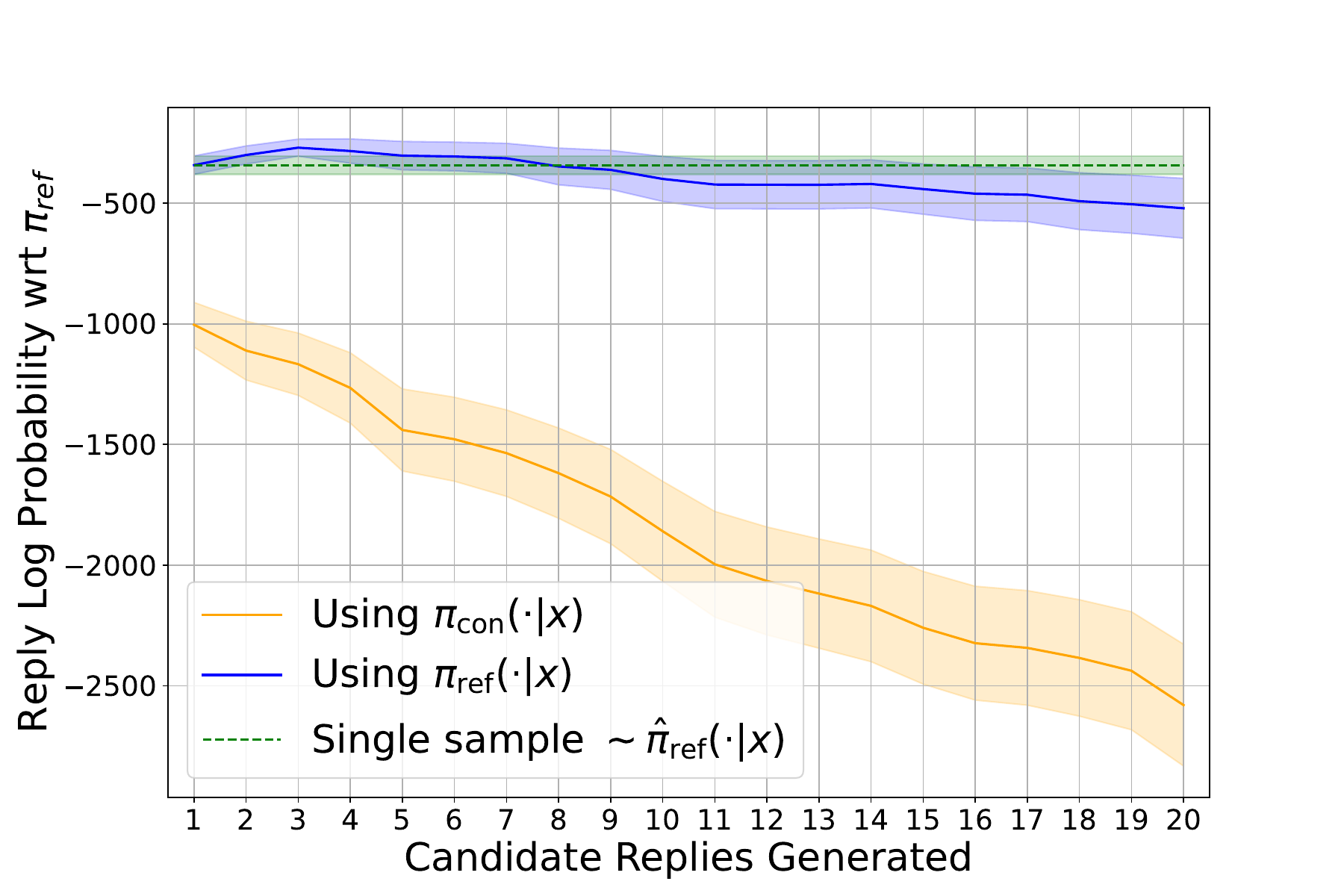}
    \caption{Reply log probability with respect to the reference LLM as a function of the number of candidate replies generated using $\pi_{\text{ref}}$ and $\pi_{\text{con}}$ (flan-t5-large model).
    }
    \label{fig:log_probability_reference_llm_t5}
\end{figure}

\newpage
\clearpage

\newpage
\clearpage

\section{Full text of prompts} \label{sec:all_prompts}

\begin{longtable}{ll}
\caption{Table of user prompts with corresponding advertisers and their descriptions. Only the first two advertisers are shown for each query; the full prompts are available in the code supplement.}
\label{tab:user_prompts_advertisers}\\
&\textbf{User Query 1} \\
&Learning to play the guitar. \\
&\textbf{Advertisers} \\
&- StrumMaster: selling guitars, tuners, and learning materials for beginners \\
&- Chordify: offering online guitar lessons and interactive practice tools \\
\midrule
&\textbf{User Query 2} \\
&Best skin care routine for dry skin. \\
&\textbf{Advertisers} \\
&- HydraSkin: providing moisturizers and serums for dry skin \\
&- GlowEssentials: selling natural and organic skin care products \\
\midrule
&\textbf{User Query 3} \\
&Improving productivity while working from home. \\
&\textbf{Advertisers} \\
&- SpaceMakers: offering ergonomic furniture and home office setups \\
&- FocusApp: providing productivity apps and time management tools \\
\midrule
&\textbf{User Query 4} \\
&Starting a small online business. \\
&\textbf{Advertisers} \\
&- EcomLaunch: offering e-commerce platform solutions and web design services \\
&- MarketMover: providing digital marketing services and SEO optimization \\
\midrule
&\textbf{User Query 5} \\
&Healthy meal planning on a budget. \\
&\textbf{Advertisers} \\
&- BudgetBites: selling affordable meal kits and recipe books \\
&- NutriSaver: offering discounts on healthy groceries and food delivery services \\
\midrule
&\textbf{User Query 6} \\
&Mastering digital photography. \\
&\textbf{Advertisers} \\
&- PixelPro: selling cameras, lenses, and photography accessories \\
&- EditCraft: offering photo editing software and online tutorials \\
\midrule
&\textbf{User Query 7} \\
&Effective ways to reduce household energy use. \\
&\textbf{Advertisers} \\
&- EcoSave: offering energy-efficient home appliances and lighting solutions \\
&- InsulaTech: providing home insulation and energy audit services \\
\midrule
&\textbf{User Query 8} \\
&Finding the perfect hiking trails. \\
&\textbf{Advertisers} \\
&- TrailFinder: offering a mobile app with detailed maps and trail reviews \\
&- GearUp: selling outdoor gear and apparel for hiking enthusiasts \\
\midrule
&\textbf{User Query 9} \\
&Building a personal brand on social media. \\
&\textbf{Advertisers} \\
&- BrandBuilder: offering personal branding courses and social media strategy consultations \\
&- VisualizeMe: providing graphic design services for social media content \\
\midrule
&\textbf{User Query 10} \\
&Learning a new language effectively. \\
&\textbf{Advertisers} \\
&- LingoLeap: offering online language learning courses and tutoring \\
&- SpeakEasy: providing language learning apps with speech recognition technology \\
\midrule
&\textbf{User Query 11} \\
&Staying fit without a gym. \\
&\textbf{Advertisers} \\
&- HomeFit: selling home workout equipment and fitness accessories \\
&- MoveIt: offering online fitness classes and personal training sessions \\
\midrule
&\textbf{User Query 12} \\
&Eco-friendly travel options. \\
&\textbf{Advertisers} \\
&- GreenPath: offering eco-friendly travel packages and sustainable tourism experiences \\
&- EcoStay: providing listings for green hotels and accommodations \\
\midrule
&\textbf{User Query 13} \\
&Mastering the art of cooking steak. \\
&\textbf{Advertisers} \\
&- GrillMaster: selling premium grills and barbecue accessories \\
&- SteakPerfection: offering online cooking classes focused on meat preparation \\
\midrule
&\textbf{User Query 14} \\
&Creating a successful YouTube channel. \\
&\textbf{Advertisers} \\
&- VidGrowth: offering video production courses and YouTube growth strategies \\
&- ChannelDesign: providing custom YouTube channel art and video thumbnails \\
\midrule
&\textbf{User Query 15} \\
&Decorating your home on a budget. \\
&\textbf{Advertisers} \\
&- DecorDeals: selling affordable home decor and furniture \\
&- StyleSavvy: offering interior design consultations and budget-friendly decorating tips \\
\midrule
&\textbf{User Query 16} \\
&Managing stress and anxiety. \\
&\textbf{Advertisers} \\
&- CalmSpace: offering mindfulness apps and stress reduction tools \\
&- WellnessWave: providing online therapy sessions and mental health resources \\
\midrule
&\textbf{User Query 17} \\
&How can I make my home more energy efficient? \\
&\textbf{Advertisers} \\
&- EcoWise: offering smart home energy solutions and solar panels \\
&- GreenInsulate: providing eco-friendly insulation materials and installation services \\
\midrule
&\textbf{User Query 18} \\
&What are the best strategies for investing in the stock market? \\
&\textbf{Advertisers} \\
&- InvestSmart: offering investment advice and financial planning services \\
&- StockGenius: providing an AI-powered stock market analysis tool \\
\midrule
&\textbf{User Query 19} \\
&How can I improve my mental health? \\
&\textbf{Advertisers} \\
&- MindBalance: offering mental health apps and mindfulness courses \\
&- SerenityNow: providing online counseling and therapy sessions \\
\midrule
&\textbf{User Query 20} \\
&Where can I find organic and sustainable clothing? \\
&\textbf{Advertisers} \\
&- EcoWear: selling sustainable and organic clothing options \\
&- GreenThreads: offering eco-friendly fabrics and fair-trade fashion \\
\midrule
&\textbf{User Query 21} \\
&How do I start a vegetable garden in my backyard? \\
&\textbf{Advertisers} \\
&- GardenGuru: selling gardening tools, seeds, and starter kits \\
&- PlantPerfect: offering online gardening courses and personalized garden planning \\
\midrule
&\textbf{User Query 22} \\
&What are the best tools for remote team collaboration? \\
&\textbf{Advertisers} \\
&- TeamSync: offering collaboration software and project management tools \\
&- ConnectSuite: providing video conferencing solutions and team communication platforms \\
\midrule
&\textbf{User Query 23} \\
&How can I learn to cook healthy meals quickly? \\
&\textbf{Advertisers} \\
&- QuickChef: offering quick and healthy cooking classes online \\
&- NutriFast: selling kitchen gadgets and appliances for healthy cooking \\
\midrule
&\textbf{User Query 24} \\
&What are the best ways to save for retirement? \\
&\textbf{Advertisers} \\
&- FutureFund: offering retirement planning and investment services \\
&- PensionPlus: providing personalized financial advice for retirement savings \\
\midrule
&\textbf{User Query 25} \\
&How can I decorate my apartment on a minimalist budget? \\
&\textbf{Advertisers} \\
&- MinimalSpace: selling minimalist furniture and decor \\
&- SimplicityDesigns: offering interior design consultation for minimalist living spaces \\
\midrule
&\textbf{User Query 26} \\
&What are the best practices for data security in small businesses? \\
&\textbf{Advertisers} \\
&- SecureIT: offering cybersecurity solutions for small businesses \\
&- DataGuard: providing data protection services and training \\
\midrule
&\textbf{User Query 27} \\
&How can I reduce my carbon footprint with daily habits? \\
&\textbf{Advertisers} \\
&- EcoHabits: offering products and tips for sustainable living \\
&- GreenStep: providing carbon offset subscriptions and eco-friendly lifestyle goods \\
\midrule
&\textbf{User Query 28} \\
&What are effective ways to learn a musical instrument online? \\
&\textbf{Advertisers} \\
&- MusicMastery: offering online music lessons and virtual tutoring \\
&- InstaTune: selling musical instruments and learning aids \\
\midrule
&\textbf{User Query 29} \\
&How can I keep my pets healthy and happy? \\
&\textbf{Advertisers} \\
&- PetWell: offering pet health insurance and wellness products \\
&- HappyPaws: selling organic pet food and eco-friendly toys \\
\midrule
&\textbf{User Query 30} \\
&What's the best way to start learning coding from scratch? \\
&\textbf{Advertisers} \\
&- CodeBegin: offering coding bootcamps and online tutorials for beginners \\
&- DevPath: providing interactive coding courses and learning platforms \\
\midrule
&\textbf{User Query 31} \\
&How to find affordable eco-friendly travel destinations? \\
&\textbf{Advertisers} \\
&- GreenTraveler: offering eco-friendly travel guides and affordable green vacation packages \\
&- EcoTripAdvisor: providing reviews and bookings for sustainable travel options \\
\midrule
&\textbf{User Query 32} \\
&What are the best ways to improve indoor air quality? \\
&\textbf{Advertisers} \\
&- AirPure: selling air purifiers and natural air freshening solutions \\
&- BreatheEasy: offering indoor plant subscriptions and ventilation solutions \\
\midrule
&\textbf{User Query 33} \\
&How to make my own natural beauty products? \\
&\textbf{Advertisers} \\
&- NatureGlow: offering DIY beauty product kits and organic ingredients \\
&- EcoBeauty: selling guides and supplies for creating homemade natural cosmetics \\
\midrule
&\textbf{User Query 34} \\
&What are the best techniques for effective online teaching? \\
&\textbf{Advertisers} \\
&- EduTech: offering online teaching tools and educational technology solutions \\
&- TeachMaster: providing courses on online pedagogy and digital classroom management \\
\midrule
&\textbf{User Query 35} \\
&How to create a successful podcast from home? \\
&\textbf{Advertisers} \\
&- PodCastify: offering podcast creation courses and audio equipment \\
&- SoundWave: selling podcast editing software and hosting services \\
\midrule
&\textbf{User Query 36} \\
&What are the most effective ways to market a small business locally? \\
&\textbf{Advertisers} \\
&- LocalBoost: offering local SEO services and marketing strategies for small businesses \\
&- MarketTown: providing targeted advertising solutions and local business networking opportunities \\
\midrule
&\textbf{User Query 37} \\
&How can I optimize my home office for productivity? \\
&\textbf{Advertisers} \\
&- WorkZen: offering ergonomic office furniture and productivity tools \\
&- SpaceCraft: providing home office design services and organization solutions \\
\midrule
&\textbf{User Query 38} \\
&What are the best platforms for freelance writers to find work? \\
&\textbf{Advertisers} \\
&- WriteFreely: offering a job board and resources for freelance writers \\
&- ContentConnect: providing a platform to connect writers with businesses \\
\midrule
&\textbf{User Query 39} \\
&How to get started with meditation and mindfulness? \\
&\textbf{Advertisers} \\
&- MindfulStart: offering guided meditation apps and mindfulness training \\
&- ZenSpace: selling meditation cushions and mindfulness journals \\
\midrule
&\textbf{User Query 40} \\
&Where can I find unique handmade gifts? \\
&\textbf{Advertisers} \\
&- CraftCorner: selling handmade gifts and artisan crafts \\
&- MakerMarket: offering a marketplace for unique handmade items from local artists \\
\midrule
&\textbf{User Query 41} \\
&How can I improve my credit score quickly? \\
&\textbf{Advertisers} \\
&- CreditBoost: offering credit monitoring services and financial advice \\
&- ScoreSurge: providing tools and tips for improving your credit score \\
\midrule
&\textbf{User Query 42} \\
&What are the best apps for tracking fitness and nutrition? \\
&\textbf{Advertisers} \\
&- FitTrack: offering comprehensive fitness and nutrition tracking apps \\
&- NutriSync: selling personalized nutrition plans and diet tracking tools \\
\midrule
&\textbf{User Query 43} \\
&How to plan an environmentally friendly wedding? \\
&\textbf{Advertisers} \\
&- GreenBride: offering eco-friendly wedding planning services and supplies \\
&- EcoWed: providing sustainable wedding attire and decor \\
\midrule
&\textbf{User Query 44} \\
&What are the safest ways to travel during a pandemic? \\
&\textbf{Advertisers} \\
&- SafeJourney: offering travel safety kits and pandemic travel advice \\
&- HealthTravel: providing information on safe destinations and travel insurance \\
\midrule
&\textbf{User Query 45} \\
&How can I learn DIY home repairs and improvements? \\
&\textbf{Advertisers} \\
&- FixItYourself: offering online courses and tutorials for home repair \\
&- DIYHome: selling DIY home improvement tools and kits \\
\midrule
&\textbf{User Query 46} \\
&What are the best resources for starting a plant-based diet? \\
&\textbf{Advertisers} \\
&- PlantEats: offering guides and meal plans for starting a plant-based diet \\
&- VeggieVibe: selling plant-based cookbooks and kitchen gadgets \\
\midrule
&\textbf{User Query 47} \\
&Tips for running a marathon. \\
&\textbf{Advertisers} \\
&- RunFastGear: specializing in high-performance running shoes \\
&- HydraFuel: offering electrolyte drinks and energy bars for athletes \\
\midrule
&\textbf{User Query 48} \\
&Best practices for organic gardening. \\
&\textbf{Advertisers} \\
&- GreenThumb Solutions: providing organic fertilizers and pest control products \\
&- EcoGrow: selling heirloom seeds and sustainable gardening tools \\
\midrule
&\textbf{User Query 49} \\
&How to improve home WiFi? \\
&\textbf{Advertisers} \\
&- SignalBoost: which offers advanced WiFi routers and extenders \\
&- NetWizard: providing network optimization services and support \\
\midrule
&\textbf{User Query 50} \\
&Ways to save on travel. \\
&\textbf{Advertisers} \\
&- BudgetJourneys: specializing in affordable travel packages and deals \\
&- StayLocal: offering discounts on boutique hotels and unique accommodations \\
\bottomrule
\end{longtable}

\newpage
\clearpage

\newpage

\end{document}